\documentclass[acmsmall, screen]{acmart}

\usepackage{amsmath}
\usepackage{amsfonts}
\usepackage{xcolor}
\usepackage{geometry} 
\usepackage[ruled,commentsnumbered,linesnumbered,noend]{algorithm2e}

\usepackage{amsthm}
\usepackage{listings}
\usepackage{graphicx}
\usepackage[normalem]{ulem} 
\usepackage{xcolor} 
\usepackage{lscape}
\usepackage{comment}
\usepackage[inline]{enumitem}
\usepackage{xspace}

\usepackage{float}
\usepackage{array} 
\usepackage{tikz}
\usetikzlibrary{arrows,graphs}
\usetikzlibrary{backgrounds,calc}
\usepackage{subcaption}
\usepackage{xparse} 
\usepackage{url}

\usepackage{thmtools} 
\usepackage{thm-restate}
\usepackage{lipsum}
\usepackage{multirow}

\usepackage{parskip}

\usepackage{pgfplots}
\usepackage{pgfplotstable}
\usepackage{longtable}

\newcommand{\Nats}{\mathbb{N}}

\newtheorem{problem}{Problem}
\newtheorem{example}{Example}
\newtheorem{remark}{Remark}

\newcommand{\figlabel}[1]{\label{fig:#1}}
\newcommand{\figref}[1]{Figure~\ref{fig:#1}}

\newcommand{\seclabel}[1]{\label{sec:#1}}
\newcommand{\secref}[1]{Section~\ref{sec:#1}}

\newcommand{\problabel}[1]{\label{prob:#1}}

\newcommand{\thmlabel}[1]{\label{thm:#1}}
\newcommand{\thmref}[1]{Theorem~\ref{thm:#1}}

\newcommand{\lemlabel}[1]{\label{lem:#1}}
\newcommand{\lemref}[1]{Lemma~\ref{lem:#1}}
\newcommand{\applabel}[1]{\label{app:#1}}
\newcommand{\appref}[1]{Appendix~\ref{app:#1}}
\newcommand{\tablabel}[1]{\label{tab:#1}}
\newcommand{\tabref}[1]{Table~\ref{tab:#1}}

\newenvironment{myparagraph}[1]{\noindent{\bf #1.}}{}
\newenvironment{mysubparagraph}[1]{\noindent{\em #1.}}{}

\newcommand{\set}[1]{\{#1\}}
\newcommand{\setpred}[2]{\{#1 \,|\, #2\}}
\newcommand{\tuple}[1]{\langle #1 \rangle}

\renewcommand{\emptyset}{\varnothing}
\newcommand{\s}[1]{\mathsf{#1}}
\newcommand{\proj}[2]{{#1}{\downharpoonright}_{#2}}

\colorlet{colorSND}{blue!80!black}
\colorlet{colorRCV}{green!40!black}
\colorlet{colorCREATE}{yellow!50!black}
\colorlet{colorSPAWN}{magenta}
\colorlet{colorCLOSE}{red!80!black}
\colorlet{colorRF}{purple!80!black}
\colorlet{colorClause}{red!30!gray}
\colorlet{colorVar}{gray!80}
\colorlet{colorSTAR}{green!40!black}

\newcommand{\ch}{\mathtt{ch}}
\newcommand{\val}{\s{val}}
\newcommand{\snd}{\mathtt{\color{colorSND}snd}}
\newcommand{\rcv}{\mathtt{\color{colorRCV}rcv}}
\newcommand{\createCh}{\mathtt{\color{colorCREATE}create}}
\newcommand{\closeCh}{\mathtt{\color{colorCLOSE}close}}
\newcommand{\spawnThr}{\mathtt{\color{colorSPAWN}spawn}}
\newcommand{\rd}{\rcv}
\newcommand{\wt}{\snd}
\newcommand{\pop}{\rcv}
\newcommand{\push}{\snd}
\newcommand{\wtJustVal}{\snd}
\newcommand{\rdJustVal}{\rcv}
\newcommand{\wtMem}{\s{w}}
\newcommand{\rdMem}{\s{r}}

\newcommand{\trace}{\sigma}
\newcommand{\thread}{\tau}
\newcommand{\tr}{\trace}
\newcommand{\prefix}{\pi}
\newcommand{\ev}[1]{\tuple{{#1}}}
\newcommand{\op}{\s{op}}
\newcommand{\reordering}{\rho}
\newcommand{\ThreadOf}[1]{\s{th}(#1)}
\newcommand{\OpOf}[1]{\s{op}(#1)}
\newcommand{\ChOf}[1]{\s{ch}(#1)}
\newcommand{\ValueOf}[1]{\s{val}(#1)}
\newcommand{\cp}[1]{\cpFunc\s{(#1)}}
\newcommand{\cpFunc}{\s{cap}}
\newcommand{\eventSet}{\mathsf{S}}
\newcommand{\AbstractExecution}{\mathcal{X}}

\newcommand{\events}[1]{\s{Events(#1)}}
\newcommand{\threads}[1]{\s{Threads(#1)}}
\newcommand{\channels}[1]{\s{Channels(#1)}}
\newcommand{\lk}{\ell}

\newcommand{\Channels}{\mathcal{C}}
\newcommand{\fgEventSet}{Y}
\newcommand{\fgChanMap}{Q}
\newcommand{\fgSyncSend}{I}

\newcommand{\ord}[2]{\leq^{#1}_{\mathsf{#2}}}
\newcommand{\strictOrd}[2]{<^{#1}_{\mathsf{#2}}}

\newcommand{\strictord}[2]{<^{#1}_{\mathsf{#2}}}

\newcommand{\acrtr}{\textsf{tr}\xspace}
\newcommand{\acrpo}{\textsf{po}\xspace}
\newcommand{\acrrf}{\textsf{rf}\xspace}

\newcommand{\trord}[1]{\strictOrd{#1}{\acrtr}}
\newcommand{\poord}[1]{\ord{#1}{\acrpo}}
\newcommand{\stricttrord}[1]{\strictord{#1}{\acrtr}}
\newcommand{\po}[1]{{\acrpo}_{#1}}
\newcommand{\rf}[1]{{\acrrf}_{#1}}

\newcommand{\satord}{\strictord{}{sat}}

\newcommand{\NP}{\textsf{NP}}
\newcommand{\NumEvents}{n}
\newcommand{\NumThreads}{t}
\newcommand{\NumChannels}{m}
\newcommand{\MaxCapacity}{k}

\newcommand{\vch}{\textsf{VCh}\xspace}
\newcommand{\vchRf}{\textsf{\vch{-}rf}\xspace}

\newcommand{\vscRd}{\textsf{VSC{-}read}\xspace}

\newcommand{\gfrontier}{G_{\s{frontier}}}
\newcommand{\gfrontierRf}{G_{\s{frontier}}^{\rf{}}}
\newcommand{\gsync}{G_{\s{sync}}}

\newcommand{\form}{\varphi}
\newcommand{\trans}{\textsf{trans}}
\newcommand{\exactlyOne}{\textsf{exactly-1}}
\newcommand{\fifo}{\textsf{FIFO}}
\newcommand{\unmatched}{\textsf{unmatched}}
\newcommand{\capOne}{\textsf{cap=1}}
\newcommand{\sync}{\textsf{sync}}
\newcommand{\pred}[2]{\s{pred}_{#1}(#2)}
\newcommand{\sucr}[2]{\s{succ}_{#1}(#2)}

\SetKwProg{myproc}{procedure}{}{}
\SetKwFunction{treeTopo}{AcyclicTopologyVChRF}
\SetKwFunction{twoThread}{TwoThreadVChRF}
\SetKwFor{For}{for}{do}{}
\SetKw{Let}{let}

\newcommand{\numVar}{{n_v}}
\newcommand{\numClause}{{n_c}}

\newcommand{\chanCnt}{{\sf cnt}}
\newcommand{\numEdge}{{n_e}}
\newcommand{\numNode}{{n_v}}
\newcommand{\outDegree}[1]{\s{out}(#1)}
\newcommand{\inDegree}[1]{\s{in}(#1)}

\newcommand{\atom}{\textcolor{cyan!80!black}{\sf atomic}}

\tikzstyle{graphNode}=[circle,  fill=gray!10, draw=black, very thick]
\tikzstyle{thread}=[->, line width=1mm]
\tikzstyle{threadName}=[text width=2cm, align=center, font=\normalfont]
\tikzstyle{event}=[draw=black, align=center, fill=white, line width=0.8pt, minimum width=\wid cm,minimum height=\height cm]
\tikzstyle{nodeRectangle}=[draw=black, align=center, fill=white, line width=0.8pt, minimum width=\wid cm,minimum height=\height cm, rounded corners]
\tikzstyle{edge}=[->, line width=0.4mm, draw=black]
\tikzstyle{rfEdge}=[->, line width=0.4mm, draw=colorRF]
\tikzstyle{executeEdge}=[->, line width=0.4mm, draw=black]

\newcommand{\pointsTo}{=}

\newcommand{\fgAlgo}{\mathsf{FG}}
\newcommand{\fgAlgoSat}{\mathsf{FG\text{-}Sat}}

\newcommand{\tsan}{\textsc{ThreadSanitizer}\xspace}
\newcommand{\smt}{\mathsf{SMT}}
\newcommand{\smtSat}{\mathsf{SMT\text{-}Sat}}

\newcommand{\satPO}{\satord}

\newcommand{\revision}[1]{#1}

\newlist{compactitem}{itemize}{3} 
\setlist[compactitem]{label=\textbullet,nosep,leftmargin=*}

\newlist{compactenum}{enumerate}{3}
\setlist[compactenum,1]{label=\arabic*., nosep, leftmargin=*}
\setlist[compactenum,2]{label=\alph*), nosep, leftmargin=*}
\setlist[compactenum,3]{label=\roman*), nosep, leftmargin=*}

\newlist{compactwideenum}{enumerate}{3} 
\setlist[compactwideenum]{label=\textit{\arabic*}., nosep,leftmargin=*,wide}

\newcommand{\execution}[2]{
\scalebox{0.85}{
  \begin{tikzpicture}%
    \foreach \x in {1,...,#1}
    \node at (1.9*\x+0.3,0.25) {$\thread_{\x}$};
    \draw (1.2,0) -- (#1*1.9+1.1,0);%
    \pgfmathsetmacro{\y}{1};%
    #2%
    \draw (1.2,0) -- (1.2,-0.5*\y);%
    \draw (#1*1.9+1.1,0) -- (#1*1.9+1.1,-0.5*\y);%
    \foreach \x in {2,...,#1}
    \draw (1.2,-0.5*\y) -- (#1*1.9+1.1,-0.5*\y);%
  \end{tikzpicture}
}
}

\newcommand{\figev}[2]{
\pgfmathsetmacro{\y}{\y+1};
\pgfmathsetmacro{\y}{\y-1};
\node [left] at (1.2, -0.5*\y)  
{\pgfmathprintnumber{\y}};%
\node at (#1*1.9 + 0.2, -0.5*\y) {$ #2 $};%
\pgfmathsetmacro{\y}{\y+1};
}

\begin{document}






\title{The Complexity of Testing Message-Passing Concurrency}

\author{Zheng Shi}
\orcid{0000-0001-5021-7134}
\affiliation{%
  \institution{National University of Singapore}
  \city{Singapore}
  \country{Singapore}
}
\email{shizheng@u.nus.edu}

\author{Lasse M\o{}ldrup}
\orcid{0009-0005-9670-7039}
\affiliation{%
  \institution{Aarhus University}
  \city{Aarhus}
  \country{Denmark}
}
\email{moeldrup@cs.au.dk}

\author{Umang Mathur}
\orcid{0000-0002-7610-0660}
\affiliation{%
  \institution{National University of Singapore}
  \city{Singapore}
  \country{Singapore}
}
\email{umathur@comp.nus.edu.sg}

\author{Andreas Pavlogiannis}
\orcid{0000-0002-8943-0722}
\affiliation{%
  \institution{Aarhus University}
  \city{Aarhus}
  \country{Denmark}
}
\email{pavlogiannis@cs.au.dk}


\begin{abstract}
A key computational question underpinning the automated testing and verification of concurrent programs is the \emph{consistency question} --- \emph{given a partial execution history, can it be completed in a consistent manner?}
Due to its importance, consistency testing has been studied extensively for memory models, as well as for database isolation levels. 
A common theme in all these settings is the use of shared-memory as the primal mode of interthread communication.
On the other hand, modern programming languages, such as Go, Rust and Kotlin, advocate a paradigm shift towards channel-based (i.e., message-passing) communication.
However, the consistency question for channel-based concurrency is currently poorly understood.

In this paper we lift the study of fundamental consistency problems to channels, taking into account various input parameters, such as the number of threads executing, the number of channels, and the channel capacities.
We draw a rich complexity landscape, including upper bounds that become polynomial when certain input parameters are fixed,
as well as hardness lower bounds.
Our upper bounds are based on 
algorithms that can drive the verification of channel consistency in automated verification tools.
Our lower bounds characterize minimal input parameters that are sufficient for hardness to arise, and thus shed light on the intricacies of testing channel-based concurrency.
In combination, our upper and lower bounds characterize the boundary of \emph{tractability/intractability} of verifying channel consistency, and imply that our algorithms are often (nearly) optimal. 
We have also implemented our main consistency checking algorithm
and designed optimizations to enhance its performance.
We evaluated the performance of our implementation
over a set of 103 instances obtained from open source Go projects, and compared it
against a constraint-solving based algorithm. 
Our experimental results demonstrate the power of our consistency-checking algorithm; it scales to around 1M events, and is
significantly faster in running-time performance,
compared to a constraint-solving approach.
\end{abstract}


\setcopyright{cc}
\setcctype{by}
\acmDOI{10.1145/3776643}
\acmYear{2026}
\acmJournal{PACMPL}
\acmVolume{10}
\acmNumber{POPL}
\acmArticle{1}
\acmMonth{1}
\received{2025-06-30}
\received[accepted]{2025-11-06}


\maketitle

\section{Introduction}
\seclabel{intro}


The verification and testing of concurrent programs has been a 
major challenge in programming languages and formal methods.
Inter-thread communication leads to a combinatorial blow-up in the set of program behaviors, 
which makes program development error prone and program analysis computationally challenging.
Nevertheless, a multitude of techniques have been developed for analyzing concurrent programs automatically, such as bounded model checking~\cite{wu2023model, dal2024model}, partial order reduction~\cite{storey2021sound, abdulla2014optimal, Kokologiannakis2022tRULY}, predictive runtime testing~\cite{Said11,wcp2017,mathur2021optimal},
fuzz testing~\cite{wolff2024greybox, wen2022controlled, musuvathi2008finding}, 
and static analysis~\cite{liu2021threads, muller2024language}.
The vast majority of these techniques operate under the assumption that 
interthread communication takes place over \emph{shared memory}.

One key problem that has driven the development 
of algorithmic techniques for concurrency testing and verification is \emph{consistency testing}.
At a high level, the input to the problem is a thread-level observable execution of the program (e.g., a sequence of events executed by each thread), without memory-level information about how threads interacted (e.g., a precise thread interleaving, or the order in which writes appeared in the shared memory).
The output to the problem is YES iff the thread-level behavior is aligned 
with the specifics of the underlying architecture (e.g., the memory model).
The complexity of consistency testing has been a subject of systematic study for both Sequential Consistency (SC)~\cite{gibbons1994testing,gibbons1997testing,cantin2005complexity,Mathur2020,Tunc2024} and weak memory models~\cite{Furbach2015,Lahav2015,chakraborty2024hard,tuncc2023optimal},
as  well as for database isolation levels \cite{biswas2019complexity,biswas2019complexity2, 
bouajjani2023dynamic,Moldrup2025}.
These results have propelled the development of techniques for model checking programs under 
SC~\cite{Chalupa2018,Abdulla2019b,Chatterjee2019,Agarwal2021,Kokologiannakis2022} and 
weak memory~\cite{Abdulla2018,Bui2021,PoncedeLeon2021,Kokologiannakis2021}, 
as well as for effective testing~\cite{Huang2014,Kalhauge2018,Pavlogiannis2020,Luo2021,mathur2021optimal,Biswas2021}.
In the context of model checking for concurrent programs, 
efficient consistency testing is key for balancing
optimality (i.e., non-redundancy) and performance of the partial-order-reduction based
exploration~\cite{Kokologiannakis2022,abdulla2014optimal,Abdulla2018}.
On the other hand, 
in the context of runtime predictive analysis, consistency checking 
plays the dual role of efficiently exploring an 
entire class of executions that can be inferred from a given execution, without
explicitly enumerating members of the class~\cite{Mathur2020,Said11,Huang2014}.

In order to make concurrent programming more seamless and reliable, 
modern programming languages advocate for interthread communication
mechanisms that are structured and offer clean abstractions.
One such case is the use of \emph{message-passing}, popularized by the use of \emph{channels} in Go~\cite{Go-language}, and also used frequently in other mainstream languages, such as
Rust~\cite{Rust-language},
Scala~\cite{Scala-language}, 
Erlang~\cite{Erlang-language} 
and Kotlin~\cite{Kotlin-language}.

Despite the structured communication offered by explicit channel-based 
concurrency in such languages, 
subtle concurrency bugs in large-scale software written in 
this paradigm remain widespread~\cite{GoBugsStudy2019,GoBugsUber2022,GFuzz2022}.
In turn, this requires the development of verification and testing methods 
that are capable of reasoning about channels effectively, in order to capture the 
behaviors entailed by the programs they target~\cite{Sulzmann2017,Sulzmann2018,Chabbi2022}.
However, the core problem of consistency testing 
has thus far been elusive for channel-based communication:~\emph{How fast can we verify the consistency of message-passing executions?}
We address this question in this work, by drawing a rich landscape  of the complexity of the problem depending on various input parameters.
Besides the technical merit of our results, they also provide a precise characterization of the ingredients that make the consistency problem  hard.
Likewise, the algorithms we propose can be employed in techniques where soundness 
and completeness are paramount and a strict upper bound on computational resources is desired.



\subsection{Motivating Example}\label{SUBSEC:MOTIVATING_EXAMPLE}

We illustrate the need for consistency checking on channels by means of a simple example where this problem arises naturally.
The Go programming language primarily uses the message-passing concurrency paradigm,
and offers \emph{channels} as a first class abstraction for interthread communication.
A channel in Go is a $\fifo$ queue, possibly with some capacity~\cite{Go-channel-impl},
which a thread can create, close, send to and receive from~\cite{Go-channel}. 

\begin{figure}[htbp]
\definecolor{LightGray}{gray}{0.95}
\centering
\begin{subfigure}[b]{0.38\textwidth}
\centering
\definecolor{MyKeywordColor}{RGB}{0,128,63}
\lstdefinelanguage{GoCustom}{
    language=Go,
    morekeywords={go, func, chan, make, close}, 
}

\lstset{
    language=GoCustom,
    basicstyle=\scriptsize\ttfamily,
    numbers=left,
    numbersep=2pt,
    xleftmargin=1.0em,
    frame=lines,
    framesep=2mm,
    keywordstyle=\color{MyKeywordColor}\bfseries,      
    commentstyle=\color{green!50!black}\itshape, 
    stringstyle=\color{orange},             
    showstringspaces=false,
    breaklines=true,
    lineskip=1.5pt
}
\begin{lstlisting}
func main() {
    asyncCh := make(chan int, 2)
    go func(asyncCh chan int) {
        asyncCh <- 1
    }(asyncCh)
    asyncCh <- 1
    <-asyncCh
    close(asyncCh)
}
\end{lstlisting}
\caption{A buggy Go code snippet}
\figlabel{go-code-demo}
\end{subfigure}
\hfill
\begin{subfigure}[b]{0.30\textwidth}
\centering
\execution{2}{
        \figev{1}{\createCh(\ch)}
	\figev{1}{\spawnThr(\thread_2)}
        \figev{2}{\wt(\ch, 1)}
        \figev{1}{\wt(\ch, 1)}
        \figev{1}{\rd(\ch, 1)}
	\figev{1}{\closeCh(\ch)}
}
\vspace{0.3cm}
\caption{A non-buggy execution $\trace$.}
\figlabel{go-code-demo-valid-exec}
\end{subfigure}
\begin{subfigure}[b]{0.30\textwidth}
\centering
\execution{2}{
        \figev{1}{\createCh(\ch)}
	\figev{1}{\spawnThr(\thread_2)}
        \figev{1}{\wt(\ch, 1)}
        \figev{1}{\rd(\ch, 1)}
        \figev{1}{\closeCh(\ch)}
        \figev{2}{\wt(\ch, 1)}
}
\vspace{0.3cm}
\caption{A buggy execution $\trace'$.}
\figlabel{go-code-demo-buggy-exec}
\end{subfigure}
\caption{A buggy Go code snippet on channels with two possible executions}
\figlabel{go-language}
\end{figure}

\myparagraph{Channel operations} 
\figref{go-code-demo} presents a snippet showing the basic channel operations in Go.
The main thread creates an asynchronous channel of capacity~2 (Line~2),
and passes it as an argument to a child thread executing the goroutine (Line~3--Line~5).
The child thread sends value $1$ to the channel (Line~4).
The main thread sends value $1$ to the channel (Line~6), and then receives from it (Line~7), before closing it (Line~8).

\myparagraph{Consistency checking in predictive testing}
Observe that the program in \figref{go-code-demo} has a bug:~the main thread may execute all its operations and close the channel before the child thread begins to execute.
This will cause the child thread to attempt to send to a closed channel, causing the program to panic.
Like with many concurrency bugs, exposing this faulty program behavior depends 
on the scheduler and can be quite challenging.
One popular approach to tackle this challenge is through \emph{predictive} runtime testing~\cite{Said11,mathur2021optimal,wcp2017}.
Here, as the first step, the program is executed without any explicit schedule control, giving rise to an execution $\trace$.
Due to the lack of explicit control, $\trace$ is likely to be error-free, i.e., it does not expose the presence of a bug.
\figref{go-code-demo-valid-exec} shows such an execution of the program in \figref{go-code-demo}.
In the second step, $\trace$ is analyzed with the goal of constructing an alternative execution $\trace'$ that exposes the bug.
Here $\trace'$ is a permutation of (a slice of) $\trace$ and is required to be \emph{sound}, i.e.,
it can provably be executed by any program that produced $\trace$.
\figref{go-code-demo-buggy-exec} shows such a permutation $\trace'$.

The soundness requirement for $\trace'$ essentially entails a consistency check.
Specifically, each thread must execute the same sequence of operations 
in $\trace'$ as it did in $\trace$, but the interleaving between
threads may differ.
This leads us to consider an \emph{abstract execution} that can be extracted
from $\trace$ and captures the per-thread
event sequences together with additional ordering constraints;
for instance, in our example, we additionally require that
$\closeCh(\texttt{ch})$ of thread $\thread_1$ executes before $\snd(\texttt{ch}, 1)$ of thread $\thread_2$, 
while leaving the precise interleaving unspecified. 
The core question then becomes: can this abstract execution be realized as a concrete
trace $\trace'$ that respects channel semantics? This is precisely the consistency checking problem.


\subsection{The consistency checking problem: Two variants}
\revision{
We begin with a brief overview of the message-passing consistency checking problem, which is formally defined in \secref{preliminaries}.
In general, this problem takes as input a set $\eventSet$ of send and receive events, along with a partial order $P$ over $\eventSet$.
The partial order $P$ captures the ordering of events within each thread and may also order a channel send before its corresponding receive.
The objective is to determine whether there exists a total order on $\eventSet$ that is consistent with $P$ and satisfies the FIFO constraints induced by the channels.
As an intuitive example, the input to the consistency checking problem may be the set of $6$ events in the execution of \figref{go-code-demo-valid-exec} together with the program order (for example, the send event $\wt(\ch, 1)$ of thread $\tau_1$ is ordered before the receive event $\rd(\ch, 1)$ of the same thread), but without the interleaving of the two threads (and the output would be true in this case, since this set can be linearized to both the total order of \figref{go-code-demo-valid-exec} as well as of \figref{go-code-demo-buggy-exec}).
}

More formally, the message-passing consistency problem takes as input either a pair
$\tuple{\AbstractExecution, \cpFunc}$ or a triplet $\tuple{\AbstractExecution, \cpFunc, \rf{}}$, where 
\begin{compactitem}
	\item $\AbstractExecution$ is an abstract execution of the 
	form $\AbstractExecution=\tuple{\eventSet, \po{}}$,
	where $\eventSet$ is a set of events  and $\po{}$ 
	(\emph{program order}) specifies a total order on 
	the events of each thread.
	The optional component $\rf{}$ (\emph{reads-from} relation) specifies, 
	for each channel receive event $\rcv$, the corresponding channel 
	send event $\snd$ that $\rcv$ obtains its value from. 

	\item The function 
	$\cpFunc\colon \channels{\AbstractExecution}\to \Nats$, 
	specifies the capacity of each channel;
	here $\channels{\AbstractExecution}$ 
	is the set of channels accessed in $\AbstractExecution$.
\end{compactitem}

In line with prior works on consistency testing~\cite{gibbons1994testing,tuncc2023optimal,chakraborty2024hard}, we distinguish between the following two variants.
\begin{compactitem}
\item The \emph{verify channel consistency ($\vch$)} problem takes as input 
the pair $\tuple{\AbstractExecution, \cpFunc}$ without any reads-from information.
This is the most general variant. 
\item The \emph{verify channel consistency with reads-from ($\vchRf$)}  problem takes
as input a triplet $\tuple{\AbstractExecution, \cpFunc, \rf{}}$ that contains reads-from information.
This variant naturally arises when, e.g., every write to a channel writes a distinct value (for example, this is often imposed during litmus testing~\cite{Alglave2011Litmus, Alglave2014Herding}), or as a general abstraction mechanism~\cite{Chalupa2018,Abdulla2019b,Kokologiannakis2022}.
\end{compactitem}
In each case, the task is to find a linear trace $\trace$ realizing $\AbstractExecution$, i.e., $\trace$ consists of the events $\eventSet$ and agrees with $\AbstractExecution$ on the $\po{}$ (and $\rf{}$, in the case of $\vchRf$).

\begin{remark}\label{rem:shared_memory}
For simplicity of presentation, we consider that all interthread communication occurs via channels, and there is no shared memory.
This is not a limitation, since a shared register can be simulated by a channel of capacity~1, as we prove in~\secref{vch-rf-np-hard-k=1}.
\end{remark}

\subsection{Summary of Results}
We now present the main results of the paper, summarized in \tabref{results_vch} and \tabref{results_vchrf},
while we refer to the following sections for details.
%
In the following, we let $\NumEvents$, $\NumThreads$ and $\NumChannels$ be the 
total number of events, threads and channels, respectively, in $\AbstractExecution$. 	
We also let $\MaxCapacity=\max_{\ch} \cp{\ch}$ be the maximum channel capacity.
	To capture common paradigms of channel programming, 
	we distinguish between channels $\ch$ that are \emph{synchronous} $(\cp{\ch}=0)$, 
	\emph{capacity-bounded} or \emph{capacity-unbounded}.
	We remark that, in our setting, $\ch$ is regarded as capacity-unbounded 
	if $\AbstractExecution$ contains $\leq \cp{\ch}$ $\wt$ events to $\ch$, 
	since then $\ch$ cannot block, 
	regardless of how $\AbstractExecution$ is 
	scheduled\footnote{This is in contrast to the colloquial 
	use of ``unbounded'' meaning ``of infinite capacity''.}.
	For example, if $\cp{\ch}=3$ but $\AbstractExecution$ only 
	contains two send events to $\ch$, then $\ch$ behaves as 
	a capacity-unbounded channel in $\AbstractExecution$ 
	(even though its capacity is finite).


\begin{table}
\caption{\tablabel{results_vch}
Results for the channel consistency problem $\vch$ on abstract executions of $\NumEvents$ events, $\NumThreads$ threads, $\NumChannels$ channels, each with capacity $\leq \MaxCapacity$.
}
\small
\begin{tabular}{l l l}
\toprule
\textbf{Reference} & \textbf{Variant} & \textbf{Complexity} \\
\midrule
\thmref{vch-same-value-hardness} & Every event sends/receives the same value & $\NP$-complete\\
\thmref{vch-two-threads-hardness} & $\NumThreads=2$ and each channel is capacity-unbounded & $\NP$-complete\\ 
\thmref{vch-one-chan-hardness} & $\NumChannels=1$ and either $\MaxCapacity=0$ (synchronous channel) or $\MaxCapacity=1$ & $\NP$-complete\\
\thmref{vch-solution} & General case & $O\left(\NumEvents ^ {\NumThreads + 1} \cdot \NumThreads ^{\MaxCapacity \NumChannels}\right)$\\ 
\bottomrule
\end{tabular}
\end{table}

To illustrate the intricacies of channels, 
we begin with two restricted cases of $\vch$ for which the problem is nevertheless intractable.
First, let us consider the case where all channel 
events send/receive the same value, and thus each receive may observe from any send.
In this setting, we show the following 
\revision{
via a reduction from the Hamiltonian cycle problem.
}
\begin{restatable}{theorem}{vchSameValHardness}
\thmlabel{vch-same-value-hardness}
$\vch$ is $\NP$-complete even if all events send/receive the same value.
\end{restatable}

The corresponding consistency problem for shared memory is trivial:~as reads/writes are on the same value, any linearization that starts with a write is a valid trace.
This is not the case for $\vch$, as $\trace$ must also respect channel capacities.
Second, we show that the problem is intractable already with just two threads
\revision{
using a reduction from postive 1-in-3 3SAT
}.
\begin{restatable}{theorem}{vchTwoThreadsHardness}
\thmlabel{vch-two-threads-hardness}
$\vch$ is $\NP$-complete even if $\NumThreads = 2$ and each channel is capacity-unbounded. 
\end{restatable}

In contrast, the smallest number of threads which make consistency for 
shared memory intractable is $\NumThreads=3$~\cite{gibbons1997testing}.
Third, we show that the problem becomes intractable already with just a single channel, which is either synchronous or has capacity $1$ 
\revision{
using a reduction from the VSC-read studied and proved to be $\NP$-hard in~\cite{gibbons1997testing}; the VSC-read problem is the analogue of the $\vchRf$ problem we study in this work, but for the case of executions with registers instead of channels.
}
This result is analogous to the hardness for shared memory on a single location~\cite{cantin2005complexity}
(but is not subsumed by it, since synchronous channels are blocking, in contrast to shared memory).

\begin{restatable}{theorem}{vchOneChanHardness}
\thmlabel{vch-one-chan-hardness}
$\vch$ is $\NP$-complete even if $\NumChannels = 1$ and either $\MaxCapacity=0$ (synchronous channel) or $\MaxCapacity = 1$. 
\end{restatable}

Given the above hardness results even on restricted inputs, it is imperative to ask --- how fast can we solve $\vch$ in general? 
The following theorem establishes an upper bound 
with explicit dependence on the input parameters.
\revision{
Our algorithm for checking $\vch$ traverses a \emph{frontier graph} whose nodes track succinct information about configurations of threads and channels.
}

\begin{restatable}{theorem}{vchSolution}
\thmlabel{vch-solution}
$\vch$ can be solved in $O\left(\NumEvents ^ {\NumThreads + 1} \cdot \NumThreads ^{\MaxCapacity \NumChannels}\right)$ time.
\end{restatable}


\begin{table}
\caption{\tablabel{results_vchrf}
Results for the channel consistency problem with a reads-from relation $\vchRf$ on abstract executions of $\NumEvents$ events, $\NumThreads$ threads, $\NumChannels$ channels, each with capacity $\leq \MaxCapacity$.
($\dagger$) holds under SETH.
}
\small
\newcolumntype{L}[1]{>{\raggedright\arraybackslash}p{#1}}
\begin{tabular}{L{1.8cm} L{7.7cm} L{3.1cm}}
\toprule
\textbf{Reference} & \textbf{Variant} & \textbf{Complexity} \\
\midrule
\thmref{vch-rf-solution} & General case  & $O(\NumEvents^{\NumThreads + 1} \cdot (\min(\MaxCapacity!, \NumThreads^\MaxCapacity)^\NumChannels)$  \vspace{0.2cm} \\
\thmref{vch-rf-hardness} & $\MaxCapacity=1$ and every channel is asynchronous, or & $\NP$-complete\\
& $\NumThreads = 3$ and $\MaxCapacity = 2$, or &\\
& $\NumThreads = 3$ and $\NumChannels = 5$ and each channel is capacity-unbounded & \vspace{0.2cm}\\
\thmref{vch-rf-tree-topology} & Acyclic topology and each channel has capacity $\leq 1$ or is unbounded  & $O(\NumEvents^2)$\vspace{0.2cm} \\
\thmref{vch-rf-two-threads-lower-bound} & $\NumThreads=2$ and each channel has capacity 1, or  & Not in$^{\dagger}$ $O(\NumEvents^{2-\epsilon})$  \\
& $\NumThreads=2$ and each channel is capacity-unbounded  \vspace{0.2cm}\\
\thmref{vch-rf-sync-upper} & Each channel is synchronous  & $O(\NumEvents)$ \\
\bottomrule
\end{tabular}
\end{table}

Let us now turn our attention to the simpler problem, $\vchRf$.
Since $\vchRf$ is a special case of $\vch$, the upper bound in \thmref{vch-solution} 
also applies to $\vchRf$. 
We show that \vchRf admits, in fact, a somewhat faster algorithm \revision{using a more succinct frontier graph}.

\begin{restatable}{theorem}{vchRfsolution}
\thmlabel{vch-rf-solution}
$\vchRf$ can be solved in $O(\NumEvents^{\NumThreads + 1} \cdot (\min(\MaxCapacity!, \NumThreads^\MaxCapacity)^\NumChannels)$ time.
\end{restatable}

Observe that both upper bounds (\thmref{vch-solution} and \thmref{vch-rf-solution}) become polynomial when the input parameters are bounded (i.e., fixed constants).
When this is not the case, we ask whether one has to suffer an exponential dependence 
on each of these parameters.
In other words, \emph{does the problem become tractable when only some, but not all, of the parameters are bounded?}
Unfortunately, as the next theorem states, even the easier problem $\vchRf$ remains intractable when only some parameters are bounded, 
\revision{via reductions from 3SAT and VSC-read problem.}

\begin{restatable}{theorem}{vchRfHardness}
\thmlabel{vch-rf-hardness}
$\vchRf$ is $\NP$-complete if any of the following three conditions holds:
(i)~ $\MaxCapacity = 1$ and every channel is asynchronous, or
(ii)~$\NumThreads = 3$ and $\MaxCapacity = 2$, or
(iii)~$\NumThreads = 3$ and $\NumChannels = 5$ and each channel is capacity-unbounded.
\end{restatable}

Given the hardness of \thmref{vch-rf-hardness}, 
the next natural question is whether $\vchRf$ becomes tractable for any natural (semantic) classes besides the (syntactic) restrictions governed by the parameters above.
Towards this, we consider the \emph{communication topology} $G=(V,E)$ of $\AbstractExecution$, 
where $V$ contains the set of threads of $\AbstractExecution$, 
and we have an edge $(\thread_1, \thread_2)\in E$ iff threads $\thread_1$ and $\thread_2$ access a common channel.
We prove that the problem becomes tractable when $G$ is acyclic 
\revision{
and can be solved via a reduction to the satisfiability problem of a quadratically long 2CNF formula.
}
.

\begin{restatable}{theorem}{vchRfTreeTopology}
\thmlabel{vch-rf-tree-topology}
$\vchRf$ is solvable in $O(\NumEvents^2)$ time on acyclic communication topologies 
if each channel is either capacity-unbounded or has capacity $\leq 1$.
\end{restatable}

Common acyclic topologies include pipelines, server-client architectures, and general tree structures.
We remark that \thmref{vch-rf-tree-topology} allows for any combination of channels that are capacity-unbounded, have capacity $1$, or are synchronous (i.e., have capacity $0$).
Observe that the case $\NumThreads = 2$ results in an acyclic communication topology.
Due to \thmref{vch-two-threads-hardness}, an analogous polynomial bound for $\vch$ on acyclic topologies is not possible, as the problem is $\NP$-complete already for $\NumThreads=2$ threads.

At this point, it is natural to ask whether any improvements are possible over this quadratic bound, e.g., does the problem admit a linear-time solution on acyclic topologies?
To answer this question, we equip techniques from fine-grained complexity theory, and in particular, the popular strong exponential time hypothesis (SETH).
We establish the following lower bound 
\revision{
with a reduction from \emph{Orthogonal Vector} (OV) problem, which is hard for quadratic time conditioned on SETH.
}

\begin{restatable}{theorem}{vchRFTwoThreadsLowerBound}
\thmlabel{vch-rf-two-threads-lower-bound}
Under SETH, $\vchRf$ cannot be solved in time $O(\NumEvents^{2 - \epsilon})$ for any $\epsilon > 0$, 
even if $\NumThreads = 2$ and either
(i)~all channels are capacity-unbounded, or
(ii)~all channels have capacity $1$.
\end{restatable}
Together, \thmref{vch-rf-tree-topology} and \thmref{vch-rf-two-threads-lower-bound} yield a tight dichotomy:~the problem takes quadratic time on acyclic topologies, and this bound is optimal, even for the simplest such topology. 
Finally, we consider fully synchronous channels, and show that \revision{the problem reduces to cycle detection on a graph}, which admits a linear time algorithm
(no matter what topology).
\begin{restatable}{theorem}{vchRfSyncUpper}
\thmlabel{vch-rf-sync-upper}
$\vchRf$ is solvable in $O(\NumEvents)$ time if all channels are synchronous.
\end{restatable}
Observe that this is in sharp contrast to $\vch$, for which the problem is intractable already with only one synchronous channel (\thmref{vch-one-chan-hardness}).

\myparagraph{Overview of empirical evaluation}
We have implemented our algorithm for channel consistency with 
reads-from (\thmref{vch-rf-solution}), primarily to demonstrate 
its efficacy over a 
vanilla approach of encoding the ($\NP$-complete) channel consistency problem as an $\smt$ formula.
We evaluated the performance of our algorithm 
as well as the vanilla $\smt$ solving based approach
on a comprehensive suite of 103 benchmarks derived from real-world Golang programs. 
The results indicate that our algorithm exhibits superior scalability compared to a $\smt$-based approach, achieving a faster completion time while encountering fewer timeouts. 
Furthermore, despite $\vchRf$ being an $\NP$-hard problem, an optimized version of our algorithm successfully scales to large instances, handling up to ~35k events, ~2k threads, and ~14k channels. 
These findings confirm our hypothesis that our frontier graph based algorithm (\thmref{vch-rf-solution}) is a highly efficient solution for channel consistency checking. 

\myparagraph{Outline}{
The technical parts of the paper are organized as follows.
In \secref{preliminaries}, we set up relevant notation and define the consistency problem for channels.
In \secref{hardness-algo}, we develop algorithms for the upper bounds in \thmref{vch-solution}, \thmref{vch-rf-solution}, \thmref{vch-rf-tree-topology} and \thmref{vch-rf-sync-upper}.
Finally, in \secref{lower-bounds-vchRf} we prove item (iii) of \thmref{vch-rf-hardness} and \thmref{vch-rf-two-threads-lower-bound}. 
Due to space restrictions, all formal proofs are relegated to the appendix.
Moreover, the remaining theorems, 
namely \thmref{vch-same-value-hardness}, \thmref{vch-two-threads-hardness}, \thmref{vch-one-chan-hardness}, and items (i) and (ii) of \thmref{vch-rf-hardness} are proven in the appendix.
}


\section{Preliminaries}
\seclabel{preliminaries}
In this section we formalize the basic concepts of channel-based executions and define the corresponding consistency-checking problems. 

\subsection{Events and executions}

\myparagraph{Channels}{
We model channels as $\fifo$ queues with (bounded or unbounded) capacities.
A \emph{send} operation on a channel $\ch$ enqueues a message to the $\fifo$ queue,
while a \emph{receive} operation pops a message from the queue.
The \emph{capacity} $\cp{\ch}$ of $\ch$ dictates how many messages can be enqueued in it simultaneously.
When $\ch$ is full (i.e., contains $\cp{\ch}$ messages), send operations on it will block, until at least one receive operation is executed on it.
We further call $\ch$ \emph{synchronous} if $\cp{\ch} = 0$.
Intuitively synchronous channels do not buffer any messages, and thus
a send operation on $\ch$ must be immediately followed by a receive operation.
An \emph{asynchronous} channel $\ch$, on the other hand, has $\cp{\ch} > 0$ and allows for
asynchronous send and receive operations. 
}

\myparagraph{Events}{
An event is a tuple $e= \ev{id, \thread, \op(\ch, \val)}$, consisting of the unique identifier $id$ of $e$,
the identifier $\thread$ of the thread that performs $e$,
the operation $\op \in \set{\snd, \rcv}$ (a channel send or receive)\footnote{Our results are easily extended to a setting that contains other common events such as thread \texttt{fork}/\texttt{join} and channel \texttt{create}/\texttt{close}. We omit such events for ease of presentation.} performed by $e$, the
identifier of the channel $\ch$ involved in the event $e$  
and the value $\val$ sent or received. 
We write $\ThreadOf{e}$, $\OpOf{e}$, $\ChOf{e}$, $\ValueOf{e}$ 
for the thread, operation, channel and value of $e$, respectively. 
We often use the more succinct notation $\snd(\ch, \val)$/ $\rcv(\ch,\val)$, when the unique identifier $id$ and thread identifier $\thread$ are clear from the context, or not important.
}

\myparagraph{Executions and well-formedness}{
An execution is a finite sequence of events $\tr = e_1e_2 \ldots e_n$ of length $|\tr| = \NumEvents$.
We denote 
by $\events{\tr}=\{e_1,\dots, e_n\}$ the set of events,
by  $\threads{\tr}$ the set of threads, and
by $\channels{\tr}$ the set of channels appearing in $\tr$.
For some channel $\ch \in \channels{\tr}$, we use $\proj{\trace}{\ch}$ 
to denote the maximal subsequence of $\tr$ containing only events accessing $\ch$. 
Likewise, we use $\proj{\trace}{\snd(\ch)}$ (resp. $\proj{\trace}{\rcv(\ch)}$) 
to denote the projection of $\tr$ onto the send (resp. receive) events on $\ch$. 
We require that executions are \emph{well-formed}, meaning that they respect the channel semantics.
Well-formedness requires that $\tr$ satisfies the following two types of constraints.
}

\mysubparagraph{Capacity Constraints}{
These require that $\tr$ respects the channel capacities.
In particular, for each channel $\ch \in \channels{\tr}$, the following hold.
\begin{compactenum}
\item (\emph{Asynchronous channels}) If $\cp{\ch} > 0$, then for each prefix $\prefix$ of $\tr$, we have
\[
|\,\proj{\prefix}{\rcv(\ch)}\,| \le |\,\proj{\prefix}{\snd(\ch)}\,| \le |\,\proj{\prefix}{\rcv(\ch)}\,| + \cp{\ch}.
\] 
In other words, every receive event should observe a send event and the number of buffered send events cannot exceed the channel capacity. 
\item (\emph{Synchronous channels}) If $\cp{\ch} = 0$, then
each send event $e = \ev{\thread, \snd(\ch)}$ on $\ch$ must immediately be followed by 
a matching receive event $e' = \ev{\thread', \rcv(\ch)}$
from a different thread $\thread' \neq \thread$. 
Likewise, each receive event $e = \ev{\thread, \rcv(\ch)}$ on $\tr$ must be immediately preceded
by a matching send event $e' = \ev{\thread', \snd(\ch)}$ from a different thread $\thread' \neq \thread$.
Observe that  this implies an equal number of send and receive events on $\ch$.
\end{compactenum}
}
A thread attempting to send on a full channel is blocked (normally by the runtime),
until the channel is received, freeing up space for the new incoming message.
The events listed in $\tr$ are executed events, meaning that each channel send completed successfully, and was thus performed on a non-full channel.
For synchronous channels, a send operation is executed simultaneously with its matching receive, since capacity $0$ does not allow storing the message sent.

\mysubparagraph{Value Constraints}{
These require that matching $\snd$/$\rcv$ events on the same channel observe identical values.
In particular, for each channel $\ch \in \channels{\tr}$, for each $1 \leq i \leq |\proj{\trace}{\rcv(\ch)}|$,
if the $i$-th send (resp., receive) event in $\ch$ is $\snd(\ch, \val_1)$ (resp., $\rcv(\ch, \val_2)$),
then $\val_1=\val_2$.
}

\begin{figure}[htbp!]
\centering
\begin{subfigure}[b]{0.24\textwidth}
\centering
\scalebox{0.9}{
\execution{2}{
        \figev{1}{\wt(\ch_1, 1)}
	\figev{1}{\wt(\ch_1, 2)}
        \figev{2}{\rd(\ch_1, 1)}
	\figev{2}{\rd(\ch_1, 2)}
        \figev{2}{\wt(\ch_2, 3)}
	\figev{1}{\rd(\ch_2, 3)}
}
}
\caption{Execution $\tr_1$}
\figlabel{well-formed}
\end{subfigure}
\hfill
\begin{subfigure}[b]{0.24\textwidth}
\centering
\scalebox{0.9}{
\execution{2}{
        \figev{1}{\wt(\ch_1, 1)}
	\figev{1}{\wt(\ch_1, 2)}
        \figev{2}{\wt(\ch_1, 3)}
	\figev{1}{\rd(\ch_1, 1)}
	\figev{2}{\rd(\ch_1, 2)}
	\figev{2}{\rd(\ch_1, 3)}
}
}
\caption{Execution $\tr_2$}
\figlabel{ill-formed-capacity-async}
\end{subfigure}
\hfill
\begin{subfigure}[b]{0.24\textwidth}
\centering
\scalebox{0.9}{
\execution{2}{
        \figev{1}{\wt(\ch_1, 1)}
	\figev{1}{\wt(\ch_2, 2)}
	\figev{1}{\rd(\ch_1, 1)}
        \figev{2}{\wt(\ch_1, 3)}
	\figev{2}{\rd(\ch_2, 2)}
	\figev{2}{\rd(\ch_1, 3)}
}
}
\caption{Execution $\tr_3$}
\figlabel{ill-formed-capacity-sync}
\end{subfigure}
\hfill
\begin{subfigure}[b]{0.24\textwidth}
\centering
\scalebox{0.9}{
\execution{2}{
        \figev{1}{\wt(\ch_1, 1)}
	\figev{1}{\wt(\ch_1, 2)}
	\figev{1}{\rd(\ch_1, 1)}
        \figev{2}{\wt(\ch_1, 3)}
	\figev{2}{\rd(\ch_1, 3)}
	\figev{2}{\rd(\ch_1, 2)}
}
}
\caption{Execution $\tr_4$}
\figlabel{ill-formed-value}
\end{subfigure}
\caption{
Four executions on two channels $\ch_1$ and $\ch_2$ with capacities $\cp{\ch_1} = 2$ and $\cp{\ch_2} = 0$. 
Execution $\tr_1$ is well-formed but $\tr_2, \tr_3, \tr_4$ are not.
}
\figlabel{well-formedness}
\end{figure}
\raggedbottom

\begin{example}
Consider the four executions $\tr_1, \tr_2, \tr_3$ and $\tr_4$ in \figref{well-formedness}.
Each $\tr_i$ contains $6$ events and uses two channels $\ch_1$ and $\ch_2$ whose capacities
are $\cp{\ch_1} = 2$ (i.e., asynchronous channel) and $\cp{\ch_2} = 0$ (i.e., synchronous channel) respectively.
We use $e_i$ to denote the $i^\text{th}$ event of an execution. 
First, consider the execution $\tr_1$ (\figref{well-formed}), which is well-formed. 
The capacity constraint on $\ch_2$ is met because the (unique) send ($e_5$) and receive ($e_6$) events on $\ch_2$ appear consecutively.
Further, the two events access the same value.
Moreover, in every prefix of $\tr_1$, the number of buffered messages in $\ch_1$ never exceeds its capacity $2$, 
and the order of values being sent ($1 \to 2$)
matches that of the values being received on $\ch_1$, 
ensuring the value constraint for $\ch_1$ as well.
Now, consider $\tr_2$ in \figref{ill-formed-capacity-async},
which is not well-formed since, at $e_3$,  $\ch_1$ contains $3$ messages, exceeding its capacity.
Next, the execution $\tr_3$ in \figref{ill-formed-capacity-sync} is
not well-formed, because the send and receive events ($e_2$ and $e_5$)
on the synchronous channel $\ch_2$ are not consecutive.
Finally, the execution $\tr_4$ in \figref{ill-formed-value}
is not well-formed since the order of values sent ($1 \to 2 \to 3$)
is not the same as the order of values received ($1 \to 3 \to 2$).
\end{example}


\myparagraph{Trace order, program order and the reads-from relation}{
The \emph{trace order} of an execution $\tr$, denoted $\stricttrord{\tr}$, is the total order on $\events{\tr}$ induced by the sequence $\tr$.
The \emph{program order} $\po{\tr}$ of $\tr$ defines a total order on the events of each thread, i.e., for any two events $e_1, e_2 \in \events{\tr}$, we have $(e_1, e_2) \in \po{\tr}$ iff $e_1 \trord{\tr} e_2$ and $\ThreadOf{e_1} = \ThreadOf{e_2}$.
The (binary) \emph{reads-from} relation $\rf{\tr}$ induced by  $\tr$ maps receive events to their matching send events.
That is, $(\snd, \rcv) \in \rf{\tr}$,  iff there is a channel $\ch \in \channels{\tr}$ and some $i\in\Nats$ such that $\snd$ is the $i^\text{th}$ send event and $\rcv$ is the $i^\text{th}$ receive event on $\ch$.
We often use the shorthand $\rf{\tr}(\rcv)$ for the event $\snd$ such that $(\snd, \rcv) \in \rf{\tr}$.

}

\begin{example}
Consider again the execution $\tr_1$ in \figref{well-formed}.
We have $\rf{\tr_1}(e_3)=e_1$, $\rf{\tr_1}(e_4)=e_2$ and $\rf{\tr_1}(e_6)=e_5$.
We have $(e_1, e_3) \in \rf{\tr_1}$ and $(e_2, e_4)\in \rf{\tr_1}$. 
The program order of $\tr_1$ is $\po{\tr_1} = \set{(e_1, e_2), (e_2, e_6), (e_3, e_4), (e_4, e_5)}^+$, where, $R^+$ denotes the  transitive closure of the binary relation $R$.
\end{example}

\subsection{Verifying the Consistency of Message-Passing Concurrency}

We now state the consistency problem we study in this work.

\myparagraph{Abstract executions and consistency}{
The consistency problem is phrased on a pair $\tuple{\AbstractExecution, \cpFunc}$, 
where an \emph{abstract} execution $\AbstractExecution$ captures the local execution of each thread
and a capacity function $\cpFunc\colon \channels{\AbstractExecution}\to \Nats$ specifies the capacity of each channel,
where $\channels{\AbstractExecution}$ is the set of channels accessed by events in $\AbstractExecution$.
An abstract execution is a tuple
$\AbstractExecution = \tuple{\eventSet, \po{}}$,
where $\eventSet$ is some set of events, and $\po{}$
describes a per-thread total order on events in $\eventSet$.
An execution $\tr$ is a \emph{concretization} of $\AbstractExecution = \tuple{\eventSet, \po{}}$ with capacity function $\cpFunc$ if 
(i)$~\events{\tr} = \eventSet$,
(ii)~$\po{\tr} = \po{}$, and
(iii)~$\tr$ is well-formed with respect to the channel capacities specified by $\cpFunc$.
Finally, $\tuple{\AbstractExecution, \cpFunc}$ is \emph{consistent} if there exists an execution $\tr$ that concretizes it.
The consistency checking problem is thus formally stated below.

\begin{problem}[Verify channel consistency, $\vch$]
\problabel{VCH}
Given an abstract execution $\AbstractExecution = \tuple{\eventSet, \po{}}$ and capacity function $\cpFunc$, decide if $\tuple{\AbstractExecution, \cpFunc}$ is consistent.
\end{problem}

\myparagraph{Consistency with a reads-from relation}{
The \emph{consistency problem with a reads-from relation} is a tuple 
$\tuple{\AbstractExecution, \cpFunc, \rf{}}$, where $\eventSet$ and $\po{}$ are, as before, respectively a set of events and a per-thread total order on this set,
while $\rf{}$ matches send and receive events of $\eventSet$ on the same channel.
An execution $\tr$ concretizes $\AbstractExecution = \tuple{\eventSet, \po{}}$ and $\rf{}$
if it concretizes $\tuple{\eventSet, \po{}}$ (as in $\vch$), and moreover $\rf{\tr}=\rf{}$. 
\revision{In later sections, we omit the values of events in $\vchRf$ instances, as the values are not relevant once the reads-from relation is given.}
The corresponding consistency problem is defined analogously.

\begin{problem}[Verify channel consistency with reads-from, $\vchRf$]
\problabel{VCH-rf}
Given an abstract execution $\AbstractExecution = \tuple{\eventSet, \po{}}$ with reads-from relation $\rf{}$ and capacity function $\cpFunc$, decide if $\tuple{\AbstractExecution, \cpFunc, \rf{}}$ is consistent.
\end{problem}
}
It is not hard to see that $\vchRf$ is an easier problem than $\vch$, in the sense that the former is a special case of the latter (e.g., by requiring that every send uses a unique value).


\begin{figure}[t]
\centering
\begin{subfigure}[b]{0.46\textwidth}
\newcommand{\xstep}{2.8}
\newcommand{\ystep}{-0.8}
\newcommand{\height}{0.5}
\newcommand{\wid}{1.6}
\tikzstyle{event}=[draw=black, align=center, fill=white, line width=0.8pt, minimum width=\wid cm,minimum height=\height cm]
\centering
\scalebox{0.9}{
\begin{tikzpicture}[font=\small]
\node[event] (t1-1) at (0 * \xstep, 0 * \ystep)  {$\snd_1(\ch, 1)$};
\node[event] (t1-2) at (0 * \xstep, 1 * \ystep)  {$\snd_2(\ch, 2)$};
\node[event] (t1-3) at (0 * \xstep, 2 * \ystep)  {$\rcv_1(\ch, 2)$};

\node[event] (t2-1) at (1 * \xstep, 0 * \ystep) {$\rcv_2(\ch, 1)$};
\node[event] (t2-2) at (1 * \xstep, 1 * \ystep)  {$\rcv_3(\ch, 2)$};
\node[event] (t2-3) at (1 * \xstep, 2 * \ystep) {$\snd_3(\ch, 2)$};

\node[threadName] at (0 * \xstep, -1.0 * \ystep) {$\thread_{1}$};
\node[threadName] at (1 * \xstep, -1.0 * \ystep) {$\thread_{2}$};

\begin{scope}[on background layer]
\draw [thread] (0 * \xstep, -0.6 * \ystep) -- (0 * \xstep, 2.7 * \ystep);
\draw [thread] (1 * \xstep, -0.6 * \ystep) -- (1 * \xstep, 2.7 * \ystep);
\end{scope}
\end{tikzpicture}
}
\caption{A $\vch$ instance $\tuple{\AbstractExecution_1, \cpFunc_1}$}
\figlabel{vch-demo}
\end{subfigure}
\hfill
\begin{subfigure}[b]{0.53\textwidth}
\newcommand{\xstep}{3}
\newcommand{\ystep}{-0.8}
\newcommand{\height}{0.5}
\newcommand{\wid}{1.6}
\tikzstyle{event}=[draw=black, align=center, fill=white, line width=0.8pt, minimum width=\wid cm,minimum height=\height cm]
\tikzstyle{edge}=[->, line width=0.4mm, color=colorRF]
\centering
\scalebox{0.9}{
\begin{tikzpicture}[font=\small]
\node[event] (t1-1) at (0 * \xstep, 0 * \ystep) {$\snd_1(\ch)$};
\node[event] (t1-2) at (0 * \xstep, 1 * \ystep) {$\rcv_3(\ch)$};
\node[event] (t1-3) at (0 * \xstep, 2 * \ystep) {$\rcv_2(\ch)$};

\node[event] (t2-1) at (1 * \xstep, 0 * \ystep) {$\snd_2(\ch)$};
\node[event] (t2-2) at (1 * \xstep, 1 * \ystep) {$\snd_3(\ch)$};
\node[event] (t2-3) at (1 * \xstep, 2 * \ystep) {$\rcv_1(\ch)$};
    
\draw[rfEdge, out=0, in=180] (t1-1) to (t2-3);
\draw[rfEdge] (t2-2) to (t1-2);
\draw[rfEdge, out=180, in=0] (t2-1) to (t1-3);

\node [threadName] at (0 * \xstep, -1.0 * \ystep) {$\thread_{1}$};
\node [threadName] at (1 * \xstep, -1.0 * \ystep) {$\thread_{2}$};
    
\begin{scope}[on background layer]
\draw [thread] (0 * \xstep, -0.6 * \ystep) -- (0 * \xstep, 2.7 * \ystep);
\draw [thread] (1 * \xstep, -0.6 * \ystep) -- (1 * \xstep, 2.7 * \ystep);
\end{scope}
\end{tikzpicture}
}
\caption{A $\vchRf$ instance $\tuple{\AbstractExecution_2, \cpFunc_2, \rf{}}$}
\figlabel{vch-rf-demo}
\end{subfigure}
\caption{
A positive $\vch$ instance (\subref{fig:vch-demo}) and a  negative $\vchRf$ instance (\subref{fig:vch-rf-demo}). 
$\cpFunc_1(\ch) = \cpFunc_2(\ch) = 1$. 
\revision{Event subscripts are used to distinguish send/receive events. The same convention applies in subsequent figures.} 
}
\figlabel{vsc-example}
\end{figure}

\begin{example}
\figref{vch-demo} is a positive instance of $\vch$, witnessed by the execution $\tr_1 = \snd_1 \cdot \rcv_2 \cdot \snd_2 \cdot \rcv_3 \cdot \snd_3 \cdot \rcv_1$.
\figref{vch-rf-demo} is a negative  instance of $\vchRf$.
This is because any execution $\tr$ that concretizes $\tuple{\AbstractExecution_2, \cpFunc_2, \rf{}}$ must satisfy $\rcv_3 \stricttrord{\tr} \rcv_2$ and $\snd_2 \stricttrord{\tr} \snd_3$, due to the imposed program order.
The former, however, implies $\snd_3 \stricttrord{\tr} \snd_2$, contradicting the latter.
\end{example}

\section{Algorithms For Checking Consistency}
\seclabel{hardness-algo}

In this section we present algorithms for solving $\vch$ and $\vchRf$.
In particular, in \secref{general-solution} we develop the general algorithms for the two problems, leading to \thmref{vch-solution} and \thmref{vch-rf-solution}.
Then, in \secref{vch-rf-sync-solution-sec}, we focus on the special case of fully synchronous channels, and develop an efficient (linear-time) algorithm towards \thmref{vch-rf-sync-upper}.
Finally, in \secref{tree-topo-upper-bound} we focus on acyclic communication topologies, and develop a quadratic-time algorithm towards \thmref{vch-rf-tree-topology}.

\subsection{Algorithms for $\vch$ and $\vchRf$}
\seclabel{general-solution}
We now present our algorithms  for $\vch$ and $\vchRf$.
A naive algorithm for either problem would enumerate all possible permutations of the input set of events and look for one permutation that serves as the witness of consistency.
However, this approach takes $\Omega(\NumEvents!)$ time, which is significantly worse than the bounds we aim for.

Our algorithms for each problem circumvent this prohibitive complexity by succinctly encoding executions as paths in a \emph{frontier graph}, which has polynomial size when the number of threads $\NumThreads$, the number of channels $\NumChannels$ and  the maximum channel capacity $\MaxCapacity$ are bounded.


\revision{
Frontier graphs have previously been developed for consistency testing under shared memory~\cite{gibbons1994testing,Abdulla2019b,Agarwal2021}, but not for channel-based concurrency.
In the shared-memory setting, each read observes the most recent write, so nodes in the frontier graph only need to record the latest value of each memory location. 
In contrast, channels pose greater challenges for constructing succinct frontier graphs due to the vastly larger space of possible consistent executions. 
Here, nodes must also encode channel contents, making state counting both complex and subtle. 
}

\myparagraph{The frontier graph for $\vch$}{
Given a $\vch$ instance $\tuple{\AbstractExecution, \cpFunc}$ where $\AbstractExecution = \tuple{\eventSet, \po{}}$, we define its frontier graph $\gfrontier = (V, E)$ as follows. 

\mysubparagraph{The node set $V$}{
Each node $v \in V$ is a tuple of the form $v = \tuple{\fgEventSet, \fgChanMap, \fgSyncSend}$. 
Intuitively, ${\fgEventSet}$ specifies the subset of events of $\AbstractExecution$ that an execution has executed when it reaches the corresponding node in $\gfrontier$.
${\fgChanMap}$ specifies the contents of the asynchronous channels, while ${\fgSyncSend}$ specifies the (at most one) send event on a synchronous channel that must be matched in the next step.
We now formally specify $\fgEventSet$, $\fgChanMap$ and $\fgSyncSend$ as follows.
\begin{compactenum}
\item $\fgEventSet \subseteq \eventSet$, and $\fgEventSet$ is downward closed with respect to $\po{}$, i.e., for each $(e, f) \in \po{}$ and if $f \in \fgEventSet$, then $e \in \fgEventSet$. 
Given a channel $\ch$, let  $\s{num}_{\fgEventSet}(\snd(\ch))$ and $\s{num}_{\fgEventSet}(\rcv(\ch))$ denote the number of send and receive events on $\ch$ in $\fgEventSet$.
First, we require that there is at most one synchronous channel $\ch$ with $\s{num}_{\fgEventSet}(\rcv(\ch)) = \s{num}_{\fgEventSet}(\snd(\ch)) - 1$, while for all other synchronous channels $\ch'$, we have $\s{num}_{\fgEventSet}(\rcv(\ch')) = \s{num}_{\fgEventSet}(\snd(\ch'))$. 
Second, we require that for any asynchronous channel $\ch$, 
the following holds. 
\[
\s{num}_{\fgEventSet}(\rcv(\ch)) \le \s{num}_{\fgEventSet}(\snd(\ch)) \le \s{num}_{\fgEventSet}(\rcv(\ch)) + \cp{\ch}
\]

\item ${\fgChanMap}\colon \channels{\AbstractExecution} \to {\fgEventSet}^{\leq k}$ maps each asynchronous channel $\ch$ in $\eventSet$ (i.e., $\cp{\ch} > 0$)  to a sequence of events in ${\fgEventSet}$, whose length is bounded by $\cp{\ch}$, i.e., ${\fgChanMap}(\ch) = e_1 \cdot e_2 \cdots e_p$, where $0 \leq p \leq \cp{\ch} \leq k$, and $e_1, \ldots, e_p \in {\fgEventSet}$.
Moreover, for any asynchronous channel $\ch$, ${\fgChanMap}$ must satisfy that 
if some event $e$ appears in ${\fgChanMap}(\ch)$, then $e$ is one of the last $|{\fgChanMap}(\ch)|$ send events to channel $\ch$ in thread $\ThreadOf{e}$.
\item ${\fgSyncSend}$ is either $\bot$ or points to a send event of a synchronous channel. 
{\fgSyncSend}n particular, if there is a synchronous channel $\ch$ such that $\s{num}_{\fgEventSet}(\rcv(\ch)) = \s{num}_{\fgEventSet}(\snd(\ch)) - 1$, then ${\fgSyncSend}=\snd$, for some send event $\snd$ on $\ch$.
Otherwise, ${\fgSyncSend}=\bot$.
\end{compactenum}

Finally, we have a distinguished \emph{source node}, defined as $\tuple{\emptyset, \lambda~\ch. \epsilon, \bot}$, 
as well as one or more \emph{sink nodes}, defined as $\tuple{\eventSet, {\fgChanMap}, \bot}$.
{\fgSyncSend}n words, the source node captures the case that no event of $\AbstractExecution$ has been executed, while a sink node captures that all events of $\AbstractExecution$ have been executed (sink nodes might differ on the contents of the channels ${\fgChanMap}$, containing messages that are never received).
}

\mysubparagraph{The edge set $E$}{
Concrete executions that serve as potential witnesses of the consistency of $\tuple{\AbstractExecution, \cpFunc}$ are captured as paths 
in $\gfrontier$ starting from the source node.
An edge $(v_1, v_2)\in E$ intuitively captures whether \emph{any} execution reaching $v_1$ can be extended to $v_2$.
The information contained in $v_1$ is sufficient to decide whether this is possible.
{\fgSyncSend}n particular, let  $v_1 = \tuple{{\fgEventSet}_1, {\fgChanMap}_1, {\fgSyncSend}_1}$ and $v_2 = \tuple{{\fgEventSet}_2, {\fgChanMap}_2, {\fgSyncSend}_2}$.
We have $(v_1, v_2)\in E$ if there is an event $e \in \eventSet\setminus {\fgEventSet}_1$ such that ${\fgEventSet}_2 = {\fgEventSet}_1 \cup \set{e}$ and the following conditions hold, where $\ch = \ChOf{e}$.
\begin{compactenum}
\item {\fgSyncSend}f $\ch$ is asynchronous and $\OpOf{e} = \rcv$,  then we require that the following hold.
\begin{compactenum}
\item ${\fgSyncSend}_1 \pointsTo \bot, {\fgSyncSend}_2 \pointsTo \bot$.
\item ${\fgChanMap}_1(\ch) \neq \epsilon$, and the first event of ${\fgChanMap}_1(\ch)$, i.e., $e_{{\fgChanMap}_1, \ch, \s{first}} = {\fgChanMap}_1(\ch)[0]$ satisfies $\ValueOf{e_{{\fgChanMap}_1, \ch, \s{first}}} = \ValueOf{e}$.
Moreover, ${\fgChanMap}_2(\ch)$ is obtained by removing the first event of ${\fgChanMap}_1(\ch)$, i.e., ${\fgChanMap}_1(\ch) = e_{{\fgChanMap}_1, \ch, \s{first}} \cdot {\fgChanMap}_2(\ch)$.
\item For all other asynchronous channels $\ch' \neq \ch$, we have ${\fgChanMap}_2(\ch') = {\fgChanMap}_1(\ch')$. 
\end{compactenum}

\item {\fgSyncSend}f $\ch$ is asynchronous and $\OpOf{e} = \snd$, then we require that the following hold.
\begin{compactenum}
\item ${\fgSyncSend}_1 = \bot, {\fgSyncSend}_2 = \bot$.
\item $|{\fgChanMap}_1(\ch)| < $ $\cp{\ch}$, and ${\fgChanMap}_2(\ch)$ is obtained by appending $e$ at the end of ${\fgChanMap}_1(\ch)$, i.e., ${\fgChanMap}_2(\ch) = {\fgChanMap}_1(\ch) \cdot e$.
\item For all other asynchronous channels $\ch' \neq \ch$, we have ${\fgChanMap}_2(\ch') = {\fgChanMap}_1(\ch')$. 
\end{compactenum}
    
\item {\fgSyncSend}f $\ch$ is synchronous and $\OpOf{e} = \snd$,  then we require that
\begin{enumerate*}
\item ${\fgSyncSend}_1 \pointsTo \bot, {\fgSyncSend}_2 \pointsTo e$, and 
\item for all asynchronous channels $\ch'$, ${\fgChanMap}_1(\ch') = {\fgChanMap}_2(\ch')$. 
\end{enumerate*}

\item {\fgSyncSend}f $\ch$ is synchronous and $\OpOf{e} = \rcv$, then we require that
\begin{enumerate*}
\item ${\fgSyncSend}_1 \pointsTo e'\neq \bot, {\fgSyncSend}_2 \pointsTo \bot$, and $e'$ satisfies $\OpOf{e'} = \snd$, $\ChOf{e'} = \ch$, $\ValueOf{e'} = \ValueOf{e}$, and $\ThreadOf{e} \neq \ThreadOf{e'}$, and,
\item for all asynchronous channels $\ch'$, ${\fgChanMap}_1(\ch') = {\fgChanMap}_2(\ch')$. 
\end{enumerate*}
\end{compactenum} 
{\fgSyncSend}f the above hold, we say that the edge $(v_1, v_2)$ is labeled by $e$, and often write $v_1 \xrightarrow{e} v_2$.
See \figref{frontier-algo-demo} for an example.
The following lemma states that $\gfrontier$ captures the consistency of $\tuple{\AbstractExecution, \cpFunc}$.
}

\begin{restatable}{lemma}{vchConsToReach}
\lemlabel{consistency-to-reachability}
$\tuple{\AbstractExecution, \cpFunc}$ is consistent iff there is a sink node reachable from the source node in $\gfrontier$.
\end{restatable}

\begin{figure}[t]
\centering
\newcommand{\xstep}{1.7}
\newcommand{\ystep}{-0.8}
\newcommand{\ystepsnd}{-1.3}
\newcommand{\height}{0.5}
\newcommand{\wid}{1.2}

\begin{subfigure}[b]{0.42\textwidth}
\centering
\scalebox{0.9}{
\begin{tikzpicture}[font=\small]
\node [event] (n1) at (0 * \xstep, 0 * \ystep) {$\wt_1(\ch, 1)$};

\node [event] (n1) at (1 * \xstep, 1 * \ystep) {$\wt_2(\ch,2)$};
\node [event] (n1) at (1 * \xstep, 2 * \ystep) {$\rd_3(\ch,1)$};

\node [event] (n1) at (2 * \xstep, 3 * \ystep) {$\rd_4(\ch,2)$};

\node[threadName] at (0 * \xstep, -1.0 * \ystep) {$\thread_{1}$};
\node[threadName] at (1 * \xstep, -1.0 * \ystep) {$\thread_{2}$};
\node[threadName] at (2 * \xstep, -1.0 * \ystep) {$\thread_{3}$};

\begin{scope}[on background layer]
\draw [thread] (0 * \xstep, -0.6 * \ystep) -- (0 * \xstep, 3.7 * \ystep);
\draw [thread] (1 * \xstep, -0.6 * \ystep) -- (1 * \xstep, 3.7 * \ystep);
\draw [thread] (2 * \xstep, -0.6 * \ystep) -- (2 * \xstep, 3.7 * \ystep);
\end{scope}
\end{tikzpicture}
}
\caption{A $\vch$ instance $\tuple{\AbstractExecution, \cpFunc}$ with $\cpFunc(\ch)=2$.}
\figlabel{frontier-algo-demo-trace}
\end{subfigure}
\hfill
\begin{subfigure}[b]{0.57\textwidth}
\renewcommand{\xstep}{0.85}
\tikzstyle{labelEvent}=[text width=1cm]
\centering
\scalebox{0.75}{
\begin{tikzpicture}[font=\small]
\node [dotted, event] (n1) at (0,0) {$Y = \emptyset$\\$Q(\ch) = \epsilon$};
\node [event] (n2) at (-1.5*\xstep, 1 * \ystepsnd) {$Y = \set{\snd_1}$\\$Q(\ch) = \wt_1$};
\node [event] (n3) at (1.5*\xstep, 1 * \ystepsnd) {$Y = \set{\snd_2}$\\$Q(\ch) = \wt_2$};
\node [event] (n4) at (-3*\xstep,2 * \ystepsnd) {$Y = \set{\snd_1, \snd_2}$\\$Q(\ch) = \wt_1 \cdot \wt_2$};
\node [event] (n5) at (0.35*\xstep,2 * \ystepsnd) {$Y = \set{\snd_1, \snd_2}$\\$Q(\ch) = \wt_2 \cdot \wt_1$};
\node [event] (n6) at (3.5*\xstep,2 * \ystepsnd) {$Y = \set{\snd_2, \rcv_4}$\\$Q(\ch) = \epsilon$};
\node [event] (n7) at (-2.0*\xstep,3 * \ystepsnd) {$Y = \set{\snd_1, \snd_2, \rcv_3}$\\$Q(\ch) = \wt_2$};
\node [event] (n8) at (2.0*\xstep,3 * \ystepsnd) {$Y = \set{\snd_1, \snd_2, \rcv_4}$\\$Q(\ch) = \wt_1$};
\node [dashed, event] (n9) at (0*\xstep,4 * \ystepsnd) {$Y = \set{\snd_1, \snd_2, \rcv_3, \rcv_4}$\\$Q(\ch) = \epsilon$};

\draw[executeEdge] (n1) to node[left, pos=0.3]{$\snd_1$} (n2);
\draw[executeEdge] (n1) to node[right, pos=0.3]{$\snd_2$} (n3);
\draw[executeEdge] (n2) to node[left, pos=0.3]{$\snd_2$} (n4);
\draw[executeEdge] (n3) to node[left, pos=0.3]{$\snd_1$} (n5);
\draw[executeEdge] (n3) to node[right, pos=0.3]{$\rcv_4$} (n6);
\draw[executeEdge] (n4) to node[left, pos=0.6]{$\rcv_3$} (n7);
\draw[executeEdge] (n5) to node[left, pos=0.6]{$\rcv_4$} (n8);
\draw[executeEdge] (n6) to node[right, pos=0.7]{$\snd_1$} (n8);
\draw[executeEdge] (n7) to node[left, pos=0.6]{$\rcv_4$} (n9);
\draw[executeEdge] (n8) to node[right, pos=0.6]{$\rcv_3$} (n9);

\end{tikzpicture}
}
\caption{The frontier graph $\gfrontier$ for $\tuple{\AbstractExecution, \cpFunc}$}
\figlabel{frontier-algo-demo-graph}
\end{subfigure}
\caption{A $\vch$ instance (\subref{fig:frontier-algo-demo-trace})  and its frontier graph (\subref{fig:frontier-algo-demo-graph}), witnessing the consistency of $\tuple{\AbstractExecution, \cpFunc}$.
There is a path from source (dotted node) to sink (dashed node), 
and the events labelling this path form a valid concretization, 
i.e., $\trace = \snd_1 \cdot \snd_2 \cdot \rcv_3 \cdot \rcv_4$. 
Therefore, 
$\tuple{\AbstractExecution, \cpFunc}$ is consistent. 
}
\figlabel{frontier-algo-demo}
\end{figure}

\myparagraph{Time complexity}{
Given  \lemref{consistency-to-reachability}, we can solve $\vch$ by constructing $\gfrontier$ and solving standard graph reachability on it.
Recall that each node is a tuple $\tuple{{\fgEventSet}, {\fgChanMap}, {\fgSyncSend}}$.
${\fgEventSet}$ is a $\po{}$-downward closed set, and there are at most 
$(\NumEvents^\NumThreads/\NumThreads^\NumThreads)$ many distinct subsets of $\eventSet$ of this form (using AM-GM inequality).
For each fixed ${\fgEventSet}$, the number of different possible ${\fgSyncSend}$ is upper bounded by $\NumThreads+1$, since ${\fgSyncSend}$ is either $\bot$ or points to the last event of a thread in ${\fgEventSet}$. 
Finally, consider the component ${\fgChanMap}$. 
For any asynchronous channel $\ch$, 
the number of messages in ${\fgChanMap}(\ch)$ is 
$\ell = \s{num}_{\fgEventSet}(\snd(\ch)) - \s{num}_{\fgEventSet}(\rcv(\ch))$. 
\revision{
A naive counting approach would enumerate the last $\ell$ sends to channel $\ch$ in every thread, 
forming a potential set of unreceived sends of size $\NumThreads \ell$. 
Selecting $\ell$ elements from this set yields a total of $O((\NumThreads \MaxCapacity)^{\MaxCapacity})$ possible combinations, as $\ell \le \MaxCapacity$. 
However, our counting method is more refined.}
We notice that the sequence of events in the queue ${\fgChanMap}(\ch)$
is uniquely determined by the sequence 
$\thread_{i_0}, \thread_{i_1} \ldots, \thread_{i_{\ell-1}}$
of their thread identifiers --- if $\tau_{i_{\ell - j}} = \tau$, 
then the $j^\text{th}$ \emph{last} event $e$ in ${\fgChanMap}(\ch)$ 
belongs to thread $\tau$ and, further $e$ is the $m^\text{th}$ 
last send event on channel $\ch$ that thread $\tau$ performs, where
$m = |\setpred{z}{\tau_z = \tau, \ell - j \leq z}|$.
Since the number of threads is $\NumThreads$, 
the total number of possible sequences corresponding to ${\fgChanMap}(\ch)$ is thus 
$\leq \NumThreads^\ell \in O(\NumThreads^{\MaxCapacity})$. 
This implies that the total number of different values that ${\fgChanMap}$ can take on 
is in $O(\NumThreads^{\MaxCapacity\NumChannels})$.
Thus, the total number of nodes of $\gfrontier$ is
$O(\NumEvents^\NumThreads/\NumThreads^\NumThreads \cdot \NumThreads \cdot \NumThreads^{\MaxCapacity \NumChannels})$.

We now count the number of edges in $\gfrontier$. 
Each node has at most $\NumThreads$ out-degree since the set ${\fgEventSet}$ is $\po{}$ downward closed for each node.
Hence the number of edges in $\gfrontier$ is bounded by 
$(\NumEvents^\NumThreads/\NumThreads^\NumThreads \cdot \NumThreads^2 \cdot \NumThreads^{\MaxCapacity \NumChannels})$. 
Thus, the time for reachability checking is $O(|V|+|E|) = O(\NumEvents^\NumThreads \cdot \NumThreads^{\MaxCapacity \NumChannels})$. 

The graph can be constructed using a simple worklist algorithm. 
The worklist is initialized with only the source node. 
The algorithm proceeds by repeatedly extracting a node $v$ from the worklist and inserting its successors until the worklist is empty. 
To compute the successor node $v'$ of the current node $v$ by extending $v$ with event $e$, 
we first copy $v$ into $v'$ and update $v'$ according to the rules of the frontier graph.  
Copying $v$ takes $O(\NumEvents)$ time 
and updating $v'$ takes constant time, and thus time to construct the graph
is $O(\NumEvents \cdot V + E) = O(\NumEvents^{\NumThreads+1} \cdot \NumThreads^{\MaxCapacity \NumChannels})$ 
as in \thmref{vch-solution}.
}

The algorithm for $\vchRf$ has similar flavor to that for $\vch$, but relies on a different frontier graph.


\myparagraph{Frontier graph for $\vchRf$}{
The reads-from frontier graph $\gfrontierRf$ of $\tuple{\AbstractExecution, \cpFunc, \rf{}}$ is slightly different from $\gfrontier$. 
First, for a node $v = \tuple{{\fgEventSet}, {\fgChanMap}, {\fgSyncSend}}$, 
the set of unmatched send events buffered in ${\fgChanMap}(\ch)$ and ${\fgSyncSend}$ is already determined by ${\fgEventSet}$ and $\rf{}$. 
Therefore, we only need to consider the permutations of these events in ${\fgChanMap}(\ch)$. 
Moreover, for an edge $v_1 \xrightarrow{e} v_2$ labeled with a receive event $e = \rcv(\ch)$ over an asynchronous (resp. synchronous) channel $\ch$, we require that the first entry $f = v_1.{\fgChanMap}(\ch)$ (resp. unique element $f = v_1.{\fgSyncSend}$) is such that $(e, f) \in \rf{}$.
The following lemma states how $\gfrontierRf$ captures the consistency of $\tuple{\AbstractExecution, \cpFunc, \rf{}}$.
}

\begin{restatable}{lemma}{vchRfConsToReach}
\lemlabel{consistency-to-reachability-vchRf}
$\tuple{\AbstractExecution, \cpFunc, \rf{}}$ is consistent iff there is a sink node reachable from the source node in $\gfrontierRf$.
\end{restatable}

\myparagraph{Time complexity for $\vchRf$}{
For each node, the set ${\fgEventSet}$, together with $\rf{}$, uniquely determine send events that are unmatched, giving us a better bound on the number of possible values for the ${\fgChanMap}$ and ${\fgSyncSend}$ components of the node. 
The number of distinct ${\fgEventSet}$ sets is still $(\NumEvents^\NumThreads/\NumThreads^\NumThreads)$. 
For each ${\fgEventSet}$, ${\fgSyncSend}$ is uniquely determined by ${\fgEventSet}$ and $\rf{}$. 
Likewise, the set of events in ${\fgChanMap}(\ch)$ for an asynchronous channel is the set of unmatched send events in ${\fgEventSet}$, whose size is bounded by $\cp{\ch} \leq \MaxCapacity$. 
The total number of permutations for ${\fgChanMap}(\ch)$ is thus $\leq \cp{\ch}! \le \MaxCapacity!$
and is also $\le \NumThreads^\MaxCapacity$ as argued before. 
Considering all $\NumChannels$ channels, the number of ${\fgChanMap}$ is bounded by $O((\min(\MaxCapacity!, \NumThreads^\MaxCapacity)^\NumChannels)$. 
{\fgSyncSend}n total, the number of nodes in the graph is 
$O(\NumEvents^\NumThreads/\NumThreads^\NumThreads \cdot (\min(\MaxCapacity!, \NumThreads^\MaxCapacity)^\NumChannels)$,
while the number of edges is
$O(\NumEvents^\NumThreads/\NumThreads^\NumThreads \cdot \NumThreads \cdot (\min(\MaxCapacity!, \NumThreads^\MaxCapacity)^\NumChannels)$,
thereby concluding \thmref{vch-rf-solution}.
}
\subsection{$\vchRf$ with Synchronous Channels}
\seclabel{vch-rf-sync-solution-sec}
We now focus on $\vchRf$ in the case where all channels are synchronous and present a linear-time algorithm for \thmref{vch-rf-sync-upper}.
Previous work shows that purely synchronous communication  enjoys some sort of ``deterministic replay'', which implies a $O(\NumEvents \cdot \NumThreads)$-time consistency algorithm~\cite{sulzmann2018two}.
Here we show that this setting admits a linear-time solution, irrespectively of the number of threads.

Our algorithm is based on the following insight.
Since all channels are synchronous, every pair of events $(\snd, \rcv)$ related by reads-from must execute consecutively.
Our algorithm packs such event pairs in a single atomic event,
and checks whether all atomic events can be scheduled in a way that respects partial order dependencies due to $\po{}$.
In turn, this reduces to checking for cycles in a suitably defined graph. 
We now make the above insight formal. 

We assume wlog that the input instance $\tuple{\AbstractExecution, \cpFunc, \rf{}}$, where $\AbstractExecution = \tuple{\eventSet, \po{}}$, is such that each send (resp. receive) event has exactly one receive 
(resp. send) event matched to it using $\rf{}$, and the two events belong to different threads. 
Otherwise, the instance is clearly inconsistent. 

\myparagraph{The send-receive graph}{
The \emph{send-receive graph} of $\tuple{\AbstractExecution, \cpFunc, \rf{}}$ is a directed graph $\gsync = (V, E)$ where $V$ is the node set and $E$ is the edge set, defined as follows.
\begin{enumerate*}
\item $V \subseteq \eventSet \times \eventSet$ is the set of matching send and receive pairs, i.e., $\tuple{\snd, \rcv} \in V$ iff $(\snd, \rcv)\in\rf{}$
\item edges $E$ capture $\po{}$ dependencies, i.e., $(\tuple{\snd_1, \rcv_1}, \tuple{\snd_2, \rcv_2}) \in E$
iff some $e_1\in \{\snd_1, \rcv_1\}$ is the immediate $\po{}$ predecessor of some $e_2\in \{\snd_2, \rcv_2\}$.
\end{enumerate*}
See \figref{vsc-sync-graph} for an illustration.
The send-receive graph precisely captures consistency, as stated in the following lemma.
}

\begin{restatable}{lemma}{vchRfSyncConsToCyclic}
\lemlabel{graph-cycle-sync-consistency}
$\tuple{\AbstractExecution, \cpFunc, \rf{}}$ is consistent iff $\gsync$ is acyclic.
\end{restatable}


\begin{figure}[htbp]
\centering
\begin{subfigure}[b]{0.63\textwidth}
\centering
\newcommand{\xstep}{2.4}
\newcommand{\ystep}{-0.8}
\newcommand{\height}{0.5}
\newcommand{\wid}{1.7}

\scalebox{0.8}{
\begin{tikzpicture}[font=\small]
\node [event] (t1-1) at (0 * \xstep, 0 * \ystep) {$\snd_1(\ch_1)$};
\node [event] (t1-2) at (0 * \xstep, 1 * \ystep) {$\rcv_3(\ch_1)$};
\node [event] (t1-3) at (0 * \xstep, 2 * \ystep) {$\snd_4(\ch_2)$};

\node [event] (t2-1) at (1 * \xstep, 0 * \ystep) {$\rcv_1(\ch_1)$};
\node [event] (t2-2) at (1 * \xstep, 1 * \ystep) {$\rcv_4(\ch_2)$};
\node [event] (t2-3) at (1 * \xstep, 2 * \ystep) {$\snd_2(\ch_2)$};

\node [event] (t3-1) at (2 * \xstep, 1 * \ystep)  {$\snd_3(\ch_1)$};
\node [event] (t3-2) at (2 * \xstep, 2 * \ystep)  {$\rcv_2(\ch_2)$};

\draw [rfEdge] (t1-1) to (t2-1);
\draw [rfEdge] (t2-3) to (t3-2);
\draw [rfEdge] (t1-3) to (t2-2);
\draw [rfEdge, bend right=13] (t3-1) to (t1-2);

\node [threadName] at (0 * \xstep, -1.0 * \ystep) {$\thread_{1}$};
\node [threadName] at (1 * \xstep, -1.0 * \ystep) {$\thread_{2}$};
\node [threadName] at (2 * \xstep, -1.0 * \ystep) {$\thread_{3}$};
    
\begin{scope}[on background layer]
\draw [thread] (0 * \xstep, -0.6 * \ystep) -- (0 * \xstep, 2.8 * \ystep);
\draw [thread] (1 * \xstep, -0.6 * \ystep) -- (1 * \xstep, 2.8 * \ystep);
\draw [thread] (2 * \xstep, -0.6 * \ystep) -- (2 * \xstep, 2.8 * \ystep);
\end{scope}
\end{tikzpicture}
}
\caption{$\vchRf$ instance $\tuple{\AbstractExecution, \cpFunc, \rf{}}$ with synchronous channels.}
\figlabel{vsc-sync-graph-1}
\end{subfigure}
\hfill
\begin{subfigure}[b]{0.35\textwidth}
\centering
\newcommand{\height}{0.6}
\newcommand{\wid}{1.8}
\newcommand{\xstep}{3}
\newcommand{\ystep}{2}
\scalebox{0.8}{
\begin{tikzpicture}[font=\small]
\node [nodeRectangle, fill=gray!10] (n4) at (0*\xstep, 0*\ystep) {$\tuple{\snd_4, \rcv_4}$};
\node [nodeRectangle, fill=gray!10] (n3) at (1*\xstep, 0*\ystep) {$\tuple{\snd_3, \rcv_3}$};
\node [nodeRectangle, fill=gray!10] (n1) at (0*\xstep, 1*\ystep) {$\tuple{\snd_1, \rcv_1}$};
\node [nodeRectangle, fill=gray!10] (n2) at (1*\xstep, 1*\ystep) {$\tuple{\snd_2, \rcv_2}$};
\draw [edge] (n1) to (n4);
\draw [edge] (n1) to (n3);
\draw [edge] (n3) to (n2);
\draw [edge] (n3) to (n4);
\draw [edge] (n4) to (n2);
\end{tikzpicture}
}
\caption{The graph $\gsync$.}
\figlabel{vsc-sync-graph-2}
\end{subfigure}
\caption{A $\vchRf$ instance $\tuple{\AbstractExecution, \cpFunc, \rf{}}$ (\subref{fig:vsc-sync-graph-1}) and the corresponding send-receive graph $\gsync$ (\subref{fig:vsc-sync-graph-2}). 
As $\gsync$ is acyclic, 
$\tuple{\AbstractExecution, \cpFunc, \rf{}}$ is consistent. 
}
\figlabel{vsc-sync-graph}
\end{figure}

\myparagraph{Algorithm and time complexity}{
Following \lemref{graph-cycle-sync-consistency}, the
algorithm for checking $\vchRf$ when all channels are synchronous is 
straightforward --- construct $\gsync$ and check for acyclicity. 
For each pair $\tuple{\snd, \rcv}$, 
there are at most two immediate $\po{}$ predecessors, 
so the in-degree of each node is at most 2. 
Therefore, $\gsync$ has $O(\NumEvents)$ nodes and $O(\NumEvents)$ edges and the time to construct the graph is also $O(\NumEvents)$.
Checking for a cycle in $\gsync$ also takes $O(\NumEvents)$ time, which concludes the proof of \thmref{vch-rf-sync-upper}.
}

\subsection{Acyclic Communication Topologies}
\seclabel{tree-topo-upper-bound}

Finally, we turn our attention to acyclic communication topologies and prove that $\vchRf$ can be solved in quadratic time, establishing \thmref{vch-rf-tree-topology}. 
We first formally define the communication topology of an abstract execution. 

\myparagraph{Communication topologies}
A set of events $\eventSet$ induces a communication topology,
represented as an undirected graph $G=(V,E)$ where $V$ is the set of threads appearing in $\eventSet$,
and we have $(\thread_i, \thread_j)\in E$ iff $\thread_i$ and $\thread_j$ access a common channel, i.e.,
there exist two events $e_1, e_2\in \eventSet$ such that $\ThreadOf{e_1}=\thread_i$, $\ThreadOf{e_2}=\thread_j$, and $\ChOf{e_1}=\ChOf{e_2}$.
The communication topology induced by an abstract execution $\AbstractExecution=\tuple{\eventSet, \po{}}$  is the topology induced by its event set $\eventSet$.

Given two threads $\thread_i$ and $\thread_j$, let $\proj{\channels{\AbstractExecution}}{\thread_i, \thread_j}$ be the set of channels accessed by both $\thread_i, \thread_j$,
and $\proj{\cpFunc}{\thread_i,\thread_j}$ be the restriction of the capacity function $\cpFunc$ to the channels in $\proj{\channels{\AbstractExecution}}{\thread_i, \thread_j}$.
We define $\proj{\AbstractExecution}{\thread_i, \thread_j}$ and $\proj{\rf{}}{\thread_i, \thread_j}$ as the abstract execution obtained from $\AbstractExecution$ and reads-from relation obtained from $\rf{}$ by only keeping events from $\thread_i, \thread_j$ that access a channel in $\proj{\channels{\AbstractExecution}}{\thread_i, \thread_j}$. 
Our proof of \thmref{vch-rf-tree-topology} is based on two key insights.
First, we prove that $\vchRf$ on acyclic topologies is \emph{compositional}:~$\tuple{\AbstractExecution, \cpFunc, \rf{}}$ is consistent iff $\tuple{\proj{\AbstractExecution}{\thread_i, \thread_j}, \proj{\cpFunc}{\thread_i,\thread_j}, \proj{\rf{}}{\thread_i, \thread_j}}$ is consistent, for every $(\thread_i, \thread_j)\in E$.
Second, we show that $\vchRf$ over two threads is solvable in quadratic time, by a reduction to 2SAT on formulas of size quadratic in the size of the input.

\myparagraph{Comparison with 2SAT encodings for shared-memory setting}
\revision{The 2SAT encodings have also been explored in the shared-memory setting~\cite{Chalupa2018}, but our work introduces several novel aspects.
First, our encoding captures constraints unique to FIFO channel semantics and bounded capacities (including synchronous and capacity-1 channels), 
while naturally extending to unbounded channels without extra constraints.
Second, our acyclic-topology result relies on a new compositionality lemma (\lemref{acyclic-topology-compositionality}), showing that any efficient solution for two threads extends compositionally to arbitrary acyclic systems, enabling direct performance gains from improved two-thread algorithms.
}

\myparagraph{Compositionality}{
The compositionality lemma is formally stated as follows.

\begin{restatable}{lemma}{acyclictopologycompositionality}
\lemlabel{acyclic-topology-compositionality}
Let $\tuple{\AbstractExecution, \cpFunc, \rf{}}$ be a $\vchRf$ instance, and $G=(V,E)$ the communication topology of $\AbstractExecution$ such that $G$ is acyclic.
Then $\tuple{\AbstractExecution, \cpFunc, \rf{}}$ is consistent iff $\tuple{\proj{\AbstractExecution}{\thread_i,\thread_j}, , \proj{\cpFunc}{\thread_i,\thread_j}, 
\proj{\rf{}}{\thread_i, \thread_j}}$ is consistent, for every pair of threads $(\thread_i, \thread_j)\in E$.
\end{restatable}

The intuition behind \lemref{acyclic-topology-compositionality} is as follows.
First, clearly for $\tuple{\AbstractExecution, \cpFunc, \rf{}}$ to be consistent, we must have that $\tuple{\proj{\AbstractExecution}{\thread_i,\thread_j}, \proj{\cpFunc}{\thread_i,\thread_j}, \proj{\rf{}}{\thread_i, \thread_j}}$ is consistent for every two threads $\thread_i, \thread_j$.
The other direction is more interesting.
Consider a thread $\thread_1$ with two neighbors $\thread_2, \thread_3$ in the communication topology, $(\thread_1, \thread_2), (\thread_1, \thread_3)\in E$, such that $\tuple{\proj{\AbstractExecution}{\thread_1,\thread_2}, \proj{\cpFunc}{\thread_1,\thread_2}, \proj{\rf{}}{\thread_1, \thread_2}}$ and $\tuple{\proj{\AbstractExecution}{\thread_1,\thread_3}, \proj{\cpFunc}{\thread_1,\thread_3}, \proj{\rf{}}{\thread_1, \thread_3}}$ are consistent, witnessed by the corresponding executions $\trace_{1,2}$ and $\trace_{1,3}$.
Then we can interleave $\trace_{1,2}$ and $\trace_{1,3}$ in any way that respects the program order of thread $\thread_1$, and the resulting execution $\trace_1$ will be well-formed.
This is because, owning to the acyclicity of $G$, we have $(\thread_2, \thread_3)\not \in E$, meaning that $\thread_2$ and $\thread_3$ do not communicate over a common channel.
In turn, this implies that the interleaving of events from $\thread_2$ and $\thread_3$ in $\trace$ cannot violate the well-formedness of $\trace_1$.
Composing all executions along edges of $G$ in such a way results in an execution $\trace$ that witnesses the consistency of  $\tuple{\AbstractExecution, \cpFunc, \rf{}}$.
}

\myparagraph{The case of $\NumThreads=2$ threads}{
Given \lemref{acyclic-topology-compositionality}, we now focus on the case of $\vchRf$ over $2$ threads,
when every channel is capacity-unbounded, has capacity $1$, or is synchronous (i.e., the setting captured in \thmref{vch-rf-tree-topology}).
We obtain a quadratic bound based on two insights.
First, for each channel, channel-related constraints on the order of events accessing it can be encoded as 2SAT.
The search for well-formed execution must also satisfy transitivity constraints, i.e., if $e_1\to e_2$ and $e_2\to e_3$, then $e_1\to e_3$.
Transitivity involves three events, and thus does not immediately fit our 2SAT approach.
Our second observation is that, with $2$ threads, every three events $e_1, e_2, e_3$, must contain two events in the same thread, thus already ordered by $\po{}$.
Then, transitivity can be succinctly captured by a 2SAT formula as well.
In the following we make these insights formal.

Consider a $\vchRf$ instance $\tuple{\AbstractExecution, \cpFunc, \rf{}}$ where ${\AbstractExecution} = \tuple{\eventSet, \po{}}$ is an abstract execution involving two threads $\thread_1, \thread_2$.
We construct a 2SAT formula $\form_{\tuple{\AbstractExecution, \cpFunc, \rf{}}}$ over propositional variables $x_{e,f}$, where $e, f\in \eventSet$.
Assigning $x_{e,f}=\top$ means ordering $e$ before $f$ in the execution witnessing the consistency of $\tuple{\AbstractExecution, \cpFunc, \rf{}}$.
Overall, $\form_{\tuple{\AbstractExecution, \cpFunc, \rf{}}}$ is a conjunction of 8 subformulae:
\begin{align*}
\form_{\tuple{\AbstractExecution, \cpFunc, \rf{}}} \equiv \form_{\exactlyOne} 
\land \form_{\po{}} 
\land \form_{\rf{}}  \land \form_{\unmatched} \land \form_{\fifo} \land \form_{\trans} \land  \form_{\capOne} \land \form_{\sync}
\end{align*}
We now proceed with defining each subformula.
}

\mysubparagraph{Exactly one}{
This formula requires that the order of two events must be resolved exactly in one way.
\begin{align*}
\form_{\exactlyOne} \equiv  \bigwedge\limits_{
e, f \in \eventSet
} \left( x_{e, f} \implies \neg x_{f, e}\right)
\end{align*}
}
\vspace{-0.1in}

\mysubparagraph{Program order}{
This formula requires that the order of two events must respect $\po{}$. 
}

\mysubparagraph{Reads from}{
This formula requires that each receive event is ordered after its matched send event.
\begin{align*}
\form_{\po{}} &\equiv \bigwedge\limits_{(e, f) \in \po{}} x_{e, f}
\quad\text{and}\quad
\form_{\rf{}} \equiv \bigwedge\limits_{(e, f) \in \rf{}} x_{e, f}
\end{align*}
}
\vspace{-0.1in}



\mysubparagraph{Unmatched sends}{
This formula requires that all unmatched send events are scheduled after all send events that have a matching receive event.
Given a channel $\ch$, let 
\begin{align*}
\textsf{Unmatched}_{\ch} &= \setpred{e \in \eventSet}{\OpOf{e} = \snd, \ChOf{e} = \ch, \nexists f \text{ s.t. } (e, f) \in \rf{}} \text{, and}\\
\textsf{Matched}_{\ch} &= \setpred{e \in \eventSet}{\OpOf{e} = \snd, \ChOf{e} = \ch, \exists f \text{ s.t. } (e, f) \in \rf{}}
\end{align*}
denote the set of unmatched and matched send events, respectively.
We have
\begin{align*}
\form_{\unmatched} \equiv
\bigwedge\limits_{
\scriptsize
\begin{aligned}
\begin{array}{c}
\ch \in \channels{\AbstractExecution},
e \in \textsf{Matched}_{\ch}, \\
f \in \textsf{Unmatched}_{\ch} 
\end{array}
\end{aligned}
}   
x_{e, f}
\end{align*}
}
\vspace{-0.1in}

\mysubparagraph{FIFO}{
This formula requires that the order of two receive events on the same channel matches the order of the corresponding send events.
\begin{align*}
\form_{\fifo} \equiv \bigwedge\limits_{
\scriptsize
\begin{aligned}
\begin{array}{c}
(e,e')\in \rf{}, (f, f') \in \rf{} \\
e \neq f, \ChOf{e} = \ChOf{f} \\ 
\end{array}
\end{aligned}
}
\left(\left(x_{e,f} \implies x_{e', f'}\right) \wedge \left(x_{e', f'} \implies x_{e,f}\right)\right)
\end{align*}
}
\vspace{-0.1in}

\mysubparagraph{Transitivity}{
This formula requires that the ordering of events is transitive.
Let $\pred{}{e}$ (resp. $\sucr{}{e}$) be the unique event (if one exists) that precedes (resp. succeeds) $e$ in $\po{}$. 
If $\pred{}{e}$ (resp. $\sucr{}{e}$) doesn't exist, 
then $\pred{}{e} = \bot$ (resp. $\sucr{}{e} = \bot$) .
We have $\form_{\trans} \equiv \form^{\s{pred}}_{\trans} \land \form^{\s{succ}}_{\trans}$, where 
\begin{align*}
\begin{array}{lcr}
\form^{\s{pred}}_{\trans} 
\equiv 
\bigwedge\limits_{
e, f \; \in \; \eventSet, \;
e' = \pred{}{e} \neq \bot
} 
(x_{e, f} \implies x_{e', f})
&
\qquad
&
\form^{\s{succ}}_{\trans} 
\equiv 
\bigwedge\limits_{
e, f \; \in \; \eventSet, \;
f' = \sucr{}{f} \neq \bot
} 
\left(x_{e, f} \implies x_{e, f'}\right)
\end{array}
\end{align*}
}
\vspace{-0.1in}

\mysubparagraph{Capacity}{
This formula requires that the capacity constraints of channels $\ch$ with $\cpFunc(\ch)\leq 1$ are met.
In particular, for two different send events $\snd_1(\ch) \neq \snd_2(\ch)$, the matching receive event of the earlier send event also precedes the other send event. 
For a synchronous channel, we encode the fact that send and receive events are consecutive.
For asynchronous channels that are capacity-unbounded, we do not need any capacity constraint.
\begin{align*}
\begin{array}{lcr}
\form_{\capOne} 
&
\equiv
&
\bigwedge\limits_{
\scriptsize
\begin{aligned}
\begin{array}{c}
\left(e, f\right) \in \rf{}, 
e' \in \eventSet, 
\OpOf{e} = \OpOf{e'} = \snd \\
\ChOf{e} = \ChOf{e'} \text{ is asynchronous} 
\end{array}
\end{aligned}
} \left(x_{e,e'} \implies x_{f, e'}\right)
\\
\form_{\sync} 
&
\equiv 
&
\bigwedge\limits_{
\scriptsize
\begin{aligned}
\begin{array}{c}
\left(e, f\right) \in \rf{},
\ChOf{e} \text{ is synchronous} \\
e' = \sucr{}{e}, f' = \pred{}{f}
\end{array}
\end{aligned}
} \left(x_{f,e'} \land x_{f', e}\right)
\end{array}
\end{align*}
}

The following lemma states the correctness of the encoding.

\begin{restatable}{lemma}{twosatcorrectness}
\lemlabel{twosatcorrectness}
$\tuple{\AbstractExecution, \cpFunc, \rf{}}$ is consistent iff $\form_{\tuple{\AbstractExecution, \cpFunc, \rf{}}}$ is satisfiable.
\end{restatable}

Finally, observe that the number of propositional variables $x_{e,f}$ is bounded by $\NumEvents^2$,
while the number of clauses is also $O(\NumEvents^2)$.
Since 2SAT is solvable in time that is linear in the size of the formula~\cite{aspvall1979linear}, together with \lemref{twosatcorrectness}, we arrive at an algorithm that solves $\vchRf$ for $2$ threads in $O(\NumEvents^2)$ time.

\myparagraph{Acyclic topologies}{
We now have all the ingredients to solve $\vchRf$ on acyclic communication topologies.
Given an input $\tuple{\AbstractExecution, \cpFunc, \rf{}}$, the algorithm iterates over all edges $(\thread_i, \thread_j)$  of the communication topology of $\AbstractExecution$, and uses the 2SAT encoding to decide the consistency of $\tuple{\proj{\AbstractExecution}{\thread_i, \thread_j}, \proj{\cpFunc}{\thread_i,\thread_j}, \proj{\rf{}}{\thread_i, \thread_j}}$.

For analyzing the time complexity, observe that every two events $e, f\in \eventSet$ appear in some propositional variable $x_{e,f}$ of at most one 2SAT instance.
In particular, let $\thread_1=\ThreadOf{e}$ and $\thread_2 = \ThreadOf{f}$.
If $\thread_1\neq \thread_2$, then $x_{e,f}$ appears in the 2SAT instance of the topology edge $(\thread_1, \thread_2)$.
On the other hand, if $\thread_1=\thread_2=\thread$, then $x_{e,f}$ appears in the 2SAT instance of the topology edge $(\thread, \thread')$, where $\thread'$ is the unique thread accessing the channels that $e$ and $f$ operate.
We thus arrive at \thmref{vch-rf-tree-topology}.
}

\section{The Hardness of Verifying Channel Consistency with Reads From}
\seclabel{lower-bounds-vchRf}
We now present some of the hardness results for $\vchRf$. 
\revision{
We first present how channels can be used to encode \emph{atomicity gadgets} -- that is, 
how to execute a sequence of events from the same thread without interleaving with other threads, 
as this encoding will be required in later reductions.
}
We then show that the problem is intractable for case (i) and (iii) stated in \thmref{vch-rf-hardness} in \secref{vch-rf-np-hard-k=1} and \secref{hardness-const-t-m}). 
In \secref{two-threads-lower-bound-sec}, 
we prove the quadratic lower bound of $\vchRf$ on $2$ threads, as stated in \thmref{vch-rf-two-threads-lower-bound}. 
The other lower bounds of $\vch$ and $\vchRf$ stated in \thmref{vch-same-value-hardness}, 
\thmref{vch-two-threads-hardness}, 
\thmref{vch-one-chan-hardness}, 
\thmref{vch-rf-hardness}, 
are proven with reductions of similar flavor, and appear in \appref{lower-bounds-vch} 
and \appref{lower-bounds-vch-rf} 
due to space limits. 
\subsection{Atomicity Gadgets}
\seclabel{atomic-gadget}
\revision{
Our reductions in later sections make use of \emph{atomic blocks} of events as gadgets.
An atomic block $\atom$ in a thread is a sequence of events such that any two such blocks $\atom_{\textcolor{cyan!80!black} 1}, \atom_{\textcolor{cyan!80!black} 2}$ cannot overlap in any concretization.
Here we show how to construct atomicity gadgets, both by using channels with capacity $1$, and by using channels with unbounded capacity.
The latter might sound counter-intuitive, in the sense that send operations to capacity-unbounded channels never block. 
}


\begin{figure}[h]
\centering
\begin{subfigure}[b]{0.45\textwidth}
\centering
\newcommand{\xstep}{2.8}
\newcommand{\ystep}{-0.7}
\newcommand{\height}{0.5}
\newcommand{\wid}{1.8}
\newcommand{\outOne}{195}
\newcommand{\inOne}{165}
\newcommand{\outTwo}{-15}
\newcommand{\inTwo}{15}
\newcommand{\Looseness}{0.7}

\scalebox{0.8}{
\begin{tikzpicture}
\node[event] (t1-1) at (0 * \xstep, 0 * \ystep) {$\wt_1({\lk})$};
\node[event] (t1-2) at (0 * \xstep, 1 * \ystep) {$\atom_{\textcolor{cyan!80!black} 1}$};
\node[event] (t1-3) at (0 * \xstep, 2 * \ystep) {$\rd_1(\lk)$};

\node[event] (t2-1) at (1 * \xstep, 0 * \ystep) {$\wt_2( {\lk})$};
\node[event] (t2-2) at (1 * \xstep, 1 * \ystep) {$\atom_{\textcolor{cyan!80!black} 2}$};
\node[event] (t2-3) at (1 * \xstep, 2 * \ystep) {$\rd_2( {\lk})$};

\node[threadName] at (0 * \xstep, -1.0 * \ystep) {$\thread_{1}$};
\node[threadName] at (1 * \xstep, -1.0 * \ystep) {$\thread_{2}$};

\draw[rfEdge, out=\outOne, in=\inOne, looseness=\Looseness] (t1-1) to (t1-3);
\draw[rfEdge, out=\outTwo, in=\inTwo, looseness=\Looseness] (t2-1) to (t2-3);

\begin{scope}[on background layer]
\draw [thread] (0 * \xstep, -0.6 * \ystep) -- (0 * \xstep, 2.7 * \ystep);
\draw [thread] (1 * \xstep, -0.6 * \ystep) -- (1 * \xstep, 2.7 * \ystep);
\end{scope}
\end{tikzpicture}
}
\vspace{0.6cm}
\caption{Atomicity using capacity $1$ channel.}
\figlabel{lock-demo-capacity-one}
\end{subfigure}
\hfill
\begin{subfigure}[b]{0.54\textwidth}
\centering
\newcommand{\xstep}{2.8}
\newcommand{\ystep}{-0.7}
\newcommand{\height}{0.5}
\newcommand{\wid}{1.8}
\newcommand{\outOne}{195}
\newcommand{\inOne}{165}
\newcommand{\outTwo}{-15}
\newcommand{\inTwo}{15}
\newcommand{\Looseness}{0.7}

\newcommand{\outOneInner}{185}
\newcommand{\inOneInner}{175}
\newcommand{\outTwoInner}{-5}
\newcommand{\inTwoInner}{5}
\newcommand{\LoosenessInner}{0.9}

\scalebox{0.8}{
\begin{tikzpicture}
\node[event] (t1-1) at (0 * \xstep, 0 * \ystep)  {$\wt_1( {\lk_1})$};
\node[event] (t1-2) at (0 * \xstep, 1 * \ystep) {$\wt_2( {\lk_2})$};
\node[event] (t1-3) at (0 * \xstep, 2 * \ystep) {$\rd_2( {\lk_2})$};
\node[event] (t1-4) at (0 * \xstep, 3 * \ystep) {$\atom_{\textcolor{cyan!80!black} 1}$};
\node[event] (t1-5) at (0 * \xstep, 4 * \ystep) {$\rd_1( {\lk_1})$};

\node[event] (t2-1) at (1 * \xstep, 0 * \ystep) {$\wt_3( {\lk_2})$};
\node[event] (t2-2) at (1 * \xstep, 1 * \ystep) {$\wt_4( {\lk_1})$};
\node[event] (t2-3) at (1 * \xstep, 2 * \ystep) {$\rd_4( {\lk_1})$};
\node[event] (t2-4) at (1 * \xstep, 3 * \ystep) {$\atom_{\textcolor{cyan!80!black} 2}$};
\node[event] (t2-5) at (1 * \xstep, 4 * \ystep) {$\rd_3( {\lk_2})$};

\node[threadName] at (0 * \xstep, -1.0 * \ystep) {$\thread_{1}$};
\node[threadName] at (1 * \xstep, -1.0 * \ystep) {$\thread_{2}$};

\draw[rfEdge, out=\outOne, in=\inOne, looseness=\Looseness] (t1-1) to (t1-5);
\draw[rfEdge, out=\outTwo, in=\inTwo, looseness=\Looseness] (t2-1) to (t2-5);

\draw[rfEdge, out=\outOneInner, in=\inOneInner, looseness=\LoosenessInner] (t1-2) to (t1-3);
\draw[rfEdge, out=\outTwoInner, in=\inTwoInner, looseness=\LoosenessInner] (t2-2) to (t2-3);

\begin{scope}[on background layer]
\draw [thread] (0 * \xstep, -0.6 * \ystep) -- (0 * \xstep, 4.7 * \ystep);
\draw [thread] (1 * \xstep, -0.6 * \ystep) -- (1 * \xstep, 4.7 * \ystep);
\end{scope}
\end{tikzpicture}
}
\caption{Atomicity using two capacity-unbounded channels $\lk_1, \lk_2$}
\figlabel{lock-demo-unbounded}
\end{subfigure}
\caption{Gadgets for implementing atomic blocks using capacity $1$ (\subref{fig:lock-demo-capacity-one}) or unbounded-capacity channels (\subref{fig:lock-demo-unbounded}).
Reads-from edges are represented by arrows.}
\figlabel{atomicity-gadgets}
\end{figure}

\myparagraph{Atomicity with capacity 1}{
\revision{
The atomicity gadget relying on channels of capacity~1 is shown in \figref{lock-demo-capacity-one}, using one channel $\lk$, 
which resembles a lock.
The thread that sends to $\lk$ first fills the channel capacity, and must execute the corresponding receive before the other thread can send to the channel.
The events between the first send and receive are thus executed atomically.
}
}

\myparagraph{Atomicity with unbounded capacity}{
\revision{
The atomicity gadget using unbounded channels is shown in \figref{lock-demo-unbounded}, relying on two channels $\lk_1$ and $\lk_2$.
Its principle of operation is as follows.
If $\snd_1(\lk_1)$ is executed before $\snd_4(\lk_1)$, then
$\rcv_1(\lk_1)$ is also executed before $\rcv_4(\lk_1)$, making the atomic section of the first thread execute before the second.
The inverse order is imposed if $\snd_4(\lk_1)$ is executed before $\snd_1(\lk_1)$,
as this orders $\snd_3(\lk_2)$ before $\snd_2(\lk_2)$, and the argument repeats.
}
}

\revision{
We note that the atomicity gadgets can be generalized to an arbitrary number of threads. 
For brevity, \figref{atomicity-gadgets} illustrates the case for two threads only.
}

\subsection{Hardness with Asynchronous Channels of Capacity $1$}
\seclabel{vch-rf-np-hard-k=1}

We establish a reduction from the $\vscRd$ problem~\cite{gibbons1997testing}.
An instance of the $\vscRd$ problem is a tuple ${\AbstractExecution} = \tuple{\eventSet, \po{}, \rf{}}$, where $\eventSet$ is a set of events of the form $\ev{\thread, \rdMem(x)}$
or $\ev{\thread, \wtMem(x)}$, 
in which $\thread$ is a thread identifier and $x$ is 
a memory location, $\po{}$ is the per-thread total order (a.k.a program order)
and $\rf{}$ maps each read event to a write event of the same register.
Such an instance is sequentially consistent (SC) if there is a total order
over $\eventSet$ that respects $\po{}$ and $\rf{}$, and ensures that
for every $(e, f) \in \rf{}$ pair on register $x$, there is no other $\wtMem(x)$
event ordered between $e$ and $f$.

\myparagraph{Overview}{
Let $\AbstractExecution = \tuple{\eventSet, \po{}, \rf{}}$ be an instance of $\vscRd$.
We construct an instance $\tuple{\AbstractExecution', \cpFunc', \rf{}'}$ of $\vchRf$, 
where $\AbstractExecution' = \tuple{\eventSet', \po{}'}$.
At a high level, each write event (and each read event) in $\AbstractExecution$ 
is mapped to a sequence of send and receive instructions in $\AbstractExecution'$ that essentially appear atomically in every concretization.
Further the reads-from relation of $\AbstractExecution$ is also accurately
reflected in $\AbstractExecution'$ through reads-from on channels.
}

\myparagraph{Reduction}{}
\figref{np-hard-vsc-read} illustrates the reduction on a small example.
The set of threads in $\AbstractExecution'$ 
is the same as ${\AbstractExecution}$.
The set of channels used in ${\AbstractExecution}'$
is $\setpred{\ch^i_x}{x \in \mathcal{R}, 1 \leq i \leq m_x} \uplus \set{\lk}$,
where $\mathcal{R}$ is the set of registers accessed in ${\AbstractExecution}$,
$m_x = \max \setpred{p_e}{e \text{ is a write on }x}$
and $p_e$ is the number of read events $f$ with $(e, f) \in \rf{}$. 
The capacity function $\cpFunc$ assigns capacity $1$ to every channel. 
At a high level, the thread-wise event sequences in ${\AbstractExecution}'$
are structurally similar to those in ${\AbstractExecution}$,
and can be characterized using a map $M$ that maps
events in $\eventSet$ to distinct atomic, thread-local sequences of events in $\eventSet'$,
so that $\eventSet' = \bigcup_{e \in \eventSet} \setpred{f}{f \in M(e)}$. 
Atomicity is guaranteed by channel $\lk$ with capacity 1. 
In \secref{atomic-gadget}, we have detailed an explanation about atomicity gadgets. 
We now describe the map $M$.


\begin{figure}[t]
\centering
\begin{subfigure}[b]{0.28\textwidth}
\centering
\newcommand{\xstep}{2.5}
\newcommand{\ystep}{-0.8}
\newcommand{\height}{0.5}
\newcommand{\wid}{1.2}
\scalebox{0.8}{
\begin{tikzpicture}[font=\small]
\node[event] (t1-1) at (0 * \xstep, 0.3 * \ystep) {$\wtMem_1(x)$};
\node[event] (t1-2) at (0 * \xstep, 1.7 * \ystep) {$\rdMem_3(y)$};

\node[event] (t2-1) at (1 * \xstep, 0 * \ystep)  {$\rdMem_1(x)$};
\node[event] (t2-2) at (1 * \xstep, 1 * \ystep)  {$\rdMem_2(x)$};
\node[event] (t2-3) at (1 * \xstep, 2 * \ystep)  {$\wtMem_2(y)$};

\node[threadName] at (0 * \xstep, -1.0 * \ystep) {$\thread_{1}$};
\node[threadName] at (1 * \xstep, -1.0 * \ystep) {$\thread_{2}$};

\draw [rfEdge] (t1-1) to (t2-1);
\draw [rfEdge] (t1-1) to (t2-2);
\draw [rfEdge] (t2-3) to (t1-2);
    
\begin{scope}[on background layer]
\draw [thread] (0 * \xstep, -0.6 * \ystep) -- (0 * \xstep, 2.7 * \ystep);
\draw [thread] (1 * \xstep, -0.6 * \ystep) -- (1 * \xstep, 2.7 * \ystep);
\end{scope}
\end{tikzpicture}
}
\caption{A $\vscRd$ instance.}
\figlabel{np-hard-vsc-read-bounded-demo}
\end{subfigure}
\hfill
\begin{subfigure}[b]{0.7\textwidth}
\centering
\newcommand{\xstep}{3.2}
\newcommand{\ystep}{-0.8}
\newcommand{\height}{0.5}
\newcommand{\wid}{1.5}
\newcommand{\memDistance}{2}
\newcommand{\outOne}{195}
\newcommand{\inOne}{165}
\newcommand{\outTwo}{-15}
\newcommand{\inTwo}{15}
\newcommand{\Looseness}{0.7}
\newcommand{\bracketWid}{0.6}
\newcommand{\bracketHeight}{0.25}
\tikzstyle{memEvent}=[draw=none]
\scalebox{0.7}{
\begin{tikzpicture}[font=\small]
\node[event] (t1-1) at (0 * \xstep, 0 * \ystep)  {$\wt({\color{red} \lk})$};
\node[event] (t1-2) at (0 * \xstep, 1 * \ystep) {$\wt(\ch_x^1)$};
\node[event] (t1-3) at (0 * \xstep, 2 * \ystep) {$\wt(\ch_x^2)$};
\node[event] (t1-4) at (0 * \xstep, 3 * \ystep) {$\rd({\color{red} \lk})$};
\draw [memEvent] ($ (t1-1) + (-\memDistance, 0) $)  to node [midway,fill=white] {$\wtMem_1(x)$} ($ (t1-4) + (-\memDistance, 0) $);
\draw[line width=1.5pt] ($(t1-1) + (-\memDistance - 0.5 * \bracketWid, 0.5*\height - \bracketHeight)$) rectangle ($(t1-1) + (-\memDistance + 0.5 * \bracketWid, 0.5*\height)$);    
\draw[white, line width=3pt] ($(t1-1) + (-\memDistance - 0.6 * \bracketWid, 0.5*\height - \bracketHeight)$) -- ($(t1-1) + (-\memDistance + 0.6 * \bracketWid, 0.5*\height - \bracketHeight)$); 

\draw[line width=1.5pt] ($(t1-4) + (-\memDistance - 0.5 * \bracketWid, -0.5*\height)$) rectangle ($(t1-4) + (-\memDistance + 0.5 * \bracketWid, -0.5*\height + \bracketHeight)$);    
\draw[white, line width=3pt] ($(t1-4) + (-\memDistance - 0.6 * \bracketWid, -0.5*\height + \bracketHeight)$) -- ($(t1-4) + (-\memDistance + 0.6 * \bracketWid, -0.5*\height + \bracketHeight)$); 

\node[event] (t1-5) at (0 * \xstep, 4 * \ystep)  {$\wt({\color{red} \lk})$};
\node[event] (t1-6) at (0 * \xstep, 5 * \ystep)  {$\rd(\ch_y^1)$};
\node[event] (t1-7) at (0 * \xstep, 6 * \ystep)  {$\rd({\color{red} \lk})$};
\draw [memEvent] ($ (t1-5) + (-\memDistance, 0) $)  to node {$\rdMem_3(y)$} ($ (t1-7) + (-\memDistance, 0) $);
\draw[line width=1.5pt] ($(t1-5) + (-\memDistance - 0.5 * \bracketWid, 0.5*\height - \bracketHeight)$) rectangle ($(t1-5) + (-\memDistance + 0.5 * \bracketWid, 0.5*\height)$);    
\draw[white, line width=3pt] ($(t1-5) + (-\memDistance - 0.6 * \bracketWid, 0.5*\height - \bracketHeight)$) -- ($(t1-5) + (-\memDistance + 0.6 * \bracketWid, 0.5*\height - \bracketHeight)$); 

\draw[line width=1.5pt] ($(t1-7) + (-\memDistance - 0.5 * \bracketWid, -0.5*\height)$) rectangle ($(t1-7) + (-\memDistance + 0.5 * \bracketWid, -0.5*\height + \bracketHeight)$);    
\draw[white, line width=3pt] ($(t1-7) + (-\memDistance - 0.6 * \bracketWid, -0.5*\height + \bracketHeight)$) -- ($(t1-7) + (-\memDistance + 0.6 * \bracketWid, -0.5*\height + \bracketHeight)$); 

\node[event] (t2-1) at (1 * \xstep, 0 * \ystep)  {$\wt({\color{red} \lk})$};
\node[event] (t2-2) at (1 * \xstep, 1 * \ystep)  {$\rd(\ch_x^1)$};
\node[event] (t2-3) at (1 * \xstep, 2 * \ystep)  {$\rd({\color{red} \lk})$};
\draw [memEvent] ($ (t2-1) + (\memDistance, 0) $)  to node [midway,fill=white] {$\rdMem_1(x)$} ($ (t2-3) + (\memDistance, 0) $);  
\draw[line width=1.5pt] ($(t2-1) + (\memDistance - 0.5 * \bracketWid, 0.5*\height - \bracketHeight)$) rectangle ($(t2-1) + (\memDistance + 0.5 * \bracketWid, 0.5*\height)$);    
\draw[white, line width=3pt] ($(t2-1) + (\memDistance - 0.6 * \bracketWid, 0.5*\height - \bracketHeight)$) -- ($(t2-1) + (\memDistance + 0.6 * \bracketWid, 0.5*\height - \bracketHeight)$); 

\draw[line width=1.5pt] ($(t2-3) + (\memDistance - 0.5 * \bracketWid, -0.5*\height)$) rectangle ($(t2-3) + (\memDistance + 0.5 * \bracketWid, -0.5*\height + \bracketHeight)$);    
\draw[white, line width=3pt] ($(t2-3) + (\memDistance - 0.6 * \bracketWid, -0.5*\height + \bracketHeight)$) -- ($(t2-3) + (\memDistance + 0.6 * \bracketWid, -0.5*\height + \bracketHeight)$); 
    
\node[event] (t2-4) at (1 * \xstep, 3 * \ystep) {$\wt({\color{red} \lk})$};
\node[event] (t2-5) at (1 * \xstep, 4 * \ystep) {$\rd(\ch_x^2)$};
\node[event] (t2-6) at (1 * \xstep, 5 * \ystep) {$\rd({\color{red} \lk})$};
\draw [memEvent] ($ (t2-4) + (\memDistance, 0) $)  to node [midway,fill=white] {$\rdMem_2(x)$} ($ (t2-6) + (\memDistance, 0) $);
\draw[line width=1.5pt] ($(t2-4) + (\memDistance - 0.5 * \bracketWid, 0.5*\height - \bracketHeight)$) rectangle ($(t2-4) + (\memDistance + 0.5 * \bracketWid, 0.5*\height)$);    
\draw[white, line width=3pt] ($(t2-4) + (\memDistance - 0.6 * \bracketWid, 0.5*\height - \bracketHeight)$) -- ($(t2-4) + (\memDistance + 0.6 * \bracketWid, 0.5*\height - \bracketHeight)$); 

\draw[line width=1.5pt] ($(t2-6) + (\memDistance - 0.5 * \bracketWid, -0.5*\height)$) rectangle ($(t2-6) + (\memDistance + 0.5 * \bracketWid, -0.5*\height + \bracketHeight)$);    
\draw[white, line width=3pt] ($(t2-6) + (\memDistance - 0.6 * \bracketWid, -0.5*\height + \bracketHeight)$) -- ($(t2-6) + (\memDistance + 0.6 * \bracketWid, -0.5*\height + \bracketHeight)$); 

\node[event] (t2-7) at (1 * \xstep, 6 * \ystep) {$\wt({\color{red} \lk})$};
\node[event] (t2-8) at (1 * \xstep, 7 * \ystep) {$\wt(\ch_y^1)$};
\node[event] (t2-9) at (1 * \xstep, 8 * \ystep) {$\rd({\color{red} \lk})$};
\draw [memEvent] ($ (t2-7) + (\memDistance, 0) $)  to node [midway,fill=white] {$\wtMem_2(y)$} ($ (t2-9) + (\memDistance, 0) $);
\draw[line width=1.5pt] ($(t2-7) + (\memDistance - 0.5 * \bracketWid, 0.5*\height - \bracketHeight)$) rectangle ($(t2-7) + (\memDistance + 0.5 * \bracketWid, 0.5*\height)$);    
\draw[white, line width=3pt] ($(t2-7) + (\memDistance - 0.6 * \bracketWid, 0.5*\height - \bracketHeight)$) -- ($(t2-7) + (\memDistance + 0.6 * \bracketWid, 0.5*\height - \bracketHeight)$); 

\draw[line width=1.5pt] ($(t2-9) + (\memDistance - 0.5 * \bracketWid, -0.5*\height)$) rectangle ($(t2-9) + (\memDistance + 0.5 * \bracketWid, -0.5*\height + \bracketHeight)$);    
\draw[white, line width=3pt] ($(t2-9) + (\memDistance - 0.6 * \bracketWid, -0.5*\height + \bracketHeight)$) -- ($(t2-9) + (\memDistance + 0.6 * \bracketWid, -0.5*\height + \bracketHeight)$); 
    
\node[threadName] at (0 * \xstep, -1.0 * \ystep) {$\thread_{1}$};
\node[threadName] at (1 * \xstep, -1.0 * \ystep) {$\thread_{2}$};

\draw[rfEdge] (t1-2) to (t2-2);
\draw[rfEdge] (t1-3) to (t2-5);
\draw[rfEdge] (t2-8) to (t1-6);
\draw[rfEdge, out=\outOne, in=\inOne, looseness=\Looseness] (t1-1) to (t1-4);
\draw[rfEdge, out=\outOne, in=\inOne, looseness=\Looseness] (t1-5) to (t1-7);

\draw[rfEdge, out=\outTwo, in=\inTwo, looseness=\Looseness] (t2-1) to (t2-3);
\draw[rfEdge, out=\outTwo, in=\inTwo, looseness=\Looseness] (t2-4) to (t2-6);
\draw[rfEdge, out=\outTwo, in=\inTwo, looseness=\Looseness] (t2-7) to (t2-9);
    
\begin{scope}[on background layer]
\draw [thread] (0 * \xstep, -0.6 * \ystep) -- (0 * \xstep, 8.7 * \ystep);
\draw [thread] (1 * \xstep, -0.6 * \ystep) -- (1 * \xstep, 8.7 * \ystep);
\end{scope}
\end{tikzpicture}
}
\caption{The corresponding $\vchRf$ instance.}
\figlabel{np-hard-vsc-read-unbounded-demo}
\end{subfigure}
\caption{
A $\vscRd$ instance (\subref{fig:np-hard-vsc-read-bounded-demo}) and the corresponding $\vchRf$ instance (\subref{fig:np-hard-vsc-read-unbounded-demo}) with channel capacities of $1$.}
\figlabel{np-hard-vsc-read}
\end{figure}

For a write event $e = \ev{t, \wtMem(x)}$, $M(e)$ is a sequence
of $m_x$-many $\snd$ events, followed by $m_x - p_x$ $\rcv$ events,
all enclosed in a block of send-receive pair on channel $\lk$;
the thread identifier of each of the following event is $\thread$,
and we omit explicitly mentioning it.
\begin{align*}
M(e) = \snd(\s{\lk}) \cdot \snd(\ch_x^1) \cdots \snd(\ch_x^{m_x}) \cdot \rcv(\ch_x^{p_e+1}) \cdots \rcv(\ch_x^{m_x})\cdot \rcv(\s{\lk})
\end{align*}
Let us now discuss the encoding of read events.
For this, we assume some arbitrary ordering $\set{f_1, f_2, \ldots, f_{p_e}}$
of the set of read events reading from some write event $e$.
Then, the event sequence corresponding to the $i^\text{th}$
read event $e = \ev{\thread, \rdMem(x)}$ of some write event is:
\begin{align*}
M(e) = \snd(\s{\lk}) \cdot \rcv(\ch_x^i) \cdot \rcv(\s{\lk})
\end{align*}

The program order $\po{}'$ is then obtained by considering all pairs
of events of the form $(e_1, e_2)$ in $\eventSet'$ such that
either they belong to $M(e)$ for some $e$ and $e_1$ appears before $e_2$ in $M(e)$,
or they belong to $M(e)$ and $M(e')$ respectively with $(e, e') \in \po{}$.
The $\rf{}'$ relation is also straightforward.
For each event of the form $\rcv(\lk)$ in $M(e)$,
its corresponding send event is the unique $\snd(\lk)$ event in $M(e)$.
The send and receive events on channels of the form $\ch_x^i$ are paired
as follows.
Let $(e, f_i) \in \rf{}$ be a pair of write and its $i^\text{th}$
read event in $\eventSet$.
Then the send event $e'_i = \snd(\ch_x^i)$ in $M(e)$ is paired to
the event $f'_i = \rcv(\snd(\ch_x^i))$ in $M(f_i)$ (i.e., $(e'_i, f'_i) \in \rf{}'$).
Further, the unmatched send event
$e'_j = \snd(\ch_x^j)$ in $M(e)$, where $p_e+1 \leq j \leq m_x$
is paired with the $(j-p_e)^\text{th}$ receive event $f'_j = \rcv(\ch_x^j)$ in $M(e)$,
(i.e., $(e'_j, f'_j) \in \rf{}'$).


The correctness of the construction is relatively straightforward, and stated in the following lemma.

\begin{restatable}{lemma}{vchRfConstK}
    \lemlabel{vscRd-vchRf-relation}
    ${\AbstractExecution}$ is SC consistent iff  
    $\tuple{\AbstractExecution', \cpFunc', \rf{}'}$ is consistent.
\end{restatable}

We now argue about the time taken to construct $\tuple{\AbstractExecution', \cpFunc', \rf{}'}$.
Each write event in $\eventSet$ can be observed by at most $|\eventSet|$ different read events.
Each $e \in \eventSet$ is thus mapped to a sequence consisting of $O(|\eventSet|)$ events. 
Thus, $|\eventSet'| \in  O(|\eventSet^2|)$, which concludes case (i) of \thmref{vch-rf-hardness}.


\subsection{Hardness with $3$ Threads, $5$ Channels and no Capacity Restrictions}
\seclabel{hardness-const-t-m}
We show that $\vchRf$ remains intractable when both the number of threads and channels are constant, 
and there are no restrictions on channel capacities.

\begin{figure}[t]
\centering
\begin{subfigure}[b]{0.36\textwidth}
\centering
\newcommand{\xstep}{1.8}
\newcommand{\ystep}{-0.8}
\newcommand{\height}{0.8}
\newcommand{\wid}{4.6}
\newcommand{\sqrwid}{4.6}
\newcommand{\sqrheight}{4.6}
\tikzstyle{phase}=[event, draw=colorClause, rounded corners]
\tikzstyle{dots}=[align=center, line width=1.5pt, minimum width=\sqrwid cm,minimum height=\sqrheight cm]
\scalebox{0.70}{
\begin{tikzpicture}
    \node (t1-3) at (1 * \xstep, 0.5 * \ystep) [phase] { Phase-$0$ };
    \node (t1-3) at (1 * \xstep, 2 * \ystep) [phase] { Phase-$1$ };
    \node (t1-4) at (1 * \xstep, 3.1 * \ystep) [dots] {\bf $\vdots$ \quad\quad\quad $\vdots$ };
    \node (t1-5) at (1 * \xstep, 4.5 * \ystep) [phase] {Phase-$j$};
    \node (t1-4) at (1 * \xstep, 5.6 * \ystep) [dots] {\bf $\vdots$ \quad\quad\quad $\vdots$ };
    \node (t1-5) at (1 * \xstep, 7 * \ystep) [phase] {Phase-$\numClause$};

    \node[threadName] at (0 * \xstep, -1.1 * \ystep) {$\thread_{1}$ };
    \node[threadName] at (1 * \xstep, -1.1 * \ystep) {$\thread_{2}$ };
    \node[threadName] at (2 * \xstep, -1.1 * \ystep) {$\thread_{3}$ };

    \begin{scope}[on background layer]
        \draw [thread] (0 * \xstep, -0.8 * \ystep) -- (0 * \xstep, 8.5 * \ystep);
        \draw [thread] (1 * \xstep, -0.8 * \ystep) -- (1 * \xstep, 8.5 * \ystep);
        \draw [thread] (2 * \xstep, -0.8 * \ystep) -- (2 * \xstep, 8.5 * \ystep);
    \end{scope}
\end{tikzpicture}
}
\vspace{0.5cm}
\caption{Overall scheme}
\figlabel{reduction-const-t-m-vchRf-overall}
\end{subfigure}
\hfill
\begin{subfigure}[b]{0.63\textwidth}
\centering
\newcommand{\ybias}{-0.5 }
\newcommand{\xstep}{2.5}
\newcommand{\ystep}{-0.8}
\newcommand{\height}{0.5}
\newcommand{\wid}{2.0}
\newcommand{\sqrwid}{7.5}
\newcommand{\sqrheight}{0.4}
\newcommand{\memDistance}{2.3}
\newcommand{\bracketWid}{0.6}
\newcommand{\bracketHeight}{0.25}
\tikzstyle{memEvent}=[draw=none]
\tikzstyle{var}=[draw=colorVar, align=center, fill=white, line width=0.4mm, minimum width=\sqrwid cm,minimum height=\sqrheight cm, dashed, rounded corners]
\tikzstyle{dots}=[align=center, line width=1.5pt, minimum width=0.2*\sqrwid cm,minimum height=0.7*\sqrheight cm]
\tikzstyle{doubleLine}=[thin, double distance=1.0pt]
\scalebox{0.65}{
\begin{tikzpicture}
    \node (var-1) at (1 * \xstep, 0 * \ystep) [var] { $x_1$};
    \node (var-dots-1) at (1 * \xstep, 0.6 * \ystep) [dots] { $\vdots$ \quad\quad\quad $\vdots$};

   \node (bounderies) at (1 * \xstep, 3.7 * \ystep) [draw=colorVar, align=center, line width=0.4mm, minimum width=\sqrwid cm,minimum height=-5.2*\ystep cm, dashed, rounded corners] {};


    \node (var-i-1-1) at (0 * \xstep, 1.7 * \ystep) [event] {$\wt^i_\bot(\ch_1)$};
    \node (var-i-1-2) at (0 * \xstep, 2.7 * \ystep) [doubleLine, event] {$\rd^i_\bot(\ch_2)$};
    \node (var-i-1-3) at (0 * \xstep, 3.7 * \ystep) [doubleLine, event] {$\wt_\bot(c_q)$ \textcolor{colorSTAR}{$\star$}};
    \node (var-i-1-4) at (0 * \xstep, 4.7 * \ystep) [doubleLine, event] {$\rd^i_\bot(\ch_1)$};
    \node (var-i-1-5) at (0 * \xstep, 5.7 * \ystep) [event] {$\wt^i_\bot(\ch_2)$};

    \node (var-i-2-1) at (1 * \xstep, 1.7 * \ystep) [event] {$\wt^i_\top(\ch_2)$};
    \node (var-i-2-2) at (1 * \xstep, 2.7 * \ystep) [doubleLine, event] {$\rd^i_\top(\ch_1)$};
    \node (var-i-2-3) at (1 * \xstep, 3.7 * \ystep) [doubleLine, event] {$\wt_\top(c_q)$ \textcolor{colorSTAR}{$\star$}};
    \node (var-i-2-4) at (1 * \xstep, 4.7 * \ystep) [doubleLine, event] {$\rd^i_\top(\ch_2)$};
    \node (var-i-2-5) at (1 * \xstep, 5.7 * \ystep) [event] {$\wt^i_\top(\ch_1)$};

    \node [dots] (var-dots-2) at (1 * \xstep, 6.7 * \ystep) { $\vdots$ \quad\quad\quad $\vdots$};
    \node [var] (var-n) at (1 * \xstep, 7.5 * \ystep){ $x_{\numVar}$};

    \node (t1-6) at (0 * \xstep, 8.5 * \ystep) [doubleLine, event] {$\rd_\top(c_1)$};
    \node (t1-7) at (0 * \xstep, 9.5 * \ystep) [doubleLine, event] {$\rd_\bot(c_2)$};

    \node (t2-6) at (1 * \xstep, 8.5 * \ystep) [doubleLine, event] {$\rd_\top(c_2)$};
    \node (t2-7) at (1 * \xstep, 9.5 * \ystep) [doubleLine, event] {$\rd_\bot(c_3)$};

    \node (t3-6) at (2 * \xstep, 8.5 * \ystep) [doubleLine, event] {$\rd_\top(c_3)$};
    \node (t3-7) at (2 * \xstep, 9.5 * \ystep) [doubleLine, event] {$\rd_\bot(c_1)$};
    
    \draw [memEvent] ($ (var-i-1-1) + (-\memDistance, 0) $) to node [midway, fill=white, align=center] {$I^r_i$ or $A^r_{j, i}$ \\ for $x_i$} 
    ($ (var-i-1-5) + (-\memDistance, 0) $);
    \draw[line width=1.5pt] ($(var-i-1-1) + (-\memDistance - 0.5 * \bracketWid, 0.5*\height - \bracketHeight)$) rectangle ($(var-i-1-1) + (-\memDistance + 0.5 * \bracketWid, 0.5*\height)$);    
    \draw[white, line width=3pt] ($(var-i-1-1) + (-\memDistance - 0.6 * \bracketWid, 0.5*\height - \bracketHeight)$) -- ($(var-i-1-1) + (-\memDistance + 0.6 * \bracketWid, 0.5*\height - \bracketHeight)$); 

    \draw[line width=1.5pt] ($(var-i-1-5) + (-\memDistance - 0.5 * \bracketWid, -0.5*\height)$) rectangle ($(var-i-1-5) + (-\memDistance + 0.5 * \bracketWid, -0.5*\height + \bracketHeight)$);    
    \draw[white, line width=3pt] ($(var-i-1-5) + (-\memDistance - 0.6 * \bracketWid, -0.5*\height + \bracketHeight)$) -- ($(var-i-1-5) + (-\memDistance + 0.6 * \bracketWid, -0.5*\height + \bracketHeight)$); 

    \draw [memEvent] ($ (t1-6) + (-\memDistance, \height/
    2) $)  to node [midway,fill=white] {$B^r_j$} ($ (t1-7) + (-\memDistance, -\height/2) $);
    \draw[line width=1.5pt] ($(t1-6) + (-\memDistance - 0.5 * \bracketWid, 0.5*\height - \bracketHeight)$) rectangle ($(t1-6) + (-\memDistance + 0.5 * \bracketWid, 0.5*\height)$);    
    \draw[white, line width=3pt] ($(t1-6) + (-\memDistance - 0.6 * \bracketWid, 0.5*\height - \bracketHeight)$) -- ($(t1-6) + (-\memDistance + 0.6 * \bracketWid, 0.5*\height - \bracketHeight)$); 

    \draw[line width=1.5pt] ($(t1-7) + (-\memDistance - 0.5 * \bracketWid, -0.5*\height)$) rectangle ($(t1-7) + (-\memDistance + 0.5 * \bracketWid, -0.5*\height + \bracketHeight)$);    
    \draw[white, line width=3pt] ($(t1-7) + (-\memDistance - 0.6 * \bracketWid, -0.5*\height + \bracketHeight)$) -- ($(t1-7) + (-\memDistance + 0.6 * \bracketWid, -0.5*\height + \bracketHeight)$); 

    \draw [memEvent] (2 * \xstep + \memDistance, 0 * \ystep)  to node [midway,fill=white] {$I^r$ or $A^r_j$} ($ (t3-7) + (\memDistance, -\height/2) $);
    \draw[line width=1.5pt] ($(2 * \xstep, 0 * \ystep) + (\memDistance - 0.5 * \bracketWid, 0.5*\height - \bracketHeight)$) rectangle ($(2 * \xstep, 0 * \ystep) + (\memDistance + 0.5 * \bracketWid, 0.5*\height)$);    
    \draw[white, line width=3pt] ($(2 * \xstep, 0 * \ystep) + (\memDistance - 0.6 * \bracketWid, 0.5*\height - \bracketHeight)$) -- ($(2 * \xstep, 0 * \ystep) + (\memDistance + 0.6 * \bracketWid, 0.5*\height - \bracketHeight)$); 

    \draw[line width=1.5pt] ($(t3-7) + (\memDistance - 0.5 * \bracketWid, -0.5*\height)$) rectangle ($(t3-7) + (\memDistance + 0.5 * \bracketWid, -0.5*\height + \bracketHeight)$);    
    \draw[white, line width=3pt] ($(t3-7) + (\memDistance - 0.6 * \bracketWid, -0.5*\height + \bracketHeight)$) -- ($(t3-7) + (\memDistance + 0.6 * \bracketWid, -0.5*\height + \bracketHeight)$); 
    
    \node[threadName] at (0 * \xstep, -1.0 * \ystep) {$\thread_{1}$};
    \node[threadName] at (1 * \xstep, -1.0 * \ystep) {$\thread_{2}$}; 
    \node[threadName] at (2 * \xstep, -1.0 * \ystep) {$\thread_{3}$};
    \begin{scope}[on background layer]
        \node (t1-3) at (1 * \xstep, 4.8 * \ystep) [draw=colorClause, align=center, fill=white, line width=0.4mm, minimum width=4*\wid cm,minimum height=-10.5*\ystep cm, dashed, rounded corners] { $x_1$};
        \draw [thread] (0 * \xstep, -0.7 * \ystep) -- (0 * \xstep, 10.4 * \ystep);
        \draw [thread] (1 * \xstep, -0.7 * \ystep) -- (1 * \xstep, 10.4 * \ystep);
        \draw [thread] (2 * \xstep, -0.7 * \ystep) -- (2 * \xstep, 10.4 * \ystep);
    \end{scope}
\end{tikzpicture}
}
\caption{Events in Phase-$j$}
\figlabel{reduction-const-t-m-pahse}
\end{subfigure}
\caption{
Reduction from 3SAT to $\vchRf$ with capacity-unbounded channels. 
Events with double boundary do not appear in Phase-$0$.
Events marked with \textcolor{colorSTAR}{$\star$} only appear when  the $q^\text{th}$ literal in clause $C_j$ is over variable $x_i$
}
\figlabel{reduction-fix-t-m}
\end{figure} 

\myparagraph{Overview}{
Starting from a 3SAT instance $\psi$ with
$\numClause$ clauses $C_1, \ldots, C_{\numClause}$ over $\numVar$ propositional variables
$\set{x_1, \ldots, x_{\numVar}}$,
we construct a $\vchRf$ instance $\tuple{\AbstractExecution, \cpFunc, \rf{}}$
with $3$ threads $\thread_1, \thread_2, \thread_3$ and 
5 channels $\ch_1, \ch_2, c_1, c_2, c_3$.
Informally, $\tuple{\AbstractExecution, \cpFunc, \rf{}}$ consists of
$\numClause + 1$ phases, arranged sequentially. 
The first \emph{initialization} phase (`Phase-0') picks an assignment of boolean
values for each variable. 
The remaining $\numClause$ phases encode the requirement
that at least one literal from each clause is set to true.
Phase-$j$ (with $j \geq 1$) duplicates the assignment to all variables 
from the previous phase and checks if the chosen 
assignment makes clause $C_j$ true.
\figref{reduction-fix-t-m} depicts this scheme.
}

\myparagraph{Reduction}{
The sequence $\tr_r$ corresponding to events of thread $\thread_r$ ($r \in \set{1, 2, 3}$)
is of the form $\tr_r = I^r \cdot A^r_1 \cdot A^r_2 \cdots A^r_{\numClause}$.
The sequence corresponding to Phase-$0$ 
is of the form  $I^r = I^r_1 \cdots I^r_{\numVar}$,
where $I^r_p$ picks an assignment to variable $x_p$ in thread $\thread_r$:
\begin{align*}
\begin{array}{ccccc}
I^1_p = \snd^p_\bot(\ch_1) \cdot \snd^p_\bot(\ch_2) & \quad\quad & I^2_p = \snd^p_\top(\ch_2) \cdot \snd^p_\top(\ch_1) & \quad\quad & I^3_p = \epsilon
\end{array}
\end{align*}
Next, the sequence corresponding to thread $\thread_r$ and Phase-$j$ ($j \geq 1$)
is of the form $A^r_j = A^r_{j, 1} \cdots A^r_{j, \numVar} \cdot B^r_j$.
where $A^r_{j, p}$ corresponds to variable $x_p$ 
and $B^r_j$ encodes the satisfaction of clause $C_j$ (see \figref{reduction-fix-t-m} for illustration).
We describe these components next.
$A^3_{j, p} = \epsilon$ for every $j \in \set{1, \ldots, \numClause}$, $p \in \set{1, \ldots, \numVar}$.
When $r \in \set{1, 2}$, then $A^r_{j, p}$ is used to encode the variable $x_p$ in clause $C_j$ of thread $\thread_r$.
If $x_p$ appears in clause $C_j$ and it is the $p^\text{th}$ literal of $C_j$ (p $\in \set{1, 2, 3}$), 
then: 
\begin{align*}
\begin{array}{rcl}
A^1_{j, p} &=& \snd^p_\bot(\ch_1) \cdot \rcv^p_\bot(\ch_2) \cdot \snd_\bot(c_q) \cdot \rcv^p_\bot(\ch_1) \cdot \snd^p_\bot(\ch_2) \\
A^2_{j, p} &=& \snd^p_\top(\ch_2) \cdot  \rcv^p_\top(\ch_1) \cdot \snd_\top(c_q) \cdot \rcv^p_\top(\ch_2) \cdot \snd^p_\top(\ch_1)
\end{array}
\end{align*}
If $x_p$ is not in clause $C_j$, then:
\begin{align*}
\begin{array}{rcl}
A^1_{j, p} &=& \snd^p_\bot(\ch_1) \cdot  \rcv^p_\bot(\ch_2) \cdot \rcv^p_\bot(\ch_1) \cdot \snd^p_\bot(\ch_2) \\
A^2_{j, p} &=& \snd^p_\top(\ch_2) \cdot \rcv^p_\top(\ch_1) \cdot \rcv^p_\top(\ch_2) \cdot \snd^p_\top(\ch_1)
\end{array}
\end{align*}
Finally,
\begin{align*}
\begin{array}{ccccc}
B^1_j = \rcv_\top(c_1) \cdot  \rcv_\bot(c_2)
&\quad&
B^2_j = \rcv_\top(c_2) \cdot  \rcv_\bot(c_3)
&\quad&
B^3_j = \rcv_\top(c_3) \cdot \rcv_\bot(c_1)
\end{array}
\end{align*}
We now discuss the reads-from mappings. 

\begin{compactitem}
    \item The receive events $\rcv^p_\bot(\ch_2)$, $\rcv^p_\bot(\ch_1)$, $\rcv^p_\top(\ch_2)$ and $\rcv^p_\top(\ch_1)$ in $A^1_{j, p}, A^1_{j, p}, A^2_{j, p}, A^2_{j, p}$ are respectively mapped
to the send events $\snd^p_\bot(\ch_2)$, $\snd^p_\bot(\ch_1)$, $\snd^p_\top(\ch_2)$, $\snd^p_\top(\ch_1)$ in $A^1_{j-1, p}, A^1_{j-1, p}, A^2_{j-1, p}, A^2_{j-1, p}$ (or in $I^1_p, I^1_p, I^2_p, I^2_p$ if $j=1$).

    \item Let $C_j = \gamma_1 \lor \gamma_2 \lor \gamma_3$ such that $\gamma_q$ is either $x_{j_q}$ or $\neg x_{j_q}$.
For each $q \in \set{1, 2, 3}$, we have the following. If $\gamma_q = x_{j_q}$, then we require that the receive event $\rcv_\top(c_q)$ reads from send $\snd_\top(c_q)$ in $A^2_{j, j_q}$, and $\rcv_\bot(c_q)$ reads from $\snd_\top(c_q)$ in $A^1_{j, j_q}$.
Otherwise, we require that $\rcv_\top(c_q)$ reads from $\snd_\bot(c_q)$ in $A^1_{j, j_q}$ 
and $\rcv_\top(c_q)$ reads from $\snd_\top(c_q)$ in $A^2_{j, j_q}$. 
\end{compactitem}

The following lemma states the correctness of the above construction.

\begin{restatable}{lemma}{vchRfConstTM}
    \lemlabel{vch-rf-t=3-m=5}
    $\psi$ is satisfiable iff 
    $\tuple{\AbstractExecution, \cpFunc, \rf{}}$ is consistent.
\end{restatable}

Finally, the number of events in $\tuple{\AbstractExecution, \cpFunc, \rf{}}$ is $O(\numVar+\numClause)$,
which concludes case (iii) of \thmref{vch-rf-hardness}. 


\subsection{Quadratic Hardness with $2$ Threads}
\seclabel{two-threads-lower-bound-sec}

Finally, in this section we prove the quadratic hardness of $\vchRf$ over just $2$ threads when either all channels have capacity $1$ or have no capacity restrictions.
We achieve this by establishing a fine-grained reduction from the Orthogonal Vectors problem (OV)~\cite{Williams2005}.

\myparagraph{The Orthogonal Vectors problem}{
The OV problem takes as input two sets $A = \set{a_1, a_2 \ldots, a_n}, B = \set{b_1, b_2 \ldots, b_n} \subseteq 2^{\set{0, 1}^d}$,
each containing $n$ boolean vectors in $d$ dimensions.
The task is to determine whether there are two vectors $a \in A, b \in B$ such that $a$ and $b$ are orthgonal, i.e., $\langle a \cdot b\rangle = \sum_{i=1}^d a[i] \cdot b[i] = 0$.
Under the SETH, OV cannot be solved in time $O(n^{2-\epsilon})$, for every fixed $\epsilon > 0$, as long as $d=\omega(\log n)$~\cite{Williams2005}.
}

\input{figures/reduction-two-threads}

\myparagraph{Overview}{
We construct a $\vchRf$ instance $\tuple{\AbstractExecution, \cpFunc, \rf{}}$ which is consistent iff $A$ and $B$ contain an orthogonal vector pair.
$\AbstractExecution$ comprises two threads $\thread_A$ and $\thread_B$, respectively containing events encoding the vectors of $A$ and $B$.
\figref{reduction-two-threads}  illustrates the overall scheme.
At a high level, the reads-from edge due to the pair $\tuple{\snd(\gamma), \rcv(\gamma)} \in \rf{}$  triggers an orthogonality check between the vectors $a_1$ and $b_1$.
The reduction is built in such a way that this process of inference, called \emph{saturation} (formally described in \secref{saturation}), simulates orthogonality comparisons of the vectors.
If $a_1[i]=b_1[i]=1$ for some $i$, witnessing that $a_1$ and $b_1$ are not orthogonal, the corresponding events encoding $a_1$ and $b_1$ will contain two sends on the same channel, 
which triggers an orthogonality check between $a_1$ and $b_2$. 
If $a_1$ and $b_2$ are also not orthogonal, 
then $a_1$ and $b_3$ are compared, and so on. 
If the check between $a_1$ and $b_n$ fails, this triggers the check between $a_2$ and $b_1$, and the process continues, until an orthogonal pair is found, or the check between $a_n$ and $b_n$ does not identify an orthogonal pair.  
The fact that $a_n, b_n$ are not orthogonal implies $\rcv(\delta)$ must be ordered before $\snd(\delta)$, which contradicts with $\tuple{\snd(\delta), \rcv(\delta)}\in \rf{}$, 
implying that the constructed instance is not consistent. 
\revision{
We now formally describe the reduction. 
For simplicity, for certain events $e$, we use subscripts to indicate the specific vector that $e$ represents. 
}
}

\myparagraph{Reduction for capacity-unbounded channels}{
Given the OV instance $A, B$, we construct the corresponding $\vchRf$ instance using two threads $\thread_A$ and $\thread_B$ and channels $\set{\ch_1, \ch_2, \ldots, \ch_d, \alpha, \beta, \gamma, \delta}$, all having unbounded capacity.
We describe the events next, while using subscripts in the event operations that ensure that the combination of the operation, the subscript and the channel uniquely identify each event.
Send and receive events on the same channel that also have the same subscript are implicitly related by $\rf{}$.
The events of threads $\thread_A$ and $\thread_B$ are organized as follows:
\begin{align*}
\thread_A=A_{\sf init} \cdot A_1 \cdot A_2 \cdots A_n \qquad\text{and}
 \qquad 
\thread_B=B_{\sf init} \cdot B_n \cdot B_{n-1} \cdots B_1
\end{align*}
Observe that the order of appearance of $A_1, \ldots, A_n$ is the reverse of that of $B_n, \ldots, B_1$.
We next describe the contents of each block.
We use the notation $\snd_{a_i}(\ch_{a_i})$ to denote the sequence $\snd_{a_i}(\ch_{j_1}) \cdot \snd_{a_i}(\ch_{j_2}) \cdots \snd_{a_i}(\ch_{j_k})$, where $j_1, j_2, \ldots, j_k$ is the unique increasing sequence of indices in $\set{1, 2, \ldots, d}$ corresponding to non-zero entries in the vector $a_i$.
Likewise, $\snd_{b_i}(\ch_{b_i})$, $\rcv_{a_i}(\ch_{a_i})$ and $\rcv_{b_i}(\ch_{b_i})$ expand in a similar fashion.
The init block in $\thread_A$ contains send events for each vector $a \in A$ (on all those channels $\ch_i$ such that $a[i]=1$)  with alternating send events on channel $\alpha$, and likewise in $\thread_B$ (but in reverse order):
\begin{align*}
\begin{array}{rcl}
A_{\sf init} 
&=&
\snd_{a_1}(\ch_{a_1}) \cdot \snd_{a_1}(\alpha) \cdots \snd_{a_n}(\ch_{a_n}) \cdot \snd_{a_n}(\alpha) \\
B_{\sf init}
&=&
\snd_{b_n}(\alpha) \cdot \snd_{b_n}(\ch_{b_n}) \cdots  \snd_{b_1}(\alpha) \cdot \snd_{b_1}(\ch_{b_1})
\end{array}
\end{align*}
The extremal blocks ($A_1, B_1, A_n, B_n$) and internal blocks ($A_i, B_i$ where $2 \leq i \leq n-1$) are as follows.
\begin{align*}
\begin{array}{rcl}
A_1 
&=&
\rcv_{a_1}(\alpha) \cdot \snd(\gamma) \cdot \snd_{a_1}(\beta) \cdot \rcv_{a_1}(\ch_{a_1}) \\
A_n
&=&
\rcv_{a_n}(\alpha) \cdot \rcv_{a_{n-1}}(\beta) \cdot \rcv(\delta) \cdot \rcv_{a_n}(\ch_{a_n}) \\
B_n
&=&
\rcv_{b_n}(\ch_{b_n}) \cdot \snd(\delta) \cdot \snd_B(\beta) \\
B_1
&=&
\rcv_{b_1}(\ch_{b_1}) \cdot \rcv_{b_2}(\alpha) \cdot \rcv_B(\beta) \cdot \rcv(\gamma) \cdot \rcv_{b_1}(\alpha) \\
A_i 
&=&
\rcv_{a_i}(\alpha) \cdot \rcv_{a_{i-1}}(\beta) \cdot \snd_{a_i}(\beta) \cdot \rcv_{a_i}(\ch_{a_i}) \\
B_i 
&=&
\rcv_{b_i}(\ch_{b_i}) \cdot \rcv_{b_{i+1}}(\alpha)
\end{array}
\end{align*}


\input{figures/reduction-two-threads-demo}

The following lemma states the correctness of the construction.

\begin{restatable}{lemma}{ovcorrectness}
\lemlabel{ov-correctness}
$\tuple{\AbstractExecution, \cpFunc, \rf{}}$ is consistent iff $A$ and $B$ contain an orthogonal vector pair.
\end{restatable}

Regarding the time complexity, the reduction takes time proportional to $|A|+|B|$ i.e., $O(n\cdot d)$.
Hence, a subquadratic algorithm for deciding the consistency of $\tuple{\AbstractExecution, \cpFunc, \rf{}}$ would imply a subquadratic algorithm for solving OV, thereby contradicting SETH.
We thus arrive at case (i) of \thmref{vch-rf-two-threads-lower-bound}.
}

\begin{example}
\figref{reduction-two-threads-demo} illustrates an example with $n = d = 2$. 
The OV instance consists of the two sets $A = \{ a_1 = \tuple{0,1}, a_2 = \tuple{1,0}\}$ and $B = \{ b_1= \tuple{0,1}, b_2 = \tuple{1,1} \}$.
Since $a_2$ and $b_1$ are orthogonal, the constructed instance
$\Gamma = \tuple{\AbstractExecution, \cpFunc, \rf{}}$ is consistent.

We now explain how $\Gamma$ encodes orthogonality checks between the 
vectors of $A$ and $B$ via saturation.
Let us use $\satord$ to denote the saturated order of $\Gamma$;
its formal definition is presented in \secref{saturation}, though it can be intuitively understood
to be a set of orderings between events of $\Gamma$ that must hold
in every consistent concretization of $\Gamma$, if one exists (\lemref{saturation-equivalence}).
Since $(\po{} \cup \rf{})^+ \subseteq \satord$, we have $\rcv_{a_1}(\alpha) \satord \rcv_{b_1}(\alpha)$, which means $\snd_{a_1}(\alpha) \satord \snd_{b_1}(\alpha)$. 
This signifies that $a_1$ and $b_1$ are compared for orthogonality. Since $a_1$ and $b_2$ both have value 1 in dimension 2, they are not orthogonal.
This is witnessed by $\snd_{a_1}(\ch_2) \satord \snd_{b_1}(\ch_2)$, further leading to $\rcv_{a_1}(\ch_2) \satord \rcv_{b_1}(\ch_2)$. 
Due to $\po{}$, we also have $\rcv_{a_1}(\alpha) \satord \rcv_{a_1}(\ch_2) \satord \rcv_{b_1}(\ch_2) \satord \rcv_{b_2}(\alpha)$, 
and therefore $\snd_{a_1}(\alpha) \satord \snd_{b_2}(\alpha)$, which means that $a_1$ and $b_2$ are now compared for orthogonality. 
As before, $a_1$ and $b_2$ are not orthogonal, so we get the following sequence of inferences:
\begin{compactenum}
    \item $\snd_{a_1}(\ch_2) \satord \snd_{a_1}(\alpha) \satord \snd_{b_2}(\alpha) \satord \snd_{b_2}(\ch_2) \implies \rcv_{a_1}(\ch_2) \satord \rcv_{b_2}(\ch_2)$
    \item $\snd_{a_1}(\beta) \satord \rcv_{a_1}(\ch_2) \satord \rcv_{b_2}(\ch_2) \satord \snd_{B}(\beta)  \implies \rcv_{a_1}(\beta) \satord \rcv_{B}(\beta)$
    \item $\rcv_{a_2}(\alpha) \satord \rcv_{a_1}(\beta) \satord \rcv_{b}(\beta) \satord \rcv_{b_1}(\alpha) \implies \snd_{a_2}(\alpha) \satord \snd_{b_1}(\alpha)$
\end{compactenum}
Notice that we now start to check orthogonality between $a_2$ and $b_1$.

As $a_2$ and $b_1$ are orthogonal, the sequences $\ch_{a_2}(\ch_{a_2})$ and $\ch_{b_1}(\ch_{b_1})$ do not share any channels. Therefore, saturation stops inferring orderings at this point. However, the receive on $\delta$ also implies orderings via saturation. 
In fact, this also leads to orthogonality comparisons, but in a reversed order: $b_2$ is compared with $a_2$, then $b_1$ with $a_2$, and so on. 
The following sequence of inferences illustrates this:
\begin{compactenum}
    \item $\rcv_{b_2}(\ch_1) \satord \snd(\delta) \satord \rcv(\delta) \satord \rcv_{a_2}(\ch_1) \implies \snd_{b_2}(\ch_1) \satord \snd_{a_2}(\ch_1)$
    \item $\snd_{b_2}(\alpha) \satord \snd_{b_2}(\ch_1) \satord \snd_{a_2}(\ch_1) \satord \snd_{a_2}(\alpha) \implies \rcv_{b_2}(\alpha) \satord \rcv_{a_2}(\alpha)$
\end{compactenum}
The ordering of $\rcv_{b_2}(\alpha)$ before $\rcv_{a_2}(\alpha)$ is then what compares $b_1$ to $a_2$ (since $\rcv_{b_1}(\ch_2) \satord \rcv_{b_2}(\alpha)$). Again, the orthogonality of $a_1$ and $b_2$ stops the saturation process. 

We claim that the resulting $\satord$ contains no cycle and is strong enough to fully sequentialize $\AbstractExecution$. 
\end{example}

\myparagraph{Channels with capacity $1$}{
Finally, we argue about the quadratic hardness of $\vchRf$ when every channel has capacity $1$. 
This result follows a recent result that verifying sequential consistency with a reads-from mapping ($\vscRd$ problem) is OV-hard for $2$ threads~\cite{Mathur2020}. 
In \secref{vch-rf-np-hard-k=1},  
we have shown for any $\vscRd$ instance with $\NumEvents$ events, 
we can construct an equivalent $\vchRf$ instance with $O(m_{\mathcal{R}} \cdot \NumEvents)$ events, 
where $m_{\mathcal{R}}$ is the maximal number of read events that observe the same write event. 
Fortunately, in the reduction developed in~\cite{Mathur2020}, 
$m_{\mathcal{R}}$ is a constant, 
and therefore our $\vchRf$ instance is of linear size as the input $\vscRd$ problem. 
Since $\vscRd$ under two threads is OV-hard, 
$\vchRf$ cannot be solved in $O(\NumEvents^{2 - \epsilon})$ time for any $\epsilon > 0$, 
when there are 2 threads and every channel has capacity 1. 
Item (ii) of \thmref{vch-rf-two-threads-lower-bound} is thus proven. 
}



\section{Evaluation}
We have implemented our frontier graph algorithm for $\vchRf$
and evaluated its performance on 103 $\vchRf$ instances
 and compare against a constraint-solving approach for $\vchRf$.
Before we discuss our experimental setup (\secref{exp_setup}) and our evaluation results (\secref{exp_positive} and \secref{exp_mutate}),
we describe a crucial optimization, \emph{saturation} in \secref{saturation}.


\subsection{Saturation}
\seclabel{saturation}
Saturation, widely used as a pre-processing step in consistency checking~\cite{BoostingSC2020} 
and dynamic race detection~\cite{Pavlogiannis2020}, for shared register based concurrency, can be adapted to channel based concurrency.
At a high level, saturation infers additional event orderings in polynomial time before the core consistency-checking algorithm executes. 
For a $\vchRf$ instance, saturation infers orderings beyond 
the input program order and reads-from relations. 
For example, if two send events on the same channel are ordered by $\po{}$, 
then their corresponding receive events must also be ordered.

Formally, let $\Gamma = \tuple{\AbstractExecution, \cpFunc, \rf{}}$ 
be an input $\vchRf$ instance, 
where $\AbstractExecution = \tuple{\eventSet, \po{}}$.
The saturated order $\satPO$ of $\Gamma$
is defined to be the smallest transitive relation on events such that 
$(\po{} \cup \rf{})^+ \subseteq \satPO$ and:
\begin{compactenum}
    \item For each channel $\ch$, and pairs 
    $(\snd_1(\ch), \rcv_1(\ch)) \in \rf{}$, $(\snd_2(\ch), \rcv_2(\ch)) \in \rf{}$, we have
    $\snd_1(\ch) \satord \snd_2(\ch)$ iff $\rcv_1(\ch) \satord \rcv_2(\ch)$. 

    \item For each channel $\ch$, if $\snd_1(\ch)$ is a matched send event and $\snd_2(\ch)$ is an unmatched send event, then $\snd_1(\ch) \satord \snd_2(\ch)$.

    \item For each synchronous channel $\ch$, pair $(\snd(\ch), \rcv(\ch)) \in \rf{}$, 
    and event $ e \in \eventSet$, we have
    $e \satord \rcv(\ch)$ iff $e \satord \snd(\ch)$ and $\snd(\ch) \satord e$ iff $\rcv(\ch) \satord e$.

    \item For each channel $\ch$ with $\cp{\ch} = 1$, pair $(\snd_1(\ch), \rcv_1(\ch)) \in \rf{}$, and event $\snd_2(\ch) \neq \snd_1(\ch)$, we have $\snd_1(\ch) \satord \snd_2(\ch)$ iff $\rcv_1(\ch) \satord \snd_2(\ch)$.
\end{compactenum}
The saturated order $\satPO$ preserves the consistency of $\AbstractExecution$, as formalized in 
\lemref{saturation-equivalence}. 

\begin{lemma}
    \lemlabel{saturation-equivalence}
    Let $\Gamma = \tuple{\AbstractExecution, \cpFunc, \rf{}}$ be a $\vchRf$ instance and let
    $\satPO$ be the saturated order of $\Gamma$.
    If $\Gamma$ is consistent, then $\satPO$ is a partial order (i.e., does not contain cycles),
    and further, each concretization $\trace$ of $\Gamma$ (if any) must respect $\satPO$. 
\end{lemma}

\myparagraph{Constructing the saturated order}
We remark that $\satPO$ can be constructed in at most $O(\NumEvents^3)$
using a straightforward fix point algorithm.
Our actual implementation is optimized and crucially makes use of
the recently proposed data structure
Collective Sparse Segment Trees (CSSTs)~\cite{tuncc2024cssts}
designed for saturation-like fix point computations.

\myparagraph{Using the saturated partial order}
Based on \lemref{saturation-equivalence}, saturation enhances 
the decision procedure for $\vchRf$ in two key ways.
First, if the constructed order $\satPO$ is cyclic, then
we return NO directly, i.e., without explicitly running any further
consistency checking  (such as our frontier graph-based or a constraint solving based) procedure.
Second, when $\satPO$ is acyclic, then it can be used for
a more aggressive on-the-fly pruning
of the frontier graph.
More specifically, in our frontier graph algorithm, 
when exploring outgoing edges of a node $u = \tuple{\fgEventSet, \fgChanMap, \fgSyncSend}$,
we can prune all edges $u \xrightarrow{e} v$
for which $e$ is not enabled in $u$ according to the saturated order,
i.e., when $\exists e'$, s.t. $(e', e) \in \satPO$
but $e'$ is not in the set $\fgEventSet$ of events executed so far.
This optimization significantly reduces the number of paths the algorithm must explore. 
We remark that we also augment the constraints (corresponding to the $\vchRf$ problem)
with additional constraints induced by $\satPO$.
However, as we observed in our evaluation, solver-based algorithms cannot benefit 
from this acceleration, as saturation increases the number of such constraints, 
slowing down performance. 

\myparagraph{Discussion} 
\revision{
A saturation-style algorithm for consistency checking in \emph{Message Sequence Charts} (MSCs) is proposed in \cite{DiGiusto2023}. 
Our setting generalizes this approach to a more expressive model. 
In MSCs, channels are assumed to be unbounded, and thus no saturation rules are needed to capture capacity constraints.
In contrast, our model supports both bounded and unbounded channels, requiring valid interleavings to respect the capacity limits.
This is reflected in our saturation rules -- specifically, rules (3) and (4) correspond to capacities 0 and 1.
While additional rules could model higher capacities, we omit them to preserve the efficiency of the heuristic saturation procedure.
}



\subsection{Experimental Setup}
\seclabel{exp_setup}

Given that the problem of consistency checking arises in practical program
analyses, the goal of our evaluation is to gauge how
our algorithms perform in practice and how they compare
against generic solutions such as the use of $\smt$ solvers.

\myparagraph{Compared methods and implementations}
For our evaluation we focus on the $\vchRf$ problem and skip
evaluation for $\vch$ for two reasons --- $\vchRf$ is more prominent
in applications such as model checking and predictive analyses, and 
further, in our experimental setup, accurate and fast logging of
values in executions is challenging.
Recall that $\vchRf$ is an $\NP$-complete problem in general
and can alternatively be solved using a constraint solving approach such as
the use of $\smt$ solvers, and we also implement such a solution for our evaluation.
We compare the performance between two approaches: 
(1) the frontier graph algorithm $\fgAlgo$ and (2) $\smt$-based solvers $\smt$, along with their respective saturated variants ($\fgAlgoSat$ and $\smtSat$).
In our $\smt$ encoding, for each event $e$, we use an integer variable 
$0 \le x_e \le \NumEvents - 1$, representing the position of $e$ in a potential concretization.
For each channel $\ch$, we introduce $2\NumEvents + 2$ 
auxiliary integer variables to model (1) the cumulative count of send events and (2) the cumulative count of receive events across all prefixes in a potential concretization
(complete encoding details can be found in \appref{smt-encoding}).
Recall that the saturation-augmented versions $\smtSat$ and $\fgAlgoSat$ 
pre-emptively reject if saturation produces a cycle.
If saturation succeeds, then in $\smtSat$,
we augment the $\smt$ formula with the constraint $x_e < x_{e'}$,
where $(e, e') \in \satPO$ and $e'$ is the earliest event in $\ThreadOf{e'}$
with this property.
All algorithms are implemented in Java 23 and $\smt$/$\smtSat$ use
the Java bindings of \texttt{Z3}-4.12.2~\cite{de2008z3}.

\myparagraph{Benchmark programs}
Our evaluation subjects comprise two distinct groups. 
The first group is primarily derived from GoBench \cite{yuan2021gobench}, a widely used Golang concurrency bug benchmark suite. 
GoBench includes 82 real-world bugs from 9 popular open-source projects (GoReal) and 103 bug kernels (GoKer)~\cite{saioc2025dynamic, jiang2023effective}. 
From GoReal, we selected 6 projects for evaluation;
the remaining 3 projects were excluded either because we could not log execution traces
from them or the generated logs were too short. 
Similarly, we omitted GoKer benchmarks because they produce executions with too few channel operations to be meaningful for our analysis.
The second group consists of additional prominent Golang open-source projects, namely 
{\tt rpcx}, {\tt raft}, {\tt go-dsp}, {\tt bigcache}, {\tt telegraf}, {\tt ccache}, and {\tt v2ray}, selected to further validate our approach.

\myparagraph{Generation of positive $\vchRf$ instances}
For each benchmark, we first select 1–3 test cases and generate a 
totally ordered log of channel events in their runtime execution,
by modifying $\tsan$ \cite{serebryany2009threadsanitizer}. 
We performed a linear time sanity check that each 
recorded execution satisfies channel capacity constraints. 
Each execution log can now be translated to a $\vchRf$ instance by discarding 
the total order  between events and only retaining program order and the inferred reads-from relations. 
As a result, the instances thus obtained are positive instances, 
since the original execution serves as a valid concretization of them.
To evaluate algorithmic scalability, 
we additionally generate new $\vchRf$ instances by
keeping varying length prefixes of existing instances
corresponding to long executions (containing thousands to millions of channel accesses).
Detailed statistics of these instances can be found in \appref{pos-statistics}.

\myparagraph{Generation of mutated instances}
Recall from the previous paragraph that $\vchRf$ instances
obtained from real executions are bound to be positive.
To obtain negative instances,
we mutate the previously obtained positive instances by performing targeted modifications 
to their reads-from relation as follows.
We randomly select a reads-from pair 
$(\snd_1(\ch), \rcv_1(\ch))$ and another send event $\snd_2(\ch)$ on the same channel.
If $\snd_2(\ch)$ has a matching receive event $\rcv_2(\ch)$, 
then we swap the two reads-from mappings, i.e. enforce
$(\snd_1(\ch), \rcv_2(\ch)), (\snd_2(\ch), \rcv_1(\ch)) \in \rf{}$. 
Otherwise, when $\snd_2(\ch)$ is not received, 
we enforce $(\snd_2(\ch), \rcv_1(\ch)) \in \rf{}$ 
and delete the reads-from pair $(\snd_1(\ch), \rcv_1(\ch))$.
For each positive $\vchRf$ instance with $\NumEvents$ events, 
we mutate it $\max(5, 0.05 \NumEvents)$ times. 
While these mutations do not theoretically guarantee inconsistency, 
our experimental results show that 88.3\% (91/103) of mutated instances 
become inconsistent, 8.7\% (9/103) remain consistent, while
the consistency status of the remaining 2.9\% (3/103) instances could not
be determined due to algorithm timeouts.

\myparagraph{Machine configuration and metrics reported}
The experiments are conducted on a 2.0GHz 64-bit Linux machine. 
We set the heap size of JVM to be 100GB and timeout to be 3 hours. 
For each $\vchRf$ instance, we report key parameters, such as number of events, 
threads, channels and maximal channel capacity, 
as well as the running time of each algorithm. 
Times reported denote average running time over 3 repeated runs. 

\subsection{Evaluation Results for Consistent Instances}
\seclabel{exp_positive}


\begin{figure*}
\begin{subfigure}{.47\textwidth}
  \centering
  \scalebox{0.70}{
  \centering
  \begin{tikzpicture}
    \centering
    \begin{loglogaxis}[
        xtick distance=10^1,
        ytick distance=10^1,
        xlabel={\large $\fgAlgo$ time (s)},
        ylabel={\large SMT time (s)},
        ymax = 2*10^4,
        ymin = 0.01,
        xmin = 0.01,
        xmax = 2*10^4,
        legend style={at={(0.48,0.1)},anchor=south west, font=\footnotesize}
    ]
    
    \addplot[
        only marks,
        mark=*,
        mark options={fill=green},
        mark size=3.0pt, 
    ]
    table[x=FG,y=SMT,col sep=comma] {tables/exp_pos_data_FG_finish_SMT_timeout.csv};
    \addplot[
        only marks,
        mark=triangle*,
        mark options={fill=yellow},
        mark size=3.5pt
    ]
    table[x=FG,y=SMT,col sep=comma] {tables/exp_pos_data_FG_finish_SMT_finish.csv};
    \addplot[
        only marks,
        mark=diamond*,
        mark options={fill=red},
        mark size=3.5pt
    ]
    table[x=FG,y=SMT,col sep=comma] {tables/exp_pos_data_FG_timeout_SMT_finish.csv};
    \addplot[
        only marks,
        mark=square*,
        mark options={fill=blue},
        mark size=2.0pt,
        scatter/classes={
            a={green},
            b={red},
            c={yellow},
            d={blue}
        }
    ]
    table[x=FG, y=SMT,col sep=comma] {tables/exp_pos_data_FG_timeout_SMT_timeout.csv};
    \legend{SMT times out (35), Both finish (7), $\fgAlgo$ times out (2), Both time out (59)}
    \addplot [
        domain=0.01:1e6, 
        samples=100, 
        color=black,
    ]{x};
    \end{loglogaxis}
    \end{tikzpicture}
  }
  \vspace{-0.1in}
  \caption{$\fgAlgo$ speedup on consistent instances}
  \figlabel{speedup-pos-fg}
\end{subfigure}
\begin{subfigure}{.47\textwidth}
  \centering
  \scalebox{0.70}{
  \begin{tikzpicture}
    \begin{loglogaxis}[
        xtick distance=10^1,
        ytick distance=10^1,
        xlabel={\large $\fgAlgoSat$ time (s)},
        ylabel={\large SMT-Sat time (s)},
        ymax = 2*10^4,
        ymin = 0.01,
        xmin = 0.01,
        xmax = 2*10^4,
        legend style={at={(0.40, 0.1)},anchor=south west, font=\footnotesize}
    ]
    \addplot[
        only marks,
        mark=*,
        mark options={fill=green},
        mark size=3.0pt, 
    ]
    table[x=FG-Sat,y=SMT-Sat,col sep=comma] {tables/exp_pos_data_FG-Sat_finish_SMT-Sat_timeout.csv};
    \addplot[
        only marks,
        mark=triangle*,
        mark options={fill=yellow},
        mark size=3.5pt
    ]
    table[x=FG-Sat,y=SMT-Sat,col sep=comma] {tables/exp_pos_data_FG-Sat_finish_SMT-Sat_finish.csv};
    \addplot[
        only marks,
        mark=square*,
        mark options={fill=blue},
        mark size=2.0pt,
        scatter/classes={
            a={green},
            c={blue},
            d={yellow}
        }
    ]
    table[x=FG-Sat,y=SMT-Sat,col sep=comma] {tables/exp_pos_data_FG-Sat_timeout_SMT-Sat_timeout.csv};
    \legend{SMT-Sat times out (86) , Both finish (10), Both time out (7)}
    \addplot [
        domain=0.01:1e6, 
        samples=100, 
        color=black,
    ]
    {x};
    \end{loglogaxis}
    \end{tikzpicture}
  }
  \vspace{-0.1in}
  \caption{$\fgAlgoSat$ speedup on consistent instances}
  \figlabel{speedup-pos-fg-sat}
\end{subfigure}
\caption{Running time of $\smt, \smtSat, \fgAlgo, \fgAlgoSat$ on consistent instances. 
The legend indicates the number of instances in each class. 
The running time for each instance can be found in \appref{pos-statistics}.
}
\figlabel{speedup-pos}
\end{figure*}

\myparagraph{Comparison between $\fgAlgo$ and $\smt$ (\figref{speedup-pos-fg})}
$\smt$ times out on 35 instances due to excessive memory consumption;
$\fgAlgo$ solves all of these successfully.
$\fgAlgo$ fails on only 2 instances that $\smt$ completes.
In addition, when both algorithms succeed, $\fgAlgo$ outperforms $\smt$
in running time by a factor of 5–50,000$\times$ 
(full statistics can be found in \appref{pos-statistics}). 
These results demonstrate that $\fgAlgo$ scales significantly better than $\smt$ on most benchmarks.

\myparagraph{Comparison between $\fgAlgoSat$ and $\smtSat$ (\figref{speedup-pos-fg-sat})}
$\smtSat$ times out on 90.3\% (93/103) instances due to increased formula size, 
which leads to higher memory consumption compared to standard $\smt$. 
In contrast, $\fgAlgoSat$ successfully completes 93.2\% (96/103) of instances 
and can often scale to instances with 50k events. 
These results demonstrate that $\fgAlgoSat$ 
achieves significantly better scalability than $\smtSat$ on consistent benchmarks.

\myparagraph{Impact of saturation}
Saturation substantially enhances the performance of 
our frontier graph algorithm --- $\fgAlgoSat$ solves 54 more instances than $\fgAlgo$.
Saturation induces a slight slowdown on some instances 
(due to the overhead of the fixpoint computation), but this is
largely limited to smaller instances where
$\fgAlgo$ already finishes very quickly.
The impact of saturation on $\smt$ solvers is limited. 
$\smtSat$ successfully solves only 1 additional instance compared to $\smt$, 
and demonstrates significant speed improvements on just 3 benchmarks. 
We hypothesize that this marginal gain occurs because saturation increases the $\smt$ formula size. 


\subsection{Evaluation Results for Mutated Instances}
\seclabel{exp_mutate}


\begin{figure*}
\begin{subfigure}{.47\textwidth}
  \centering
  \scalebox{0.70}{
  \centering
  \begin{tikzpicture}
    \centering
    \begin{loglogaxis}[
        xtick distance=10^1,
        ytick distance=10^1,
        xlabel={\large $\fgAlgo$ time (s)},
        ylabel={\large SMT time (s)},
        ymax = 2*10^4,
        ymin = 0.01,
        xmin = 0.01,
        xmax = 2*10^4,
        legend style={at={(0.45,0.1)},anchor=south west}
    ]
    \addplot[
        only marks,
        mark=*,
        mark options={fill=green},
        mark size=3.0pt, 
    ]
    table[x=FG,y=SMT,col sep=comma] {tables/exp_mut_data_FG_finish_SMT_timeout.csv};
    \addplot[
        only marks,
        mark=triangle*,
        mark options={fill=yellow},
        mark size=3.5pt
    ]
    table[x=FG,y=SMT,col sep=comma] {tables/exp_mut_data_FG_finish_SMT_finish.csv};
    \addplot[
        only marks,
        mark=diamond*,
        mark options={fill=red},
        mark size=3.5pt
    ]
    table[x=FG,y=SMT,col sep=comma] {tables/exp_mut_data_FG_timeout_SMT_finish.csv};
    \addplot[
        only marks,
        mark=square*,
        mark options={fill=blue},
        mark size=2.0pt,
        scatter/classes={
            a={green},
            b={red},
            c={yellow},
            d={blue}
        }
    ]
    table[x=FG, y=SMT,col sep=comma] {tables/exp_mut_data_FG_timeout_SMT_timeout.csv};
    \legend{SMT times out (32), Both finish (6), $\fgAlgo$ times out (6), Both time out (59)}
    \addplot [
        domain=0.01:1e6, 
        samples=100, 
        color=black,
    ]{x};
    \end{loglogaxis}
    \end{tikzpicture}
  }
  \vspace{-0.1in}
  \caption{$\fgAlgo$ speedup on mutated instances}
  \figlabel{speedup-mut-fg}
\end{subfigure}
\begin{subfigure}{.47\textwidth}
  \centering
  \scalebox{0.70}{
  \begin{tikzpicture}
    \begin{loglogaxis}[
        xtick distance=10^1,
        ytick distance=10^1,
        xlabel={\large $\fgAlgoSat$ time (s)},
        ylabel={\large SMT-Sat time (s)},
        ymax = 2*10^4,
        ymin = 0.01,
        xmin = 0.01,
        xmax = 2*10^4,
        legend style={at={(0.42, 0.1)},anchor=south west}
    ]
    \addplot[
        only marks,
        mark=*,
        mark options={fill=green},
        mark size=3.0pt, 
    ]
    table[x=FG-Sat,y=SMT-Sat,col sep=comma] {tables/exp_mut_data_FG-Sat_finish_SMT-Sat_timeout.csv};
    \addplot[
        only marks,
        mark=triangle*,
        mark options={fill=yellow},
        mark size=3.5pt
    ]
    table[x=FG-Sat,y=SMT-Sat,col sep=comma] {tables/exp_mut_data_FG-Sat_finish_SMT-Sat_finish.csv};
    \addplot[
        only marks,
        mark=square*,
        mark options={fill=blue},
        mark size=2.0pt,
        scatter/classes={
            a={green},
            c={blue},
            d={yellow}
        }
    ]
    table[x=FG-Sat,y=SMT-Sat,col sep=comma] {tables/exp_mut_data_FG-Sat_timeout_SMT-Sat_timeout.csv};
    \legend{SMT-Sat times out (6) , Both finish (94), Both time out (3)}
    \addplot [
        domain=0.01:1e6, 
        samples=100, 
        color=black,
    ]
    {x};
    \end{loglogaxis}
    \end{tikzpicture}
  }
    \vspace{-0.1in}
  \caption{$\fgAlgoSat$ speedup on mutated instances}
  \figlabel{speedup-mut-fg-sat}
\end{subfigure}
\caption{Running time of $\smt, \smtSat, \fgAlgo, \fgAlgoSat$ on mutated instances. 
The legend indicates the number of instances in each class. 
The running time for each instance can be found in \appref{pos-statistics}.
}
\figlabel{speedup-mut}
\end{figure*}

\myparagraph{Comparison between $\fgAlgo$ and $\smt$ (\figref{speedup-mut-fg})}
$\fgAlgo$ successfully solves 26 more instances than $\smt$. 
For instances where both algorithms complete, $\fgAlgo$ achieves a speedup ranging from 3$\times$ to 3000$\times$
(full statistics can be found in \appref{pos-statistics}). 
Compared to its performance on consistent instances, 
$\fgAlgo$ solves $4$ fewer cases. 
This happens for inconsistent $\vchRf$ instances where $\fgAlgo{}$ performs a complete 
traversal of the frontier graph, resulting in increased computational time.

\myparagraph{Comparison between $\fgAlgoSat$ and $\smtSat$ (\figref{speedup-mut-fg-sat})}
Both algorithms demonstrate strong performance, 
successfully completing most instances, with $\smtSat$ solving 91.3\% (94/103) 
instances and $\fgAlgoSat$ solving 97.1\% (100/103).
The superior performance can be attributed to saturation's ability to efficiently reject nearly all inconsistent instances before initiating the core consistency checking algorithm. 
Notably, among the 6 instances where $\smtSat$ times out but $\fgAlgoSat$ succeeds, 
all are consistent instances (recall that mutated instances are not guaranteed to be inconsistent). 
In these cases, saturation not only fails to benefit $\smt$ solvers 
but actually degrades their performance due to the increased formula size.

\begin{figure*}
\begin{subfigure}{.48\textwidth}
\centering
\scalebox{0.70}{
\centering
\begin{tikzpicture}
\begin{axis}[
    xmode=log,
    xmin=0.05, xmax=10800,  
    ymin=0, ymax=110,            
    xlabel=Time (s),
    ylabel=Finished Instances,
    legend style={at={(0.03,0.55)},anchor=south west, font=\footnotesize},
]

\addplot [black, thick] coordinates {(0.05,103) (10800,103)};
\addlegendentry{total 103}

\addplot [
    mark=diamond*,
    mark options={fill=red},
    mark size=3pt,
    thick
] table [
x=Time, 
y=SMT, 
col sep=comma,
] {tables/survival_plot_pos.csv};
\addlegendentry{SMT}

\addplot [
    mark=*,
    mark options={fill=green},
    mark size=3pt,
    thick
] table [x=Time, y=FG, col sep=comma] {tables/survival_plot_pos.csv};
\addlegendentry{FG}

\addplot [
    mark=triangle*,
    mark options={fill=yellow},
    mark size=3pt,
    thick
] table [x=Time, y=SMT-Sat, col sep=comma] {tables/survival_plot_pos.csv};
\addlegendentry{SMT-Sat}

\addplot [
    mark=square*,
    mark options={fill=blue},
    mark size=2pt,
    thick
] table [x=Time, y=FG-Sat, col sep=comma] {tables/survival_plot_pos.csv};
\addlegendentry{FG-Sat}

\end{axis}
\end{tikzpicture}
}
\caption{Survival plot of consistent instances}
\figlabel{survival-plot-pos}
\end{subfigure}
\begin{subfigure}{.48\textwidth}
\centering
\scalebox{0.70}{
\centering
\begin{tikzpicture}
\begin{axis}[
    xmode=log,
    xmin=0.05, xmax=10800,  
    ymin=0, ymax=110,            
    xlabel=Time (s),
    ylabel=Finished Instances,
    legend style={at={(0.03,0.55)},anchor=south west, font=\footnotesize}
]

\addplot [black, thick] coordinates {(0.05,103) (10800,103)};
\addlegendentry{total 103}

\addplot [
    mark=diamond*,
    mark options={fill=red},
    mark size=3pt,
    thick
] table [
x=Time, 
y=SMT, 
col sep=comma,
] {tables/survival_plot_mut.csv};
\addlegendentry{SMT}

\addplot [
    mark=*,
    mark options={fill=green},
    mark size=3pt,
    thick
] table [x=Time, y=FG, col sep=comma] {tables/survival_plot_mut.csv};
\addlegendentry{FG}

\addplot [
    mark=triangle*,
    mark options={fill=yellow},
    mark size=3pt,
    thick
] table [x=Time, y=SMT-Sat, col sep=comma] {tables/survival_plot_mut.csv};
\addlegendentry{SMT-Sat}

\addplot [
    mark=square*,
    mark options={fill=blue},
    mark size=2pt,
    thick
] table [x=Time, y=FG-Sat, col sep=comma] {tables/survival_plot_mut.csv};
\addlegendentry{FG-Sat}

\end{axis}
\end{tikzpicture}
}
\caption{Survival plot of mutated instances}
\figlabel{survival-plot-mut}
\end{subfigure}
\caption{
Cumulative number of finished instances over time for each algorithm.
}
\figlabel{survival-plot}
\end{figure*}
\myparagraph{Survival analysis}{
\revision{
In \figref{survival-plot}, we present the cumulative number of finished instances over time for each algorithm.
All algorithms reach saturation after approximately 100 seconds,
illustrating the NP-hard nature of the $\vchRf$ problem -- additional time does not lead to solving more instances.
For consistent instances, $\fgAlgoSat$ initially completes fewer instances than $\fgAlgo$ due to the overhead introduced by the saturation process.
However, as time progresses, $\fgAlgoSat$ overtakes $\fgAlgo$ by efficiently pruning infeasible execution paths.
For mutated instances, $\fgAlgoSat$ behaves almost identically to $\smtSat$, benefiting from the saturation phase, which filters out inconsistent instances before the main solving procedure begins.
Overall, the survival analysis demonstrates that $\fgAlgoSat$ consistently solves the largest number of instances within the same time budget.
}
}

In summary, the frontier graph algorithm demonstrates 
better performance over $\smt$ solving based approach, both with or without saturation.
Frontier graph algorithm successfully completes more instances across all benchmarks. 
When both approaches terminate, 
the frontier graph algorithm achieves significant speedups ranging from 3$\times$ to 50,000$\times$. 
Further, we observe that
saturation can significantly enhance the effectiveness of consistency checking in practice.


\section{Other Related Work}
\seclabel{sec:related}

\myparagraph{Verifying linearizability}{
Verifying channel consistency resembles verifying linearizability (\textsc{VL})~\cite{herlihy1990linearizability,gibbons1997testing,bouajjani2018reducing,bouajjani2017proving,emmi2019violat,linearizabilityParosh2025,gibbons2002VL,linearizabilityLee2025}, which checks whether a concurrent history over a (queue) object is equivalent to some sequential one.
Unlike $\vch$, \textsc{VL} enjoys locality, enabling a linear-time decomposition into per-object histories.
Moreover, \textsc{VL} typically operates on interval partial orders, making it simpler than $\vch$.
Finally, \textsc{VL} over queues reduces to $\vch$ in linear time.
}

\myparagraph{Message sequence charts}{
A closely related problem is that of Message Sequence Charts (MSCs)~\cite{Madhusudan2001,Alur2005,DiGiusto2023,GenestMuschollSurvey2005}, where threads communicate via peer-to-peer channels. Among these, the work of Di~Giusto et al.~\cite{DiGiusto2023} is most relevant. However, our problem strictly generalizes MSCs:
(i)~send and receive events in $\vch$ need not be paired a priori;
(ii)~both $\vch$ and $\vchRf$ allow channels shared by more than two threads;
(iii)~the same thread pair may communicate over multiple channels; and
(iv)~channels may have bounded capacities.
\revision{
Our communication model, inspired by languages like Go, assumes channel-based FIFO semantics where all threads can access any channel, with operations on each channel totally ordered. 
Consequently, their results do not directly extend to our setting. 
Indeed, in \cite{DiGiusto2023}, it is shown that the consistency checking problem for these communication models can be solved in polynomial time, because the consistency predicate is MSO definable.
In contrast, the consistency checking problem for our setting is NP-hard.
}
}

\myparagraph{Model checking}{
\revision{
The \texttt{SPIN} model checker~\cite{SPIN1997} provides support for programs with message passing, and the \texttt{GOMELA} \cite{Dilley2022} bounded model checker extends it to support Go programs and uses \texttt{SPIN} as its backend. 
Given an LTL formula $\psi$ and a program $P$, \texttt{SPIN} constructs an automaton for $\neg \psi$ and $P$, 
and then searches for traces accepted by their product. 
If such a trace exists, the property $\psi$ does not hold for $P$. 
That is, \texttt{SPIN} does not explicitly address the consistency-checking problem. 
Moreover, $\mathsf{MUST}$~\cite{MUST2024} is another model checker that operates under the message-passing semantics of MSCs. 
While $\mathsf{MUST}$ implicitly performs consistency checks as part of its model checking algorithm, 
\cite{MUST2024} does not investigate the complexity-theoretic aspects of these checks.
}
}

\myparagraph{Register consistency checking}{
The consistency checking problem for registers has been extensively studied in prior work ~\cite{gibbons1997testing, gibbons1994testing, cantin2005complexity, chen2012program}. 
As demonstrated in this paper, channel consistency checking is strictly harder than register consistency checking due to a key difference in their semantics: registers can only retain the most recent write event, whereas channels can remember up to capacity send events. 
Related algorithms have also been developed for consistency checking under weak memory models, including TSO ~\cite{manovit2006completely, hu2011linear,chini2020} and C11 ~\cite{tuncc2023optimal,  chakraborty2024hard}.
}

\myparagraph{Predictive analysis}{
    Predictive analysis is a dynamic analysis technique that takes a program execution as input and reorders it to expose potential concurrency bugs. 
    Recent work has developed predictive algorithms for detecting data races ~\cite{mathur2021optimal, Pavlogiannis2020, flanagan2009fasttrack,ang2025enhanceddataraceprediction,OSR2024,SHB2018,FarzanMathur2024,Kulkarni2021FineGrained}, deadlocks ~\cite{kalhauge2018sound, tuncc2023sound}, atomicity violations ~\cite{flanagan2008velodrome, mathur2020atomicity} and more general properties~\cite{AngMathur2024Trace,AngMathur2024Pattern}. 
    These algorithms typically compute a candidate set of events and attempt to serialize them into a witness execution, often using a consistency checking oracle. 
    Thus, predictive analysis can be viewed as a downstream application of consistency checking. 
    However, existing prediction algorithms almost exclusively target shared-memory concurrency, neglecting executions involving message-passing via channels. 
    Our work bridges this gap by establishing theoretical foundations for channel-based predictive analysis.
}


\section{Conclusion}
Consistency testing is a fundamental task for analyses of concurrent programs such as model checking and predictive testing. 
We conducted a thorough complexity-theoretic investigation for
this problem for the message-passing programming paradigm with $\fifo$ channels.
We have developed novel algorithms and established hardness results for a range of inputs parameters.
We further implemented and empirically evaluated the performance of our algorithms.
Together, our upper and lower bounds reveal an intricate complexity landscape and our evaluation
of our algorithms demonstrate their effectiveness in practice.
Future work includes applying these algorithms to partial order reduction based model checking and predictive testing for message-passing languages such as Go.

\begin{acks}
Zheng Shi and Umang Mathur are partially supported by the National Research Foundation, Singapore, and Cyber Security Agency of Singapore under its National Cybersecurity R\&D Programme (Fuzz Testing <NRF-NCR25-Fuzz-0001>). Any opinions, findings and conclusions, or recommendations expressed in this material are those of the author(s) and do not reflect the views of National Research Foundation, Singapore, and Cyber Security Agency of Singapore. 
Lasse M\o{}ldrup and Andreas Pavlogiannis were partially supported by a research grant (VIL42117) from VILLUM FONDEN, and by a research grant from STIBOFONDEN.
\end{acks}

\bibliographystyle{ACM-Reference-Format}
\bibliography{ref}

\newpage
\appendix

\section{Proofs for \secref{hardness-algo}}
\subsection{Proof for \secref{general-solution}}

\applabel{proof-frontier-algo}

\vchConsToReach*
\begin{proof}
For convenience, for a node $v = \tuple{Y, Q, I}$ in $\gfrontier$, 
we use $v.Y, v.Q, v.I$ to denote $Y, Q, I$. 
We prove each direction separately.

\myparagraph{Correctness(Reachability $\Rightarrow$ Consistency)}
We now show if there exists a sink node $v = \tuple{\eventSet, Q, \bot}$ for some $Q$ and $v$ is reachable from the source node $v' = \tuple{\emptyset, \lambda~\ch. \epsilon, \bot}$ in $\gfrontier$, 
then there is a concretization of $\tuple{\AbstractExecution, \cpFunc}$. 
Let $\pi$ be a path from $v'$ to $v$ in $\gfrontier$. 
We directly give out the concretization $\tr$ of $\tuple{\AbstractExecution, \cpFunc}$ as the sequence of labelling events corresponding to $\pi$. 

We now show $\tr$ is indeed a valid concretization. 
Firstly, $\events{\trace} = \eventSet$, 
because each edge $v_1 \xrightarrow{e} v_2$ guarantees $e \notin v_1.Y$ and $v_2.Y = v_1.Y \cup \set{e}$. 
Since we start from $v'.Y = \varnothing$ and end at $v.Y = \eventSet$, 
the path $\pi$ must contain all events in $\eventSet$. 
Therefore, we have $\events{\trace} = \eventSet$. 

Secondly, $\trace$ satisfies $\po{}$, 
otherwise if $e_1 \trord{\tr} e_2$ and $e_2$ is program ordered before $e_1$, 
then the event set $Y$ extended by $e_1$ is not $\po{}$-closed, 
which violates our definition for the node. 

Thirdly, every receive operation should observe a send operation with the same value, 
and this property is already captured when we define the edges of $\gfrontier$. 

Lastly, $\trace$ should meet the capacity constraints. 
For synchronous channels, 
we already guarantee that no events can execute between and send and the its corresponding receiver on a synchronous channel, 
because send and receive to any channel can execute only when $I \pointsTo \bot$, 
so that it's impossible for send or receive events on other synchronous channels to interleave. 
The capacity constraints for asynchronous channels are also met, 
because when we define the edges of $\gfrontier$, 
an arbitrary node $u = \tuple{Y_u, Q_u, I_u}$ can only be extended by a send event on channel $\ch$, 
if the number of buffered send events to $\ch$ in $Y_u$ is below $\cp{\ch}$. 
With all these observations combined, $\trace$ is indeed a correct concretization. 

\myparagraph{Correctness(Consistency $\Rightarrow$ Reachability)}
We now show if there is a concretization $\trace$ of $\tuple{\AbstractExecution, \cpFunc}$, 
then there exists a sink node $v = \tuple{\eventSet, Q, \bot}$ for some $Q$ and $v$ is reachable from the source node $v' = \tuple{\emptyset, \lambda~\ch. \epsilon, \bot}$ in $\gfrontier$. 
We can start from the source node $v'$, 
and in the $i$-th step, 
we just extend current node by an edge, 
which is labelled by the $i$-th event in $\trace$. 
By the definition of edges in the frontier graph, 
every step of extension is allowed. 
Moreover, since $\trace$ is a valid concretization, 
then $\events{\trace} = \eventSet$ and thus this path ends at a node whose event set is exactly $\eventSet$.
Therefore, the correctness is guaranteed. 
\end{proof}



 


\subsection{Proof for \secref{vch-rf-sync-solution-sec}}

\vchRfSyncConsToCyclic*
\begin{proof}
We prove each direction separately.

\myparagraph{Consistency $\Rightarrow$ Acyclicity}
Suppose $\gsync$ has a cycle, then $\tuple{\AbstractExecution, \cpFunc, \rf{}}$ is not consistent.
We assume there is a cycle in $\gsync$, 
which contains two nodes $\tuple{\snd_1, \rcv_1}$ and $\tuple{\snd_2, \rcv_2}$. 
In any concretization $\tr$ of $\tuple{\AbstractExecution, \cpFunc, \rf{}}$, 
we must have $\snd_1 \trord{\tr} \rcv_1 \trord{\tr} \snd_2 \trord{\tr} \rcv_2$, 
because there is a path from $\tuple{\snd_1, \rcv_1}$ to $\tuple{\snd_2, \rcv_2}$. 
Similarly, we have $\snd_2 \trord{\tr} \rcv_2 \trord{\tr} \snd_1 \trord{\tr} \rcv_1$, 
because there is a path from $\tuple{\snd_2, \rcv_2}$ to $\tuple{\snd_1, \rcv_1}$. 
No total order $\trord{\tr}$ can satisfy both requirements. 

\myparagraph{Acyclicity $\Rightarrow$ Consistency}
Suppose $\gsync$ is acyclic.
Consider an arbitrary topological sort $\pi = \tuple{\snd_1, \rcv_1} \cdot \tuple{\snd_2, \rcv_2}$ $\cdots \tuple{\snd_k, \rcv_k}$ of $\gsync$
and using it, define $\tr = \snd_1 \cdot \rcv_1 \cdots \snd_k \cdot \rcv_k$.
We will show that this $\tr$ defined is a concretization of $\tuple{\AbstractExecution, \cpFunc, \rf{}}$.
First, $\tr$ respects $\rf{}$, because a send event is immediately followed by its receiver. 
We now prove it also satisfies $\po{}$. 
Assume on the contrary that this is not the case.
Then, there are two events $e, e'$, s.t. $(e, e') \in \po{}$, 
but  $e' \trord{\tr} e$.
First, $e$ and $e'$ cannot be matching send-receive events since 
they belong to the same thread.
Let $e_\snd$ be either $e$ if $\OpOf{e} = \snd$, and $\rf{}(e)$ otherwise,
and let $e_\rcv = \rf{}(e_\snd)$.
Likewise, let $e'_\snd$ be either $e'$ if $\OpOf{e'} = \snd$, and $\rf{}(e')$ otherwise,
and let $e'_\rcv = \rf{}(e'_\snd)$.
Also, by virtue of how $\tr$ was constructed, there is no path from $\tuple{e_\snd, e_\rcv}$ 
to $\tuple{e'_\snd, e'_\rcv}$ in $\gsync$; or else we will have $e \trord{\tr} e'$.
But this is a contradiction since $\gsync$ must add an edge
from $\tuple{e_\snd, e_\rcv}$ to $\tuple{e_\snd, e_\rcv}$ because $(e, e') \in \po{}$.
\end{proof}


\subsection{Proof for \secref{tree-topo-upper-bound}}

\applabel{proof-tree-topology}

\acyclictopologycompositionality*
\begin{proof}
We prove each direction separately.

\myparagraph{Correctness ($\tuple{\AbstractExecution, \cpFunc, \rf{}}$ $\Rightarrow$ $\tuple{\proj{\AbstractExecution}{\thread_i, \thread_j}, \proj{\cpFunc}{\thread_i, \thread_j}, \proj{\rf{}}{\thread_i, \thread_j}}$)}
If $\tuple{\AbstractExecution, \cpFunc, \rf{}}$ is consistent, 
then $\tuple{\proj{\AbstractExecution}{\thread_i, \thread_j}, \proj{\cpFunc}{\thread_i, \thread_j}, \proj{\rf{}}{\thread_i, \thread_j}}$ must be consistent for any $({\thread_i, \thread_j}) \in E$. 
Otherwise, assuming $\tuple{\proj{\AbstractExecution}{\thread_i, \thread_j}, \proj{\cpFunc}{\thread_i, \thread_j}, \proj{\rf{}}{\thread_i, \thread_j}}$ is not consistent, 
any concretization $\trace$ of $\tuple{\AbstractExecution, \cpFunc, 
\rf{}}$ will not be consistent, 
because $\trace$ is not a valid concretization for thread ${\thread_i, \thread_j}$. 
It contradicts with the fact that $\tuple{\AbstractExecution, \cpFunc, \rf{}}$ is consistent. 

\myparagraph{Correctness ($\tuple{\proj{\AbstractExecution}{\thread_i, \thread_j}, \proj{\cpFunc}{\thread_i, \thread_j}, \proj{\rf{}}{\thread_i, \thread_j}}$ $\Rightarrow$ $\tuple{\AbstractExecution, \cpFunc, \rf{}}$)}
Let $G$ be the topology graph of the input $\vchRf$ instance. 
To construct the concretization $\trace$ for $\tuple{\AbstractExecution, \cpFunc, \rf{}}$, 
we define a graph $G' = \tuple{V', E'}$, 
where $V'$ is the set of all events in $\eventSet$. 
We have $(e_1, e_2) \in E'$, 
iff $(e_1, e_2) \in \po{}$ or $x_{e_1, e_2}$ exists in a 2SAT formula for some thread ${\thread_1, \thread_2}$ and $x_{e_1, e_2}$ is assigned true. 
We claim $G'$ is acyclic and any linearization of $G'$ is a valid concretization. 

If $G'$ is cyclic, then we pick an arbitrary cycle $C$. 
Assuming the size of $C$ is $r$, 
then let $C = e_{c_1} \rightarrow \dots \rightarrow e_{c_r} \rightarrow e_{c_1}$. 
We look at the threads of events in $C$, 
i.e. $\ThreadOf{e_{c_1}}, \dots, \ThreadOf{e_{c_r}}, \ThreadOf{e_{c_1}}$. 
Clearly, each pair of adjacent threads in this sequence shares at least one common channel. 
Since the topology graph $G$ is acyclic, 
we claim there are at most two distinct threads among the threads of all events in $C$,  
as otherwise, $G$ has a cycle. 
All events in the same thread are already ordered by $\po{}$, 
so that $C$ cannot contain events only from one thread. 
Therefore, $C$ contains exactly two threads (say $\thread_1, \thread_2$). 
Note that $e_i \rightarrow e_j$ is an edge in $G'$ iff $e_i, e_j$ are either in the same thread or access one of the common channels between $\ThreadOf{e_i}, \ThreadOf{e_j}$. 
This means $e_{c_1}, \dots, e_{c_r}$ all access the common channels between ${\thread_1, \thread_2}$. 
However, this contradicts with the fact that $\tuple{\proj{\AbstractExecution}{\thread_1, \thread_2}, \proj{\cpFunc}{\thread_1, \thread_2}, \proj{\rf{}}{\thread_1, \thread_2}}$ is consistent, 
so that $G'$ must be acyclic.

Now we show an arbitrary linearization $\trace$ of $G'$ is a valid concretization. 
$\trace$ respects $\po{}$ and $\rf{}$, 
because if $(e_1, e_2) \in (\po{} \cup \rf{})$, then $(e_1, e_2)$ must be an edge in $G'$, 
so that $e_1 \trord{\trace} e_2$. 
The capacity constrains, $\fifo$ property are already taken care of in each sub-instance, 
because a channel is at most accessed by two threads. 
Therefore, $\trace$ is indeed a valid concretization of $\tuple{\AbstractExecution, \cpFunc, \rf{}}$ and it is indeed consistent. 
\end{proof}

\twosatcorrectness*
\begin{proof}
We prove each direction separately.

\myparagraph{Correctness (Satisfiability $\Rightarrow$ Consistency)}
Now assuming there $\form_{\tuple{\AbstractExecution, \cpFunc, \rf{}}}$ can be satisfied, 
then the input $\vchRf$ instance $\tuple{\AbstractExecution, \cpFunc, \rf{}}$ is consistent and we sketch one concretization $\trace$ as following. 
For every event pair $(e, f)$, 
if $x_{e, f}$ is true, then we order $e$ before $f$ in $\trace$. 
Firstly, $\trace$ is indeed a linear trace, 
because $\form_{\exactlyOne}$ guarantees $x_{e, f} = \neg x_{f, e}$, 
so that for events from different thread,
we have a unique assignment for their relative ordering in $\trace$. 
Also, 
$\form_{\trans}$ guarantees that the transitivity of events orderings is taken care of. 
That is if $e_1 \trord{\trace} e_2$ and $e_2 \trord{\trace} e_3$, 
then $e_1 \trord{\trace} e_3$. 
To see this, we enumerate all 4 possible situations. 
\begin{enumerate}
    \item $\ThreadOf{e_1} = \ThreadOf{e_2} = \ThreadOf{e_3}$. 
    This implies $(e_1, e_2), (e_2, e_3) \in \po{}$, 
    so that $(e_1, e_3) \in \po{}$ and thus $e_1 \trord{\trace} e_3$. 

    \item $\ThreadOf{e_1} = \ThreadOf{e_2} \neq \ThreadOf{e_3}$. 
    Since $(e_1, e_2) \in \po{}$, 
    $\form_{\po{}}$ guarantees $e_1 \trord{\trace} \pred{}{e_2}$, 
    and $\form_{\trans}$ guarantees $\pred{}{e_2} \trord{\trace} e_3$. 
    Therefore, $e_1 \trord{\trace} e_3$.

    \item $\ThreadOf{e_1} \neq \ThreadOf{e_2} = \ThreadOf{e_3}$. 
    Since $(e_2, e_3) \in \po{}$, 
    $\form_{\po{}}$ guarantees $\sucr{}{e_2} \trord{\trace} e_3$, 
    and $\form_{\trans}$ guarantees $e_1 \trord{\trace} \sucr{}{e_2}$. 
    Therefore, $e_1 \trord{\trace} e_3$.

    \item $\ThreadOf{e_1} = \ThreadOf{e_3} \neq \ThreadOf{e_2} \neq $. 
    Since $e_2 \trord{\trace} e_3$, 
    by transitivity we have
    $e_2 \trord{\trace} \sucr{}{e_3} $. 
    If $e_3 \trord{\trace} e_1$, 
    then $(e_3, e_1) \in \po{}$, 
    and we would have $e_2 \trord{\trace} e_3 \trord{e_1}$, 
    which contradicts with the fact that $e_1 \trord{\trace} e_2$. 
\end{enumerate}
Therefore, $\trace$ cannot be cyclic. 

Secondly, 
the capacity constraints are also met, 
because a channel $\ch$ is capacity-unbounded, capacity 1 or capacity 0. 
If $\ch$ is unbounded, 
then the capacity constraints are already satisfied 
and $\form_{\capOne}, \form_{\sync}$ are designed to satisfy the capacity constraints for channels with capacity 1 or 0. 

Now we argue $\trace$ also respects $\po{}$ and $\rf{}$. 
$\trace$ respects $\po{}$, 
because if $(e_1, e_2) \in \po{}$, 
then $x_{e_1, e_2}$ must be true, 
so that $e_1$ appears earlier than $e_2$ in $\trace$. 
On the other hand, $\trace$ respects $\rf{}$, because $\form_{\rf{}}$ orders all send events before their receive events. 
$\form_{\fifo}$ guarantees for each channel $\ch$, $\snd_1(\ch)$ is before $\snd_2(\ch)$, iff $\rcv_1(\ch)$ is before $\rcv_2(\ch)$, 
where $(\snd_i(\ch), \rcv_i(\ch)) \in \rf{}$ for $i = 1, 2$. 
Therefore, $\rf{}$ is also satisfied. 
We have so far proved $\trace$ is indeed a concretization of $\tuple{\AbstractExecution, \cpFunc, \rf{}}$ and therefore $\tuple{\AbstractExecution, \cpFunc, \rf{}}$ is consistent. 

\myparagraph{Correctness (Consistency $\Rightarrow$ Satisfiability)}
If $\tuple{\AbstractExecution, \cpFunc, \rf{}}$ is consistent, 
then we pick an arbitrary concretization $\trace$. 
We assign $x_{e, f}$ to be true and $x_{f, e}$ to be false, 
iff in $\trace$, $e$ is ordered before $f$.
We now show this assignment satisfies $\form_{\tuple{\AbstractExecution, \cpFunc, \rf{}}}$. 
Firstly, $\form_{\exactlyOne}$ is obviously satisfied, 
because we guarantee $x_{e, f} = \neg x_{f, e}$ 
by our assignments. 
Secondly, $\form_{\po{}}$ and $\form_{\rf{}}$ are satisfied, 
because $\trace$ must respect $\form_{\po{}}$ and $\rf{}$ relation. 
Thirdly, $\trace$ guarantees the $\fifo$ property of each channel $\ch$, 
and for each channel $\ch$, 
$\trace$ orders all send events to $\ch$ with no receivers after all send events to $\ch$ with a receiver. 
Otherwise, these unmatched send events will block the channel. 
Therefore, $\form_{\fifo}$ and $\form_{\unmatched}$ are satisfied. 
$\form_{\trans}$ is also satisfied, 
because if $e \le_{\s{tr}}^\trace f$, then 
(1) $e' = \s{pred}(e)$ is ordered before $f$, 
and (2) $f' = \sucr{}{f}$ is ordered after $e$, 
otherwise, 
$\trace$ is not a linear trace. 
Finally, $\form_{\capOne}$ and $\form_{\sync}$ are satisfied, 
because $\trace$ satisfies the capacity constraints. 
\end{proof}

\section{Lower bounds of $\vch$}
\applabel{lower-bounds-vch}
We now turn our attention to the hardness of $\vch$. 
In \secref{atomic-gadget}, we introduce atomicity gadgets, which is a construction to ensure a sequence of events to be executed atomically, and will be frequently used in later sections. 
In \secref{lower-vch-sameValues} we prove \thmref{vch-same-value-hardness}, namely that the problem is intractable even when all send/receive events use the same value.
In \secref{hardness-vch-two-threads} we prove \thmref{vch-two-threads-hardness}, stating that hardness for $\vch$ arises already with $2$ threads, and even if there are no capacity constraints on the channels.
Finally, in \secref{lower-vch-chanel=1-cap=0-or-cap=1} we prove \thmref{vch-one-chan-hardness}, showing that the problem is also hard already with $1$ channel.



\subsection{Hardness with Same Values}
\seclabel{lower-vch-sameValues}

We consider the $\vch$ problem for instances where all events (no matter what channel they access) send/receive the same value.
We remark that, in the case of shared memory, this problem is known to  be solvable in linear time ---- simply check if, for each memory location $x$, there is some thread that writes to $x$ before reading from it.
In the case of channels, $\fifo$ and capacity constraints turn this problem intractable, as we prove here.

\myparagraph{Overview}{
Our proof is via a reduction from the Hamiltonian cycle problem on an directed graph $G$.
Given $G$ with $\numNode$ nodes, we construct a $\vch$ instance $\tuple{\AbstractExecution, \cpFunc}$
which is consistent iff $G$ has a Hamiltonian cycle.
In high level, $\tuple{\AbstractExecution, \cpFunc}$ is constructed so that any concretization of $\tuple{\AbstractExecution, \cpFunc}$ can be conceptually split into three phases, based on the following scheme.
The initial phase picks an arbitrary node $v_1$ as the start of the Hamiltonian cycle,
and also sends $\numNode$ messages to a channel, which act as a counter to keep track of the length of the Hamiltonian cycle constructed in the next phase.
The second phase guesses the Hamiltonian cycle edge-by-edge while decrementing the counter and also ensuring no node repeats.
The last phase executes residual send/receive events and verifies that the sequence of edges guessed in the second phase is indeed a Hamiltonian cycle.
See \figref{push-pop-reduction} for an illustration on a small example.
}


\begin{figure}[t]
\centering
\newcommand{\xstep}{1.7}
\newcommand{\ystep}{-0.7}
\newcommand{\height}{0.5}
\newcommand{\wid}{1.5}

\centering
\scalebox{0.83}{
\begin{tikzpicture}[font=\small]

\begin{scope}[shift={(-1.5*\xstep,5*\ystep)}]
\node [graphNode] (n1) at (0, 0) {$u$};
\node [graphNode] (n2) at (0, -1.5) {$v$};
\node [graphNode] (n3) at (0, -3) {$w$};
\draw [->, very thick, ] (n1) to (n2);
\draw [->, very thick, ] (n2) to [bend left] (n3);
\draw [->, very thick, ] (n3) to [bend left] (n1);
\draw [->, very thick, left] (n3) -- (n2);
\node[text width=3cm, align=center] at (0, -4.5) {A graph $G$ with Hamiltonian cycle $u \to v \to w \to u$};
\node[fill=white, minimum height=0.2cm] at (-1, 0) {};
\end{scope}

\node[event] (n1) at (0 * \xstep, 0 * \ystep) {$\snd({\color{red} \lk})$};
\node[event] (n1) at (0 * \xstep, 1 * \ystep) {$\snd(\s{cnt})$};
\node[event] (n1) at (0 * \xstep, 2 * \ystep) {$\snd(\s{cnt})$};
\node[event] (n1) at (0 * \xstep, 3 * \ystep) {$\snd(\s{cnt})$};
\node[event] (n1) at (0 * \xstep, 4 * \ystep) {$\snd(\ch_u)$};
\node[event] (n1) at (0 * \xstep, 5 * \ystep) {$\snd(\ch'_u)$};
\node[event] (n1) at (0 * \xstep, 6 * \ystep) {$\snd(\ch'_v)$};
\node[event] (n1) at (0 * \xstep, 7 * \ystep) {$\snd(\ch'_w)$};
\node[event] (n1) at (0 * \xstep, 8 * \ystep) { $\rcv({\color{red} \lk})$};

\begin{scope}[on background layer]
\draw [thread] (0 * \xstep, -0.6 * \ystep) -- (0 * \xstep, 13.8 * \ystep);
\end{scope}

\node[threadName] at (0 * \xstep, -1.0 * \ystep) { $\thread_{\sf init}$};
    
\node[event] (n1) at (1 * \xstep, 0 * \ystep) { $\snd({\color{red} \lk})$};
\node[event] (n1) at (1 * \xstep, 1 * \ystep) { $\rcv(\ch_u)$};
\node[event] (n1) at (1 * \xstep, 2 * \ystep) {$\rcv(\s{cnt})$};
\node[event] (n1) at (1 * \xstep, 3 * \ystep) {$\rcv(\ch'_v)$};
\node[event] (n1) at (1 * \xstep, 4 * \ystep) {$\snd(\ch_v)$};
\node[event] (n1) at (1 * \xstep, 5 * \ystep) { $\rcv({\color{red} \lk})$};

\begin{scope}[on background layer]
\draw [thread] (1 * \xstep, -0.6 * \ystep) -- (1 * \xstep, 13.8 * \ystep);
\end{scope}

\node[threadName] at (1 * \xstep, -1.0 * \ystep) {$\thread_{(u,v)}$};

\begin{scope}[shift={(-0.25*\xstep, 0)}]
\node[event] (n1) at (2 * \xstep, 6 * \ystep) {$\snd({\color{red} \lk})$};
\node[event] (n1) at (2 * \xstep, 7 * \ystep) {$\rcv(\ch_v)$};
\node[event] (n1) at (2 * \xstep, 8 * \ystep) {$\rcv(\s{cnt})$};
\node[event] (n1) at (2 * \xstep, 9 * \ystep) {$\rcv(\ch'_w)$};
\node[event] (n1) at (2 * \xstep, 10 * \ystep) {$\snd(\ch_w)$};
\node[event] (n1) at (2 * \xstep, 11 * \ystep) { $\rcv({\color{red} \lk})$};

\begin{scope}[on background layer]
\draw [thread] (2 * \xstep, -0.6 * \ystep) -- (2 * \xstep, 13.8 * \ystep);
\end{scope}
\node[threadName] at (2 * \xstep, -1.0 * \ystep) {$\thread_{(v,w)}$};
\end{scope}
    
\begin{scope}[shift={(-0.5*\xstep, 0)}]
\node[event] (n1) at (3 * \xstep, 0 * \ystep) { $\snd({\color{red} \lk})$};
\node[event] (n1) at (3 * \xstep, 1 * \ystep) { $\rcv(\ch_w)$};
\node[event] (n1) at (3 * \xstep, 2 * \ystep) {$\rcv(\s{cnt})$};
\node[event] (n1) at (3 * \xstep, 3 * \ystep) {$\rcv(\ch'_u)$};
\node[event] (n1) at (3 * \xstep, 4 * \ystep) {$\snd(\ch_u)$};
\node[event] (n1) at (3 * \xstep, 5 * \ystep) { $\rcv({\color{red} \lk})$};

\begin{scope}[on background layer]
\draw [thread] (3 * \xstep, -0.6 * \ystep) -- (3 * \xstep, 13.8 * \ystep);
\end{scope}
\node[threadName] at (3 * \xstep, -1.0 * \ystep) {$\thread_{(w,u)}$};
\end{scope}

\begin{scope}[shift={(-0.75*\xstep, 0)}]
\node[event] (n1) at (4 * \xstep, 6 * \ystep) { $\snd({\color{red} \lk})$};
\node[event] (n1) at (4 * \xstep, 7 * \ystep) { $\rcv(\ch_w)$};
\node[event] (n1) at (4 * \xstep, 8 * \ystep) {$\rcv(\s{cnt})$};
\node[event] (n1) at (4 * \xstep, 9 * \ystep) {$\rcv(\ch'_v)$};
\node[event] (n1) at (4 * \xstep, 10 * \ystep) {$\snd(\ch_v)$};
\node[event] (n1) at (4 * \xstep, 11 * \ystep) { $\rcv({\color{red} \lk})$};

\begin{scope}[on background layer]
\draw [thread] (4 * \xstep, -0.6 * \ystep) -- (4 * \xstep, 13.8 * \ystep);
\end{scope}
\node[threadName] at (4 * \xstep, -1.0 * \ystep) {$\thread_{(w,v)}$};
\end{scope}

\begin{scope}[shift={(-0.75*\xstep, 0)}]
\node[event] (n1) at (5 * \xstep, 0 * \ystep) { $\snd({\color{red} \lk})$};
\node[event] (n1) at (5 * \xstep, 1 * \ystep) {$\snd(\s{cnt})$};
\node[event] (n1) at (5 * \xstep, 2 * \ystep) {$\snd(\s{cnt})$};
\node[event] (n1) at (5 * \xstep, 3 * \ystep) {$\snd(\s{cnt})$};
\node[event] (n1) at (5 * \xstep, 4 * \ystep) {$\snd(\s{cnt})$};
\node[event] (n1) at (5 * \xstep, 5 * \ystep) {$\rcv(\ch_u)$};
\node[event] (n1) at (5 * \xstep, 6 * \ystep) {$\rcv(\s{cnt})$};
\node[event] (n1) at (5 * \xstep, 7 * \ystep) { $\rcv(\s{cnt})$};
\node[event] (n1) at (5 * \xstep, 8 * \ystep) {$\rcv(\s{cnt})$};
\node[event] (n1) at (5 * \xstep, 9 * \ystep) {$\rcv(\s{cnt})$};
\node[event] (n1) at (5 * \xstep, 10 * \ystep) { $\rcv({\color{red} \lk})$};
\node[event] (n1) at (5 * \xstep, 11 * \ystep) {$\snd(\alpha)$};
\node[event] (n1) at (5 * \xstep, 12 * \ystep) {$\snd(\alpha)$};
\node[event] (n1) at (5 * \xstep, 13 * \ystep) {$\snd(\alpha)$};

\begin{scope}[on background layer]
\draw [thread] (5 * \xstep, -0.6 * \ystep) -- (5 * \xstep, 13.8 * \ystep);
\end{scope}

\node[threadName] at (5 * \xstep, -1.0 * \ystep) {$\thread_{\sf free}$};
\end{scope}

\begin{scope}[shift={(-0.75*\xstep, 0)}]
\node[event] (n1) at (6 * \xstep, 0 * \ystep) {$\rcv(\alpha)$};
\node[event] (n1) at (6 * \xstep, 1 * \ystep) {$\snd(\ch_u)$};
\node[event] (n1) at (6 * \xstep, 2 * \ystep) {$\snd(\ch'_v)$};
\node[event] (n1) at (6 * \xstep, 3 * \ystep) {$\snd(\s{cnt})$};
\begin{scope}[on background layer]
\draw [thread] (6 * \xstep, -0.6 * \ystep) -- (6 * \xstep, 13.8 * \ystep);
\end{scope}
\node[threadName] at (6 * \xstep, -1.0 * \ystep) {$\thread_{u}$};
\end{scope}

\begin{scope}[shift={(-1*\xstep, 0)}]
\node[event] (n1) at (7 * \xstep, 7 * \ystep) {$\rcv(\alpha)$};
\node[event] (n1) at (7 * \xstep, 8 * \ystep) {$\snd(\ch_v)$};
\node[event] (n1) at (7 * \xstep, 9 * \ystep) {$\snd(\ch'_w)$};
\node[event] (n1) at (7 * \xstep, 10 * \ystep) {$\snd(\s{cnt})$};

\begin{scope}[on background layer]
\draw [thread] (7 * \xstep, -0.6 * \ystep) -- (7 * \xstep, 13.8 * \ystep);
\end{scope}

\node[threadName] at (7 * \xstep, -1.0 * \ystep) {$\thread_{v}$};
\end{scope}

\begin{scope}[shift={(-1.25*\xstep, 0)}]
\node[event] (n1) at (8 * \xstep, 0 * \ystep) {$\rcv(\alpha)$};
\node[event] (n1) at (8 * \xstep, 1 * \ystep) {$\snd(\ch_w)$};
\node[event] (n1) at (8 * \xstep, 2 * \ystep) {$\snd(\ch'_u)$};
\node[event] (n1) at (8 * \xstep, 3 * \ystep) {$\snd(\s{cnt})$};
\node[event] (n1) at (8 * \xstep, 4 * \ystep) {$\snd(\ch_w)$};
\node[event] (n1) at (8 * \xstep, 5 * \ystep) {$\snd(\ch'_v)$};
\node[event] (n1) at (8 * \xstep, 6 * \ystep) {$\snd(\s{cnt})$};

\begin{scope}[on background layer]
\draw [thread] (8 * \xstep, -0.6 * \ystep) -- (8 * \xstep, 13.8 * \ystep);
\end{scope}

\node[threadName] at (8 * \xstep, -1.0 * \ystep) {$\thread_{w}$};
\end{scope}

\end{tikzpicture}
}

\caption{
A graph $G$ with a Hamiltonian cycle (left) and the corresponding $\vch$ instance in which all send/receive events use the same value (right). 
A concretization is $\trace = [{\color{black} \thread_{\s{init}}}] \cdot 
[{\color{black} \thread_{u, v}} \cdot 
{\color{black} \thread_{v, w}} \cdot 
{\color{black} \thread_{w, u}}] \cdot 
[{\color{black} \thread_{\s{free}}} \cdot 
{\color{black} \thread_{u}} \cdot 
{\color{black} \thread_{v}} \cdot 
{\color{black} \thread_{w}} \cdot 
{\color{black} \thread_{w, v}}$],
obtained by fully executing every thread according to this sequence. 
Brackets separate the three phases.
}
\figlabel{push-pop-reduction}
\end{figure}




\myparagraph{Reduction}{
We now make the above idea formal.
Let $G = (V, E)$ be the instance of the Hamiltonian cycle problem with $\numNode$ nodes and $\numEdge$ edges.
We construct $\vch$ instance $\tuple{\AbstractExecution, \cpFunc}$ that uses
$\numEdge + \numNode + 2$ threads $\setpred{\thread_{(u, v)}}{(u, v) \in E} \cup \setpred{\thread_v}{v \in V} \cup \set{\thread_{\sf init}, \thread_{\sf free}}$, and $2\numNode + 3$ channels $\setpred{\ch_v, \ch'_v}{v \in V} \cup \set{\lk, \alpha, \chanCnt}$,
with the following capacities:
$\cp{\ch_{v}} = \outDegree{v} + \inDegree{v}$, $\cp{\ch'_{v}} = \inDegree{v}$, 
$\cp{\chanCnt} = \numEdge$, $\cp{\lk} = 1$, $\cp{\alpha} = \numNode$. 
We use $\outDegree{v}$ and $\inDegree{v}$ to denote the  out-degree and in-degree of a node $v \in V$. 
For each $v \in V$, the sequence of events in thread $\thread_v$ comprises $\outDegree{v_i}$ blocks.
\begin{align*}
\thread_v = \pop(\alpha) \cdot A^1_v \cdots A^{\outDegree{v}}_v
\end{align*}
The $j^\text{th}$ block $A^j_v$ encodes the $j^\text{th}$ outgoing edge $(v, w)$ from $v$.
 \begin{align*}
A^j_v = \push(\ch_v) \cdot \push(\ch'_w) \cdot \push(\chanCnt)
\end{align*}
For each edge $(u, v) \in E$, the sequence of events in $\thread_{(u, v)}$ is:
\begin{align*}
\thread_{(u, v)} = \push(\s{\lk}) \cdot  \pop(\ch_u) \cdot \pop(\chanCnt) \cdot \pop(\ch'_v) \cdot \push(\ch_v) \cdot \pop(\s{\lk}) 
\end{align*}
The thread $\thread_{\s init}$ sends a message to the channel $\ch_v$ of a designated initial node $v$ of the Hamiltonian cycle, sends $\numNode$ messages to the channel $\chanCnt$, and also sends one message to each channel $\set{\ch'_u}_{u \in V}$.
After the Hamiltonian cycle has been constructed,  the thread $\thread_{\s{free}}$, together with $\set{\thread_{u}}_{u \in V}$ empties all channels while ensuring that $v$ is reached by a path that visits every node once.
\begin{align*}
\begin{array}{rcl}
\thread_{\sf init} &=& \push(\s{\lk}) \cdot \push_1(\chanCnt) \cdots \push_{\numNode}(\chanCnt) \cdot \push(\ch_{v_1}) \cdot \push(\ch'_{v_1}) \cdots \push(\ch'_{v_\numNode}) \cdot \pop(\s{\lk}) \\
\thread_{\sf free} &=& \push(\s{\lk}) \cdot \push_1(\chanCnt) \cdots \push_{\numEdge}(\chanCnt) \cdot \pop(\ch_{v_1}) \cdot \pop_1(\chanCnt) \cdots \pop_{\numEdge}(\chanCnt) \cdot  \pop(\s{\lk}) \\
&&\cdot \push_1(\alpha) \cdots \push_{\numNode}(\alpha)
\end{array}
\end{align*}
The following lemma states the correctness of the construction.
\begin{restatable}{lemma}{samevaluereductioncorrectness}
\lemlabel{same-value-reduction-correctness}
$\tuple{\AbstractExecution, \cpFunc}$ is consistent iff $G$ contains a Hamiltonian cycle.
\end{restatable}

Finally, observe that the size of $\tuple{\AbstractExecution, \cpFunc}$ is $O(\numNode + \numEdge)$, thereby concluding the proof of \thmref{vch-same-value-hardness}.
}


\applabel{proof-np-same-value}

\samevaluereductioncorrectness*
\begin{proof}
We prove each direction separately.

\myparagraph{Proof of correctness(Hamiltonian cycle $\Rightarrow$ Consistency)}
Given a Hamiltonian cycle in $G$, 
assuming it is of form $v_1 \rightarrow v_2 \rightarrow \dots \rightarrow v_{\numNode} \rightarrow v_1$, 
we sketch the concretization $\trace$ as following. 
Here $\trace$ can be represented as a sequence of threads. 
That is, we execute all events in each thread according to the thread sequence. 
\begin{equation}
    \trace = A_{cycle} \circ A_1 \circ \dots \circ A_{\numNode}
    \nonumber
\end{equation}
where $A_{cycle}$ is a sequence of threads constructing the cycle 
and each $A_j$ is a sequence of threads executing the encoded events for unselected outgoing edges from $v_j$. 
\begin{equation}
    A_{cycle} = \thread_{\s{init}} \cdot \thread_{(v_1, v_2)} \cdot \thread_{(v_2, v_3)} \cdots \thread_{(v_{\numNode-1}, v_{\numNode})}, \thread_{(v_{\numNode}, v_1)}, \thread_{\s{free}}
    \nonumber
\end{equation}
We assume the outgoing edges from $v_j$ are $v_{j_1}, \dots, v_{j_q}$, where $q$ is the out-degree of $v_j$, 
and $v_{j_p}$ is an edge in the Hamiltonian cycle, 
then $A_j$ can be represented as following. 
\begin{equation}
    A_{j} = \thread_{v_j} \cdot \thread_{(v_j, v_{j_1})} \cdots \thread_{(v_j, v_{j_{p-1}})} \cdot  \thread_{(v_j, v_{j_{p+1}})} \cdots \thread_{(v_j, v_{j_{q}})}
    \nonumber
\end{equation}
Since every send/receive event has the same value, 
we mainly discuss the capacity constraints. 
In the first stage, the capacity constraints are clearly met. 
That is, the in-degree and out-degree of every node should clearly $\ge 1$, 
as otherwise $G$ has no Hamiltonian cycle. 
Therefore, it's perfectly fine to send once to $\ch'_v$ and $\ch_v$ in the first phase. 
Similarly, $\numEdge \ge \numNode$, 
as otherwise $G$ has no Hamiltonian cycle, 
so that it's also fine to send $\numNode$ times to $\s{cnt}$. 

In the second stage, we  receive once for every $\ch'_v$, 
and only receive $\ch_u$ if there is a message inside. 
We receive exactly $\numNode$ times for $\s{cnt}$. 
Therefore, the capacity constraints for the second stage is also met. 

Lastly, 
for the third phase, a channel $\ch'_u$ is sent at most $\inDegree{u}$ times, 
i.e., once by each $\thread_{v}$ for all $(v, u) \in E$. 
Moreover, 
a channel $\ch_u$ is sent at most $\inDegree{u} + \outDegree{u}$ times, 
i.e., once by each $\thread_{v}$ for all $(u, v) \in E$, 
and once by each $\thread_{w, u}$ for all $(w, u) \in E$. 
Therefore, the capacity constrains for the third stage is also met. 

\myparagraph{Proof of correctness(Consistency $\Rightarrow$ Hamiltonian cycle)}
In this direction, we prove if the input problem is consistent, 
then there is a Hamiltonian cycle.  
Firstly we argue that 
any concretization of this instance will order all events in $\thread_{v_i}$ after the $\rcv(\s{\lk})$ in $\thread_{\s{free}}$, 
because all $\thread_{v_i}$ starts with $\rcv(\alpha)$ and only $\thread_{\s{free}}$ sends to $\alpha$. 
Secondly we argue that 
the first thread to execute must be $\thread_{\s{init}}$. 
This is because all other threads start with receive events on some channels 
and all channels have 0 messages at the beginning. 
Therefore, so far we can conclude that in any concretization, 
the events executed between $\thread_{\s{init}}$ and $\thread_{\s{free}}$ must be from thread $\thread_{(v_i, v_j)}$ for some $i, j$, 
which encodes edge $v_i \rightarrow v_j$.

Also, as $\thread_{v_i}$, $\thread_{\s{init}}, \thread_{\s{free}}$ are composed of event blocks protected by $\snd(\lk)$ and $\rcv(\lk)$, 
we claim that none of these blocks can overlap with each other, 
because $\s{\lk}$ has capacity 1 (see the construction in \figref{lock-demo-capacity-one}). 
If any two blocks overlap, this means there are two continuous sends to $\s{\lk}$, 
which violates the capacity constraints. 
Finally, we can conclude that in any concretization $\trace$ of the instance, 
it must be of the following form. 
$\trace = \trace_1 \circ \trace_2$, 
where $\trace_1$ is a sequence of atomic blocks from $\thread_{\s{init}}, \thread_{\s{free}}$ or $\thread_{(v_i, v_j)}$ and $\trace_2$ is a sequence of other events not in $\trace_1$. 
Moreover $\trace_1$ is of the following form
\begin{equation}
    \trace_1 = \thread_{\s{init}} \circ \thread_{(1, p_1)} \circ \thread_{(p_1, p_2)}, \dots \circ \thread_{(p_{\numNode-2}, p_{\numNode-1})} \circ \thread_{(p_{\numNode-1}, 1)} \circ \thread^1_{\s{free}}
    \nonumber
\end{equation}
where $\thread^1_{\s{free}}$ is the atomic block in $\thread_{\s{free}}$. 
To verify the form of $\trace_1$, we have the following observations.
\begin{itemize}
    \item There must be exactly $\numNode$ atomic blocks between $\thread_{\s{init}}$ and $\thread^1_{\s{free}}$, because $\thread_{\s{init}}$ sends $\numNode$ times on $\s{cnt}$ and $\thread^1_{\s{free}}$ receives $\numEdge$ times on $\s{cnt}$. 
    Every thread $\thread_{(v_i, v_j)}$ receives once from $\s{cnt}$, 
    so that as $\cp{\s{cnt}} = \numEdge$, 
    there must be exactly $\numNode$ edge threads between $\thread_{\s{init}}$ and $\thread^1_{\s{free}}$. 
    Otherwise, the $\numEdge$ send events in $\thread^1_{\s{free}}$ cannot be executed. 

    \item If two threads between $\thread_{\s{init}}$ and $\thread^1_{\s{free}}$ are next to each other in $\trace_1$, 
    then they must share a common node. 
    For example, if $\thread_{(v_i, v_j)}$ is immediately before $\thread_{(v_p, v_q)}$ in $\trace_1$, then $j = p$. 
    To show this, $\thread_{(v_p, v_q)}$ will receive $\ch_{v_p}$ once and only those encoded threads for edges that end in $v_p$ will send once to $\ch_{v_p}$. 
    This observation proves that the threads between $\thread_{\s{init}}$ and $\thread^1_{\s{free}}$ correspond to a walk in the graph. 

    \item The first thread immediately after $\thread_{\s{free}}$ in $\trace_1$ must correspond to an edge starting from $v_1$, 
    because $\thread_{(v_i, v_j)}$ receives once to $\ch_{v_i}$ and after $\thread_{\s{free}}$ being executed, 
    only $\ch_{v_i}$ is not empty. 
    Also, the last thread immediately before $\thread^1_{\s{free}}$ must correspond to an edge ending in $v_1$, 
    because $\thread^1_{\s{free}}$ receives once to $\ch_{v_1}$ and only edge threads ending in $v_1$ will send $v_1$ once. 
    This observation proves that the threads between $\thread_{\s{init}}$ and $\thread^1_{\s{free}}$ correspond to a cycle in the graph $v_1 \rightarrow v_{p_1} \rightarrow v_{p_2} \rightarrow \dots \rightarrow v_{p_{\numNode-1}} \rightarrow v_1$. 

    \item Lastly, we show that every node will appear exactly once in the walk except $v_1$ appearing twice, 
    as it is both the starting and ending node of the walk. 
    We send once to every $\ch'_{v_i}$ in $\thread_{\s{init}}$, and every thread $\thread_{(v_i, v_j)}$ receives $\ch'_{v_j}$ once. 
    Therefore, we cannot have two edges in the walk that end in the same node $v_j$, 
    as there is only one message in $\ch'_{v_j}$. 
\end{itemize}
\end{proof}
Given the observations above, 
it is proved that the walk we found is indeed a Hamiltonian cycle. 



\subsection{Hardness with $2$ Threads and no Capacity Restrictions}
\seclabel{hardness-vch-two-threads}

\vchRf takes quadratic time when $\NumThreads=2$ and channels have unrestricted capacity or capacity $\leq 1$ (as per \thmref{vch-rf-tree-topology}).
Does this advantage of limiting threads carry over to $\vch$?
We show that this is not the case as $\vch$ remains $\NP$-hard when $\NumThreads=2$, even with no capacity restrictions.

\myparagraph{Overview}{
Our reduction is from positive 1-in-3 SAT, which takes as input a 3CNF formula $\psi$ for which every clause contains three distinct positive literals, and the task is to determine whether there is a truth assignment that makes exactly 1-in-3 literals true in each clause.
Given $\psi$, we construct a corresponding $\vch$ instance $\tuple{\AbstractExecution, \cpFunc}$, where events in $\AbstractExecution$ comprise two phases.
The first phase guesses a truth assignment for the propositional variables of $\psi$,
while the second phase verifies that every clause satisfies the 1-in-3 property.
}

\input{figures/VCH-two-threads-reduction}

\myparagraph{Reduction}{
Given a formula $\psi$ with $\numVar$ propositional variables and $\numClause$ clauses,
we construct a $\vch$ instance $\tuple{\AbstractExecution, \cpFunc}$ where $\AbstractExecution$ comprises $2$ threads $\thread_\top$ and $\thread_\bot$,
and $\numClause + 3$ capacity-unbounded channels 
$\set{\lk_1, \lk_2, \alpha, C_1, \dots, C_{\numClause}}$. 
\figref{reduction-two-threads-vch} illustrates the overall scheme.
For each $p \in \set{\top, \bot}$, the thread $\thread_p$ consists of two sequential phases,
corresponding to propositional variables and clauses on $\psi$.
\begin{align*}
\thread_p = A^p_1 \cdots A^p_{\numVar} \cdot B^p_1 \cdots B^p_{\numClause}
\end{align*}
The events in $A^p_i$ correspond to the $i^\text{th}$ variable $x_i$,
while the events in $B^p_j$ correspond to the $j^\text{th}$ clause $C_j$.
Let $C_{k_1}, C_{k_2}, \ldots, C_{k_{f_i}}$ be an ordered list of clauses in which variable $x_i$ appears.
The above sequences make use of the atomicity gadget (\figref{lock-demo-unbounded}), and are defined as follows.

\begin{align*}
\begin{array}{rcl}
A^\top_i &=& \!\!\!\!\!
\begin{aligned}
\begin{array}{l}
\;\snd(\alpha, v_i^3) \cdot \rcv(\alpha, v_i^4) \cdot \snd(\lk_1, v^1_i) \cdot 
\snd(\lk_2, v^1_i)) \cdot \rcv(\lk_2, v^1_i)) \\ 
\cdot \snd(C_{k_1}, \top)  \cdots \snd(C_{k_{f_i}}, \top) 
\cdot \rcv(\lk_1, v^1_i)
\end{array}
\end{aligned} \\
\vspace{0.1cm}
A^\bot_i &=& \!\!\!\!\!
\begin{aligned}
\begin{array}{l}
\;\snd(\alpha, v_i^4) \cdot \rcv(\alpha, v_i^3) \cdot \snd(\lk_2, v^2_i) \cdot \snd(\lk_1, v^2_i) \cdot \rcv(\lk_1, v^2_i) \\ \cdot \snd(C_{k_1}, \bot) \cdots \snd(C_{k_{f_i}}, \bot)
\cdot \rcv(\lk_2, v^2_i) 
\end{array}
\end{aligned} \\
\vspace{0.1cm}
B^\top_j &=& \!\!\!\!\!
\begin{aligned}
\begin{array}{l}
\;\snd(\alpha, w^4_j) \cdot \rcv(\alpha, w^5_j) \cdot \snd(\lk_1, w^1_j) \cdot \snd(\lk_2, w^1_j) \cdot \rcv(\lk_2, w^1_j) \\ \cdot \rcv(C_j, \top) \cdot \rcv(C_j, \bot)  
\cdot \rcv(\lk_1, w^1_j)
\end{array}
\end{aligned} \\
\vspace{0.1cm}
B^\bot_j &=& \!\!\!\!\!
\begin{aligned}
\begin{array}{l}
\;\snd(\alpha, w^5_j) \cdot \rcv(\alpha, w^4_j) \cdot \snd(\lk_2, w^2_j) \cdot \snd(\lk_1, w^2_j)) \cdot \rcv(\lk_1, w^2_j)) \cdot \rcv(C_j, \bot) \cdot \rcv(C_j, \top)  \\ 
\cdot \rcv(\lk_2, w^2_j))
\cdot \snd(\lk_2, w^3_j)) \cdot \snd(\lk_1, w^3_j) \cdot \rcv(\lk_1, w^3_j) \cdot \rcv(C_j, \bot) \cdot \rcv(C_j, \top) \cdot \rcv(\lk_2, w^3_j)
\end{array}
\end{aligned}
\end{array}
\end{align*}
where $v^1_i, v^2_i, v^3_i, v^4_i$ and  $w^1_j, w^2_j, w^3_j, w^4_j, w^5_j$ are distinct values associated with $i$ and $j$, 
guaranteeing atomicity and ensuring that the encoded events for different variables or clauses appear sequentially. 

\begin{restatable}{lemma}{vchTwoThreadsNPHard}
\lemlabel{vch-two-threads-np-hard}
$\tuple{\AbstractExecution, \cpFunc}$ is consistent iff $\psi$ is 1-in-3 satisfiable. 
\end{restatable}

Finally, our reduction takes $O(\numVar + \numClause)$ time, thereby concluding \thmref{vch-two-threads-hardness}.
}

\applabel{proof-np-t-k=unbounded}



\vchTwoThreadsNPHard*
\begin{proof}
We prove each direction separately.

    \myparagraph{Correctness (Satisfiability $\Rightarrow$ Consistency)} 
Given an assignment that satisfies $\psi$, 
we encode a concretization $\trace$ of $\tuple{\AbstractExecution, \cpFunc}$ as following. 
\begin{equation}
    \trace = A_1 \cdots A_n \cdot B_1 \cdots B_m
    \nonumber
\end{equation}
where $A_i$ is an interleaving of $A^\top_i, A^\bot_i$ 
and $B_j$ is an interleaving of $B^\top_j, B^\bot_j$. 
Moreover, if $x_i$ is assigned to be true, 
then 
\begin{equation}
    \begin{aligned}
        A_i &= \wt(\alpha, v_i^3) \cdot \wt(\alpha, v_i^4) \cdot \rd(\alpha, v_i^3) \cdot \rd(\alpha, v_i^4) \\
        & \cdot \wt(\lk_1, v_i^1) \cdot \wt(\lk_2, v_i^1) \cdot \rd(\lk_2, v_i^1) \cdot \wt(C_{i, 1}, {\top}) \cdots \wt(C_{i, k_i}, {\top}) \cdot \rd(\lk_1, v_i^1) \; (\s{from} \; \thread_{\top}) \\
        & \cdot \wt(\lk_2, v_i^2) \cdot \wt(\lk_1, v_i^2) \cdot \rd(\lk_1, v_i^2) \cdot \wt(C_{i, 1}, {\bot}) \cdots \wt(C_{i, k_i}, {\bot}) \cdot \rd(\lk_2, v_i^2) \; (\s{from} \; \thread_{\bot})
        \nonumber
    \end{aligned}
\end{equation}
Otherwise if $x_i$ is assigned to be false, 
then 
\begin{equation}
    \begin{aligned}
        A_i &= \wt(\alpha, v_i^3) \cdot \wt(\alpha, v_i^4) \cdot \rd(\alpha, v_i^3) \cdot \rd(\alpha, v_i^4) \\
        & \cdot \wt(\lk_2, v_i^2) \cdot \wt(\lk_1, v_i^2) \cdot \rd(\lk_1, v_i^2) \cdot \wt(C_{i, 1}, {\bot}) \cdots \wt(C_{i, k_i}, {\bot}) \cdot \rd(\lk_2, v_i^2) \; (\s{from} \; \thread_{\bot}) \\
        & \cdot \wt(\lk_1, v_i^1) \cdot \wt(\lk_2, v_i^1) \cdot \rd(\lk_2, v_i^1) \cdot \wt(C_{i, 1}, {\top}) \cdots \wt(C_{i, k_i}, {\top}) \cdot \rd(\lk_1, v_i^1) \; (\s{from} \; \thread_{\top})
        \nonumber
    \end{aligned}
\end{equation}
For $B_j$, we sort the variables in clause $C_j$ by the variable index. 
Then there are three possibilities, 
i.e., the variable assigned to be true can be the first, second or third variable in $C_j$. 
If it is the first one, then 
\begin{equation}
    \begin{aligned}
        B_j &= \wt(\alpha, w_j^4) \cdot  \wt(\alpha, w_j^5) \cdot \rd(\alpha, w_j^4) \cdot \rd(\alpha, w_j^5) \\
        & \cdot \wt(\lk_1, w_j^1) \cdot \wt(\lk_2, w_j^1) \cdot \rd(\lk_2, w_j^1) \cdot \rd(C_j, {\top}) \cdot \rd(C_j, {\bot}) \cdot \rd(\lk_1, w_j^1) \; (\s{from} \; \thread_{\top}) \\
        & \cdot \wt(\lk_2, w_j^2) \cdot \wt(\lk_1, w_j^2) \cdot \rd(\lk_1, w_j^2) \cdot \rd(C_j, {\bot}) \cdot \rd(C_j, {\top}) \cdot \rd(\lk_2, w_j^2) \; (\s{from} \; \thread_{\bot}) \\ 
        & \cdot \wt(\lk_2, w_j^3) \cdot \wt(\lk_1, w_j^3) \cdot \rd(\lk_1, w_j^3) \cdot \rd(C_j, {\bot}) \cdot \rd(C_j, {\top}) \cdot \rd(\lk_2, w_j^3) \; (\s{from} \; \thread_{\bot})
        \nonumber
    \end{aligned}
\end{equation}
If it is the second one, then 
\begin{equation}
    \begin{aligned}
        B_j &= \wt(\alpha, w_j^4) \cdot  \wt(\alpha, w_j^5) \cdot \rd(\alpha, w_j^4) \cdot \rd(\alpha, w_j^5) \\
        & \cdot \wt(\lk_2, w_j^2) \cdot \wt(\lk_1, w_j^2) \cdot \rd(\lk_1, w_j^2) \cdot \rd(C_j, {\bot}) \cdot \rd(C_j, {\top}) \cdot \rd(\lk_2, w_j^2) \; (\s{from} \; \thread_{\bot}) \\ 
        & \cdot \wt(\lk_1, w_j^1) \cdot \wt(\lk_2, w_j^1) \cdot \rd(\lk_2, w_j^1) \cdot \rd(C_j, {\top}) \cdot \rd(C_j, {\bot}) \cdot \rd(\lk_1, w_j^1) \; (\s{from} \; \thread_{\top}) \\
        & \cdot \wt(\lk_2, w_j^3) \cdot \wt(\lk_1, w_j^3) \cdot \rd(\lk_1, w_j^3) \cdot \rd(C_j, {\bot}) \cdot \rd(C_j, {\top}) \cdot \rd(\lk_2, w_j^3) \; (\s{from} \; \thread_{\bot})
        \nonumber
    \end{aligned}
\end{equation}
If it is the third one, then 
\begin{equation}
    \begin{aligned}
        B_j &= \wt(\alpha, w_j^4) \cdot  \wt(\alpha, w_j^5) \cdot \rd(\alpha, w_j^4) \cdot \rd(\alpha, w_j^5) \\
        & \cdot \wt(\lk_2, w_j^2) \cdot \wt(\lk_1, w_j^2) \cdot \rd(\lk_1, w_j^2) \cdot \rd(C_j, {\bot}) \cdot \rd(C_j, {\top}) \cdot \rd(\lk_2, w_j^2) \; (\s{from} \; \thread_{\bot}) \\ 
        & \cdot \wt(\lk_2, w_j^3) \cdot \wt(\lk_1, w_j^3) \cdot \rd(\lk_1, w_j^3) \cdot \rd(C_j, {\bot}) \cdot \rd(C_j, {\top}) \cdot \rd(\lk_2, w_j^3) \; (\s{from} \; \thread_{\bot}) \\
         & \cdot \wt(\lk_1, w_j^1) \cdot \wt(\lk_2, w_j^1) \cdot \rd(\lk_2, w_j^1) \cdot \rd(C_j, {\top}) \cdot \rd(C_j, {\bot}) \cdot \rd(\lk_1, w_j^1) \; (\s{from} \; \thread_{\top})
        \nonumber
    \end{aligned}
\end{equation}

The $\po{}$ and capacity constraints are clearly satisfied. 
We now argue the value constraints are also satisfied. 
For $\alpha, \lk_1, \lk_2$, the value constraints are satisfied in each $A_i, B_j$. 
Now for each $C_j$,  we consider the value sent to this channel, 
and claim it is one of the three situations, 
i.e., 
[$\top, \bot, \bot, \top, \bot, \top$] (first literal in $C_j$ is true), 
[$\bot, \top, \top, \bot, \bot, \top$](second literal in $C_j$ is true),  
[$\bot, \top, \bot, \top, \top, \bot$](third literal in $C_j$ is true). 
This is because every clause has distinct variables and we schedule $A_i$ sequentially. 
This value pattern is exactly matched by $B_j$.

\myparagraph{Correctness (Consistency $\Rightarrow$ Satisfiability)} 
Now we show if $\tuple{\AbstractExecution, \cpFunc}$ is consistent, 
then $\psi$ is 1-in-3 satisfiable. 
Given a concretization $\trace$ of $\tuple{\AbstractExecution, \cpFunc}$, 
we assign values for each $x_i$ as following. 
Since $A^\top_i, A^\bot_i$ both contain send events to $C_{i, 1}, \dots, C_{i, k_i}$, 
which are protected by channel $\lk_1$ and $\lk_2$,  
these send events should be executed atomically (see \figref{lock-demo-unbounded} for explanation). 
That is, for encoded events of each variable $x_i$, 
either all send events with value $\top$ are before all send events with value $\bot$ or 
all send events with value $\bot$ are before all send events with value $\top$. 
We assign $x_i$ to be true iff in $\trace$, 
all send events with value $\top$ are before all send events with value $\bot$. 

Now we prove that this assignment satisfies that in $\psi$, 
there is one and only one variable in an arbitrary clause $C_j$ being assigned true. 
There are totally six messages being sent to each channel $C_j$. 
That is, encoded events for each variable in clause $C_j$ will send two messages to channel $C_j$ 
and there are three variables in clause $C_j$. 
We consider the value sending to channel $C_j$, 
and can observe the $k$-th and $k+1$-th value are either $[\bot, {\top}]$ or $[\top, {\bot}]$ for all $k = 1, 3, 5$, 
because a variable will send both $\top, \bot$ once and $C_j$ has no duplicated variables. 
If $x_i$ is assigned to be true, then it corresponds to a message sequence of $[\top, {\bot}]$, 
and otherwise if $x_i$ is assigned to be false, then it corresponds to a message sequence of$[\bot, {\top}]$. 
In $B^\top_j, B^\bot_j$, 
we require in the three message sequences, 
exactly one of them should be $[\top, {\bot}]$ and the other two should be $[\bot, {\top}]$, 
which guarantees exactly one of the three variables in clause $C_j$ is assigned to be true. 
Therefore, $\psi$ is 1-in-3 satisfiable. 
\end{proof}


\subsection{Hardness with $1$ Channel}
\seclabel{lower-vch-chanel=1-cap=0-or-cap=1}
Finally, in this section we prove the hardness of $\vch$ even when threads communicate over a single channel, which can be either synchronous or have capacity $1$ (as per \thmref{vch-one-chan-hardness}).

\myparagraph{Overview}{
Our reduction is from positive 1-in-3 SAT.
Given $\psi$, we construct a corresponding $\vch$ instance $\tuple{\AbstractExecution, \cpFunc}$ with only one channel (either synchronous, or asynchronous with capacity $1$), where events in $\AbstractExecution$ comprise two phases.
The first phase guesses an assignment of the propositional variables of $\psi$, 
while the second phase verifies that every clause satisfies the 1-in-3 property, and executes residual events from the first phase.
The construction works for both when the unique channel is synchronous and when it has capacity $1$).
}


\begin{figure}[t]
\centering
\begin{subfigure}[b]{0.40\textwidth}
\centering
\newcommand{\xstep}{2.0}
\newcommand{\ystep}{-0.7}
\newcommand{\height}{0.5}
\newcommand{\wid}{1.5}
\scalebox{0.73}{
\begin{tikzpicture}
    \node (n1) at (0 * \xstep, 0 * \ystep) [event] {$\wtJustVal(x_1)$};
    \node (n1) at (0 * \xstep, 1 * \ystep) [event] {$\dots$};
    \node (n1) at (0 * \xstep, 2 * \ystep) [event] {$\wtJustVal(x_{\numVar})$};

    \node (n1) at (0.5 * \xstep, 4 * \ystep) [event] { $\wtJustVal(\Bar{x}_1)$};
    \node (n1) at (0.5 * \xstep, 5 * \ystep) [event] { $\dots$};
    \node (n1) at (0.5 * \xstep, 6 * \ystep) [event] { $\wtJustVal(\Bar{x}_{\numVar})$};

    \node (n1) at (1.3 * \xstep, 0 * \ystep) [event] {$\rdJustVal_1(\gamma)$};
    \node (n1) at (1.3 * \xstep, 1 * \ystep) [event] {$\dots$};
    \node (n1) at (1.3 * \xstep, 2 * \ystep) [event] {$\rdJustVal_{\numVar}(\gamma)$};
    \node (n1) at (1.3 * \xstep, 3 * \ystep) [event] {$\rdJustVal_1(\beta)$};
    \node (n1) at (1.3 * \xstep, 4 * \ystep) [event] {$\dots$};
    \node (n1) at (1.3 * \xstep, 5 * \ystep) [event] {$\rdJustVal_\numClause(\beta)$};
    \node (n1) at (1.3 * \xstep, 6 * \ystep) [event] {$\wtJustVal_1(\alpha)$};
    \node (n1) at (1.3 * \xstep, 7 * \ystep) [event] {$\dots$};
    \node (n1) at (1.3 * \xstep, 8 * \ystep) [event] {$\wtJustVal_{2{\numVar} + \numClause + 1}(\alpha)$};

    \node (n1) at (2.3 * \xstep, 0 * \ystep) [event] {$\rdJustVal_{2{\numVar} + \numClause + 1}(\alpha)$};
    \node (n1) at (2.3 * \xstep, 1 * \ystep) [event] {$\rdJustVal_1(\gamma)$};
    \node (n1) at (2.3 * \xstep, 2 * \ystep) [event] {$\dots$};
    \node (n1) at (2.3 * \xstep, 3 * \ystep) [event] {$\rdJustVal_{\numVar}(\gamma)$};

    \node[threadName] at (0 * \xstep, -1.0 * \ystep) {$\thread_\top$};
    \node[threadName] at (0.5 * \xstep, -1.0 * \ystep) {$\thread_\bot$};
    \node[threadName] at (1.3 * \xstep, -1.0 * \ystep) {$\thread_{\sf conn}$};
    \node[threadName] at (2.3 * \xstep, -1.0 * \ystep) {$\thread_{\gamma}$};
    
    \begin{scope}[on background layer]
        \draw [thread] (0 * \xstep, -0.6 * \ystep) -- (0 * \xstep, 8.8 * \ystep);
        \draw [thread] (0.5 * \xstep, -0.6 * \ystep) -- (0.5 * \xstep, 8.8 * \ystep);
        \draw [thread] (1.3 * \xstep, -0.6 * \ystep) -- (1.3 * \xstep, 8.8 * \ystep);
        \draw [thread] (2.3 * \xstep, -0.6 * \ystep) -- (2.3 * \xstep, 8.8 * \ystep);
    \end{scope}
    
\end{tikzpicture}
}
\figlabel{vch-one-channel-reduction-helper}
\caption{Auxiliary threads}
\end{subfigure}
\hfill
\begin{subfigure}[b]{0.36\textwidth}
\centering
\newcommand{\xstep}{1.4}
\newcommand{\ystep}{-0.7}
\newcommand{\height}{0.5}
\newcommand{\wid}{1.7}
\scalebox{0.73}{
\begin{tikzpicture}
    \node (n1) at (0 * \xstep, 0 * \ystep) [event] {$\rdJustVal(x_i)$};
    \node (n1) at (0 * \xstep, 1 * \ystep) [event] {$\rdJustVal(\Bar{x}_i)$};
    \node (n1) at (0 * \xstep, 2 * \ystep) [event] {$\wtJustVal(C_{i,1})$};
    \node (n1) at (0 * \xstep, 3 * \ystep) [event] {$\dots$};
    \node (n1) at (0 * \xstep, 4 * \ystep) [event] {$\wtJustVal(C_{i,f_i})$};
    \node (n1) at (0 * \xstep, 5 * \ystep) [event] {$\wtJustVal(\gamma)$};

    \node (n1) at (1 * \xstep, 6 * \ystep) [event] {$\rdJustVal(\Bar{x}_i)$};
    \node (n1) at (1 * \xstep, 7 * \ystep) [event] { $\rdJustVal(x_i)$};
    \node (n1) at (1 * \xstep, 8 * \ystep) [event] { $\wtJustVal(\gamma)$};

    \node (n1) at (2 * \xstep, 0 * \ystep) [event] {$\rdJustVal_{2i-1}(\alpha)$};
    \node (n1) at (2 * \xstep, 1 * \ystep) [event] {$\wtJustVal(x_i)$};

    \node (n1) at (3 * \xstep, 4 * \ystep) [event] {$\rdJustVal_{2i}(\alpha)$};
    \node (n1) at (3 * \xstep, 5 * \ystep) [event] {$\wtJustVal(\Bar{x}_i)$};

    \node[threadName] at (0 * \xstep, -1.0 * \ystep) {$\thread_{i,1}$};
    \node[threadName] at (1 * \xstep, -1.0 * \ystep) {$\thread_{i,2}$};
    \node[threadName] at (2 * \xstep, -1.0 * \ystep) {$\thread_{i,3}$};
    \node[threadName] at (3 * \xstep, -1.0 * \ystep) {$\thread_{i,4}$};

    \begin{scope}[on background layer]
        \draw [thread] (0 * \xstep, -0.6 * \ystep) -- (0 * \xstep, 8.8 * \ystep);
        \draw [thread] (1 * \xstep, -0.6 * \ystep) -- (1 * \xstep, 8.8 * \ystep);
        \draw [thread] (2 * \xstep, -0.6 * \ystep) -- (2 * \xstep, 8.8 * \ystep);
        \draw [thread] (3 * \xstep, -0.6 * \ystep) -- (3 * \xstep, 8.8 * \ystep);
    \end{scope}
\end{tikzpicture}
}
\figlabel{vch-one-channel-reduction-var}
\caption{Threads for variable $x_i$}
\end{subfigure}
\hfill
\begin{subfigure}[b]{0.22\textwidth}
\newcommand{\xstep}{1.6}
\newcommand{\ystep}{-0.7}
\newcommand{\height}{0.5}
\newcommand{\wid}{1.8}
\centering
\scalebox{0.73}{
\begin{tikzpicture}
    \node (n1) at (0 * \xstep, 0 * \ystep) [event] {$\rdJustVal(C_j)$};
    \node (n1) at (0 * \xstep, 1 * \ystep) [event] {$\wtJustVal(\beta)$};

    \node (n1) at (1 * \xstep, 4 * \ystep) [event] {$\rdJustVal_{2n + j}(\alpha)$};
    \node (n1) at (1 * \xstep, 5 * \ystep) [event] {$\rdJustVal(C_j)$};
    \node (n1) at (1 * \xstep, 6 * \ystep) [event] {$\rdJustVal(C_j)$};
    
    \node[threadName] at (0 * \xstep, -1.0 * \ystep) {$\thread_{j}^1$};
    \node[threadName] at (1 * \xstep, -1.0 * \ystep) {$\thread_{j}^2$};

    \begin{scope}[on background layer]
        \draw [thread] (0 * \xstep, -0.6 * \ystep) -- (0 * \xstep, 8.8 * \ystep);
        \draw [thread] (1 * \xstep, -0.6 * \ystep) -- (1 * \xstep, 8.8 * \ystep);
    \end{scope}
    \end{tikzpicture}
}
\figlabel{vch-one-channel-reduction-clause}
\caption{Threads for clause $C_j$}
\end{subfigure}
\caption{An example of the reduction from positive one-in-three satisfiability. }
\figlabel{one-channel-reduction}
\end{figure}

\myparagraph{Reduction}{
\figref{one-channel-reduction} illustrates this construction.
Given a formula $\psi$ with $\numVar$ variables and $\numClause$ clauses, 
$\tuple{\AbstractExecution, \cpFunc}$ has $4 + 4\numVar + 2\numClause$ threads
$\set{\thread_\top, \thread_\bot, \thread_\gamma, \thread_{\sf conn}}$
$\uplus \setpred{\thread_{i, 1}, \thread_{i, 2}, \thread_{i, 3}, \thread_{i, 4}}{1 \leq i \leq \numVar}$
$\uplus \setpred{\thread_{j}^1, \thread_{j}^2}{1 \leq j \leq \numClause}$.
Since we use a single channel $\ch$, we will omit explicitly mention it
and use the shorthand $\wtJustVal(\val)$ or $\rdJustVal(\val)$ to denote
send and receive events on $\ch$ with value $\val$.
We first describe the sequences of the $4$ auxiliary threads:
\begin{align*}
\begin{array}{l}
\thread_\top = \wtJustVal(x_1) \cdots \wtJustVal(x_{\numVar})
\quad
\thread_\bot = \wtJustVal(\Bar{x}_1) \cdots \wtJustVal(\Bar{x}_{\numVar})
\quad
\thread_{\gamma} = \rdJustVal_{2\numVar + \numClause + 1}(\alpha) \cdot \rdJustVal_1(\gamma) \cdots \rdJustVal_{\numVar}(\gamma) \\
\thread_{\sf conn} = \rdJustVal_1(\gamma) \cdots \rdJustVal_{\numVar}(\gamma) \cdot \rdJustVal_1(\beta) \cdots \rdJustVal_{\numClause}(\beta) \cdot  \wtJustVal_1(\alpha) \cdots \wtJustVal_{2\numVar + \numClause + 1}(\alpha) \\
\end{array}
\end{align*}
Here, $\alpha, \beta, \gamma, x_1, \ldots, x_{n_1}, \Bar{x}_{1}, \ldots, \Bar{x}_{\numVar}$
are distinct values.
Next, we describe the content of thread the four threads corresponding to each variable $x_i$ in $\psi$:
\begin{align*}
\begin{array}{rclcrcl}
\thread_{i, 1} &=& \rdJustVal(x_i) \cdot \rdJustVal(\Bar{x}_i) \cdot \wtJustVal(C_{i, 1}) \cdots \wtJustVal(C_{i, f_i}) \cdot \wtJustVal(\gamma) 
&
\quad
&
\thread_{i, 2} &=& \rdJustVal(\Bar{x}_i) \cdot \rdJustVal(x_i) \cdot \wtJustVal(\gamma) \\
\thread_{i, 3} &=& \rdJustVal_{2i-1}(\alpha) \cdot \wtJustVal(x_i)
&
\quad
&
\thread_{i, 4} &=& \rdJustVal(\alpha)_{2i} \cdot \wtJustVal(\Bar{x}_i)
\end{array}
\end{align*}
where $f_i$ is the frequency of $x_i$ in $\psi$,
and $C_{i, p}$ is the clause in which $x_i$ appears for the $p^\text{th}$ time. 
Finally, we have two threads for each clause $C_j$:
\begin{align*}
\begin{array}{ccc}
\thread_{j}^1 = \rdJustVal(C_j) \cdot \wtJustVal(\beta)
&
\quad
& \thread_{j}^2 = \rdJustVal_{2\numVar + j}(\alpha) \cdot \rdJustVal(C_j) \cdot \rdJustVal(C_j)
\end{array}
\end{align*}

The following lemma states the correctness of the construction.

\begin{restatable}{lemma}{vchOneChannelNPHard}
\lemlabel{vch-one-channel-np-hard}
$\tuple{\AbstractExecution, \cpFunc}$ is consistent iff $\psi$ is 1-in-3 satisfiable. 
\end{restatable}

Overall, $\tuple{\AbstractExecution, \cpFunc}$ has $O(\numVar + \numClause)$ events, 
concluding \thmref{vch-one-chan-hardness}. 

}


\applabel{proof-np-m=1}

\vchOneChannelNPHard*

\begin{proof}
We prove each direction separately.

    \myparagraph{Proof of correctness (Satisfiability $\Rightarrow$ Consistency)}
Given an assignment that satisfies $\psi$, we sketch the concretization $\trace$ as following. 
In general $\trace = \trace_1 \circ \trace_2$, 
where $\trace_1$ encodes the execution of $\thread_{\top}, \thread_{\bot}, \thread_{j_1}$ for all $1 \le j \le m$ and one of $\thread_{i, 1}, \thread_{i, 2}$ for each $1 \le i \le n$, 
and $\trace_2$ is a sequence of the rest events. 
\begin{equation}
    \trace_1 = S_1 \circ \dots \circ S_{\numVar}
    \nonumber
\end{equation}
where $S_i$ is a sequence of events we encode for variable $x_i$. 
For convenience, when multiple threads contain send or receive events with the same value, then we denote events as $\wtJustVal(a, t)$ to show that this is an event $\wtJustVal(a)$ from thread $t$. 
If $x_i$ is assigned to be true, 
let clauses $C_{i, 1}, \dots, C_{i, f_i}$ be all the clauses $x_i$ appears in. 
\begin{equation}
\begin{aligned}
    S_i &= \wtJustVal(x_i, \thread_{\top}) \cdot \rdJustVal(x_i, \thread_{i, 1}) \cdot \wtJustVal(\Bar{x}_i, \thread_{\bot}) \cdot \rdJustVal(\Bar{x}_i, \thread_{i, 1}) \\
    & \cdot 
    \wtJustVal(C_{i, 1}, \thread_{i, 1}) \cdot 
    \rdJustVal(C_{i, 1}, \thread_{i_1}^1) \cdots 
    \wtJustVal(C_{i, f_i}, \thread_{i, 1}) \cdot \rdJustVal(C_{i, f_i}, \thread_{f_i}^1) \cdot \wtJustVal(\gamma, \thread_{i, 1}) \cdot \rdJustVal(\gamma, \thread_{conn})
    \nonumber
\end{aligned}
\end{equation}
Otherwise 
\begin{equation}
    S_i = \wtJustVal(\Bar{x}_i, \thread_{\bot}) \cdot \rdJustVal(\Bar{x}_i, \thread_{i, 2}) \cdot \wtJustVal(x_i, \thread_{\top}) \cdot \rdJustVal(x_i, \thread_{i, 2}) \cdot \wtJustVal(\gamma, \thread_{i, 2}) \cdot \rdJustVal(\gamma, \thread_{conn})
    \nonumber
\end{equation}

Now we describe the details of $\trace_2 = B_0 \circ B_1 \circ \dots \circ B_n$, 
where $B_0$ is a sequence of events to link two phases and $B_i$ ($i>0$) is the encoded events for $x_i$. 
\begin{equation}
\begin{aligned}
    B_0 &= \wtJustVal(\beta, \thread^1_{1}) \cdot \rdJustVal_1(\beta) \cdots \wtJustVal(\beta, \thread^1_{\numClause}) \cdot \rdJustVal_\numClause(\beta) \\
    & \cdot \wtJustVal_1(\alpha) \cdot \rdJustVal_1(\alpha) \cdots \wtJustVal_{2\numVar+\numClause+1}(\alpha) \cdot \rdJustVal_{2\numVar+\numClause+1}(\alpha)
    \nonumber
\end{aligned}
\end{equation}
Here $\rcv(\alpha)$ are just the first event in every $\thread_{i, 3}, \thread_{i, 4}, \thread_j^2, \thread_\gamma$. 
Any permutations of these events suffice. 

If $x_i$ is assigned to be true, then 
\begin{equation}
\begin{aligned}
    B_i &= \wtJustVal(\Bar{x}_i, \thread_{i, 4}) \cdot 
    \rdJustVal(\Bar{x}_i, \thread_{i, 2}) \cdot \wtJustVal(x_i, \thread_{i, 3}) \cdot \rdJustVal(x_i, \thread_{i, 2}) \cdot \wtJustVal(\gamma, \thread_{i, 2}) \cdot \rdJustVal(\gamma, \thread_{\gamma}) 
    \nonumber
\end{aligned}
\end{equation}
Otherwise, 
\begin{equation}
\begin{aligned}
    B_i &= \wtJustVal(x_i, \thread_{i, 3}) \cdot \rdJustVal(x_i, \thread_{i, 1}) \cdot \wtJustVal(\Bar{x}_i, \thread_{i, 4}) \cdot \rdJustVal(\Bar{x}_i, \thread_{i, 1}) 
    \cdot \wtJustVal(C_{i, 1}, \thread_{i, 1}) \cdot \rdJustVal(C_{i, 1}, \thread_{i_1}^2) \\
    & \cdot \wtJustVal(C_{i, 2}, \thread_{i, 1}) \cdot \rdJustVal(C_{i, 2}, \thread_{i_2}^2) \cdot \wtJustVal(C_{i, 3}, \thread_{i, 1}) \cdot \rdJustVal(C_{i, 3}, \thread_{i_3}^2), \wtJustVal(\gamma, \thread_{i, 1}) \cdot \rdJustVal(\gamma, \thread_{\gamma})
    \nonumber
\end{aligned}
\end{equation}
The program order is clearly satisfied. 
For value constraints, we can observe that in $\trace$, 
a send event is immediately followed by a receive event with the same value, 
so that $\trace$ must satisfy value constraints. 
Moreover, we note the above completion is also valid when the channel is synchronous, 
because a send event is immediately followed by a receive event with the same value from another thread. 

\myparagraph{Proof of correctness (Consistency $\Rightarrow$ Satisfiability)}
Given a completion $\trace$, 
we assign an arbitrary variable $x_q = T$ iff in $\trace$, 
$\wtJustVal(x_q)$ in $\thread_{\top}$ is ordered before $\wtJustVal(\Bar{x}_q)$ in $\thread_{\bot}$. 
Now we prove this assignment makes $\psi$ one-in-three satisfiable. 

Firstly, we show some simple observations. 
Because of the value $\alpha, \beta$, 
one must execute $\rdJustVal(C_j)$ in $\thread^1_j$ before all events in $\thread_{j}^2, \thread_{i, 3}, \thread_{i, 4}$ for all $i, j$. 
Then we consider the value $\gamma$, 
and notice that only $\thread_{i, 1}, \thread_{i, 2}$ send value $\gamma$ once per thread. 
This means, 
in order to execute the events in $\thread_{i, 3}$, $\thread_{i, 4}$, $\thread_{j}^2$ for all $i, j$, 
we have to fully execute at least $n$ threads among $\thread_{i, 1}, \thread_{i, 2}$ for all $i$. 
For each fixed $i$, 
we must execute exactly one of $\thread_{i, 1}$, $\thread_{i, 2}$, 
because there is only one $\wtJustVal(x_i), \wtJustVal(\Bar{x}_i)$ in $\thread_{\top}, \thread_{\bot}$ and the other two are in thread $\thread_{i, 3}, \thread_{i, 4}$. 

Secondly, we show $\psi$ must be satisfied. 
Given the observations above, 
one must execute the sent event at least once with value $C_j$ for all $1 \le j \le m$ before $\thread_{j}^2, \thread_{i, 3}, \thread_{i, 4}$ can be executed. 
By our assignment, this means for each clause $C_j$, 
there is at least one variable in $C_j$ being assigned true, 
so that $C_j$ is satisfied. 

Thirdly, we show each clause is satisfied by exactly one variable. 
If more than one variable are assigned to be true in $C_j$, 
then value $C_j$ must be sent more than once before $\thread_{j}^2, \thread_{i, 3}, \thread_{i, 4}$ can be executed. 
However, because of value $\gamma, \alpha$, we cannot immediately receive the second (or third) $C_j$ value, 
which makes the $\vch$ instance not consistent. 
Therefore, exactly one literal per clause is assigned to be true. 
The same reasoning works for synchronous channel. 
\end{proof}

\section{Lower bounds for $\vchRf$}
\applabel{lower-bounds-vch-rf}
\subsection{Hardness with Asynchronous Channels of Capacity 1}

\applabel{proof-np-k=1}


\vchRfConstK*

\begin{proof}
We prove each direction separately.

\myparagraph{Correctness ($\vscRd \Rightarrow \vchRf$)}
If $\vscRd$ instance $\AbstractExecution = \tuple{\eventSet, \po{}, \rf{}}$ is consistent, then $\vchRf$ instance $\tuple{\AbstractExecution', \cpFunc', \rf{}'}$ is consistent. 
For a sequence of events $\pi = e_1 \cdots e_n$ in $\eventSet$, 
we define the mapping of $\pi$ using $M$ as $M(\pi) = M(e_1) \cdots M(e_n)$. 
Let $\reordering$ be a linear sequence concretizing $\AbstractExecution$, 
and we show $\trace = M(\reordering)$ is the concretization of $\tuple{\AbstractExecution', \cpFunc', \rf{}'}$. 
For convenience, we define the reverse map of $M$ as $M^{-1}$, 
where $M^{-1}(e) = f$ iff $e$ is in $M(f)$. 
That is, $M^{-1}$ maps a event $e$ in $\eventSet'$ back to the event $f \in \eventSet$, 
such that $e \in M(f)$. 

Firstly, we argue that $\trace$ respects $\po{}'$. 
Assuming $(e_1, e_2) \in \po{}'$ and $e_1$ is ordered after $e_2$ in $\trace$, 
then there are two possible situations. 
(1) $M^{-1}(e_1) = M^{-1}(e_2) = f$. 
This is impossible, 
because $\trace$ doesn't reorder events in $M(f)$. 
(2) $M^{-1}(e_1) \neq M^{-1}(e_2)$, 
then by definition of $\po{}'$, we have $(M^{-1}(e_1), M^{-1}(e_2)) \in \po{}$.  
In this case, $\trace$ should order $e_1$ before $e_2$, 
so that it's also impossible. 
Therefore, $\trace$ must respect $\po{}'$. 

Secondly, we argue the $\rf{}'$ is also satisfied. 
Assuming there is a receive event $\rd(\ch) \in \eventSet'$, 
it should observe $\wt(\ch)$, 
but turns out to observe the wrong send event $\wt'(\ch)$ in $\trace$. 
First, we argue $\ch$ cannot be $\lk$, 
because for each event $e \in \eventSet$, 
$M(e)$ contains exactly one send and its receiver to $\lk$, 
and $M(e)$ doesn't interleave with $M(e')$ in $\trace$ for all $e' \neq e$.  
Therefore, $\ch$ can only be $\ch_x^i$ for some register $x$ and index $i$. 
In this case, 
since every $\ch_x^i$ has capacity 1, 
we claim there is no event $\wtMem(x)$, 
such that $M^{-1}(\snd'(\ch)) \trord{\reordering} \wtMem(x) \trord{\reordering} M^{-1}(\rcv(\ch))$. 
Otherwise, there will be two continuous send events to $\ch_x^i$. 
Then $\reordering$ fails to meet the $\rf{}$ relation, 
because $(M^{-1}(\snd(\ch)), M^{-1}(\rcv(\ch))) \in \rf{}$, 
but $M^{-1}(\rcv(\ch))$ observes $M^{-1}(\snd'(\ch))$ in $\reordering$, 
which is impossible. 


Lastly, $\trace$ is well-formed. 
Indeed, by our construction, 
for each write event $\wtMem(x) \in \eventSet$ together with all read events observing $\wtMem(x)$, 
there will be exactly $m_x$ send and receive events. 
That is, we construct one send and receive event to each $\ch_x^i$. 
After executing all of them, $\ch_x^i$ will be empty again. 
Since the reads-from relation is satisfied, 
then $\trace$ should be well-formed.

\myparagraph{Correctness ($\vchRf \Rightarrow \vscRd$)}
Secondly, if $\tuple{\AbstractExecution', \cpFunc', \rf{}'}$ is consistent, 
then $\AbstractExecution$ is consistent. 
Let $\trace$ be a concretization of $\tuple{\AbstractExecution', \cpFunc', \rf{}'}$, 
we construct $\reordering$ as a concretization of $\AbstractExecution$ as following. 
We note that because of the channel $\lk$, 
every event sequence $M(e)$ should not interleave with each other for all $e \in \eventSet$ (see \figref{lock-demo-capacity-one} for explanation). 
Otherwise, there will be at least two continuous send events to channel $\lk$, 
which only has capacity 1. 
This implies we can map $\trace$ back into a serialized sequence $\reordering$ of $\eventSet$, 
s.t. $M(\reordering) = \trace$. 
We argue $\reordering$ is a valid concretization of ${\AbstractExecution}$. 

Firstly, we argue that $\po{}$ is satisfied. 
Assuming $(e_1, e_2) \in \po{}$ and $e_1$ is ordered after $e_2$ in $\reordering$, 
then it implies $M(e_1)$ should be ordered after $M(e_2)$ in $\trace$, 
which violates $\rf{}'$. 

Secondly, we argue $\rf{}$ is also satisfied. 
Assuming $\rdMem(x)$ should observe $\wtMem(x)$, 
but it observes $\wtMem'(x)$ in $\reordering$, 
we now consider the mapped event sequences in $\trace$. 
This implies one of the following two situations should happen. 
(1) $M(\wtMem(x)) \trord{\trace} M(\rdMem(x))$, 
which violates $\rf{}'$. 
(2) $M(\wtMem(x)) \trord{\trace} M(\wtMem'(x)) \trord{\trace} M(\rdMem(x))$. 
In this case, there will be two continuous send events to some channel $\ch_x^i$, 
which is impossible as well. 
Therefore $\rf{}$ must be satisfied and thus ${\AbstractExecution}$ is indeed consistent. 
\end{proof}



\subsection{Hardness with $3$ Threads and Small Channel Capacity}

\applabel{proof-np-t=3-k=2}

\seclabel{vch-rf-np-hard-t=3-k=2}
Here we show that $\vchRf$ is $\NP$-hard already with $3$ threads and maximum channel capacity $\MaxCapacity\leq 2$.

\myparagraph{Overview}{
Our reduction is from the 3SAT problem, and constructs
a $\vchRf$ instance 
$\tuple{\AbstractExecution, \cpFunc, \rf{}}$, 
where $\AbstractExecution = \tuple{\eventSet, \po{}}$ 
starting from a given 3CNF formula $\psi$,
such that $\tuple{\AbstractExecution, \cpFunc, \rf{}}$ is consistent iff $\psi$ is satisfiable. 
Let $\psi = C_1 \land C_2 \cdots C_{\numClause}$ be a conjunction of $\numClause$
clauses over $\numVar$ propositional variables $x_1, \ldots, x_\numVar$.
At a high level, ${\AbstractExecution}$ is structured in $2$
phases.
The first phase, divided into $\numVar$ sub-phases arranged sequentially, 
picks an assignment for each variable $x_i$. 
The second phase, divided into $\numClause$ sub-phases arranged sequentially,
encode the constraint that for clause $C_j$, the assignment to
at least one of three literals in $C_j$ was picked to be true in the first phase. 
\figref{reduction-fix-t-k} shows the schema of our hardness construction.
}


\begin{figure}[t]
\centering
\begin{subfigure}{0.20\textwidth}
\centering
\newcommand{\xstep}{0.8}
\newcommand{\ystep}{-0.5}
\newcommand{\height}{0.68}
\newcommand{\wid}{2.2}
\tikzstyle{var}=[event, draw=colorVar, rounded corners]
\tikzstyle{clause}=[event, draw=colorClause, rounded corners]
\scalebox{0.7}{
\begin{tikzpicture}
\begin{scope}[shift={(0,-0.5*\ystep)}]
\node[var] (t1-3) at (1 * \xstep, 0 * \ystep) {$x_1$};
\node[var] (t1-4) at (1 * \xstep, 2 * \ystep) {$\dots$};
\node[var] (t1-5) at (1 * \xstep, 4 * \ystep) {$x_{\numVar}$};
\node[clause] (t1-5) at (1 * \xstep, 6 * \ystep) {$C_1$};
\node[clause] (t1-5) at (1 * \xstep, 8 * \ystep) {$\dots$};
\node[clause] (t1-5) at (1 * \xstep, 10 * \ystep){$C_{\numClause}$};
\end{scope}

\node[threadName] at (0 * \xstep, -2.6 * \ystep) {$\thread_{1}$ };
\node[threadName] at (1 * \xstep, -2.6 * \ystep) {$\thread_{2}$ };
\node[threadName] at (2 * \xstep, -2.6 * \ystep) {$\thread_{3}$ };

\begin{scope}[on background layer]
\draw[thread] (0 * \xstep, -2.0 * \ystep) -- (0 * \xstep, 11.3 * \ystep);
\draw[thread] (1 * \xstep, -2.0 * \ystep) -- (1 * \xstep, 11.3 * \ystep);
\draw[thread] (2 * \xstep, -2.0 * \ystep) -- (2 * \xstep, 11.3 * \ystep);
\end{scope}
\end{tikzpicture}
}
\caption{Overall Scheme.}
\figlabel{reduction-fix-t-k-2-overall}
\end{subfigure}
\begin{subfigure}[b]{0.35\textwidth}
\centering
\newcommand{\xstep}{2.6}
\newcommand{\ystep}{-0.8}
\newcommand{\height}{0.6}
\newcommand{\wid}{1.9}
\newcommand{\memDistance}{1.6}
\newcommand{\bracketWid}{0.45}
\newcommand{\bracketHeight}{0.25}
\tikzstyle{memEvent}=[draw=none]
\scalebox{0.7}{
\begin{tikzpicture}

\draw[rounded corners, line width=0.4mm , color=colorVar, dashed]  (-0.45 * \xstep, -0.6 * \ystep) rectangle (1.7 * \xstep, 6.6 * \ystep) {};

\node[event] (n1-1) at (0 * \xstep, 0 * \ystep) {$\wt(\alpha)$};
\node[event] (n1-2) at (0 * \xstep, 1 * \ystep) {$\rd(\alpha)$};
\node[event] (n1-3) at (0 * \xstep, 2 * \ystep) {$\wt(\lk)$};
\node[event] (n1-4) at (0 * \xstep, 3 * \ystep) {$\wt_\top(\ch_i^1)$};
\node[event] (n1-5) at (0 * \xstep, 4 * \ystep) {$\dots$};
\node[event] (n1-6) at (0 * \xstep, 5 * \ystep) {$\wt_\top(\ch_i^{f_i})$};
\node[event] (n1-7) at (0 * \xstep, 6 * \ystep) {$\rd(\lk)$};

\node[event] (n2-1) at (1 * \xstep, 0 * \ystep) {$\wt(\alpha)$};
\node[event] (n2-2) at (1 * \xstep, 1 * \ystep) {$\rd(\alpha)$};
\node[event] (n2-3) at (1 * \xstep, 2 * \ystep) {$\wt(\lk)$};
\node[event] (n2-4) at (1 * \xstep, 3 * \ystep) {$\wt_\bot(\ch_i^1)$};
\node[event] (n2-5) at (1 * \xstep, 4 * \ystep) {$\dots$};
\node[event] (n2-6) at (1 * \xstep, 5 * \ystep) {$\wt_\bot(\ch_i^{f_i})$};
\node[event] (n2-7) at (1 * \xstep, 6 * \ystep) {$\rd(\lk)$};

\draw [memEvent] ($ (n1-1) + (-\memDistance, 0) $)  to node [midway,fill=white] {$A^r_i$} ($ (n1-7) + (-\memDistance, 0) $);
    \draw[line width=1.5pt] ($(n1-1) + (-\memDistance - 0.5 * \bracketWid, 0.5*\height - \bracketHeight)$) rectangle ($(n1-1) + (-\memDistance + 0.5 * \bracketWid, 0.5*\height)$);    
    \draw[white, line width=3pt] ($(n1-1) + (-\memDistance - 0.6 * \bracketWid, 0.5*\height - \bracketHeight)$) -- ($(n1-1) + (-\memDistance + 0.6 * \bracketWid, 0.5*\height - \bracketHeight)$); 

    \draw[line width=1.5pt] ($(n1-7) + (-\memDistance - 0.5 * \bracketWid, -0.5*\height)$) rectangle ($(n1-7) + (-\memDistance + 0.5 * \bracketWid, -0.5*\height + \bracketHeight)$);    
    \draw[white, line width=3pt] ($(n1-7) + (-\memDistance - 0.6 * \bracketWid, -0.5*\height + \bracketHeight)$) -- ($(n1-7) + (-\memDistance + 0.6 * \bracketWid, -0.5*\height + \bracketHeight)$); 

\node[threadName] at (0 * \xstep, -1.5 * \ystep) {$\thread_{1}$};
\node[threadName] at (1 * \xstep, -1.5 * \ystep) {$\thread_{2}$};
\node[threadName] at (1.6 * \xstep, -1.5 * \ystep) {$\thread_{3}$};

\draw[rfEdge] (n1-1) to (n2-2);
\draw[rfEdge] (n2-1) to (n1-2);

\begin{scope}[on background layer]
\draw [thread] (0 * \xstep, -1 * \ystep) -- (0 * \xstep, 7.3 * \ystep);
\draw [thread] (1 * \xstep, -1 * \ystep) -- (1 * \xstep, 7.3 * \ystep);
\draw [thread] (1.6 * \xstep, -1 * \ystep) -- (1.6 * \xstep, 7.3 * \ystep);
\end{scope}
\end{tikzpicture}
}
\caption{Sub-phase for variable $x_i$.}
\figlabel{reduction-fix-t-k-1}
\end{subfigure}
\begin{subfigure}[b]{0.43\textwidth}
\centering
\newcommand{\xstep}{2.5}
\newcommand{\ystep}{-1.0}
\newcommand{\height}{0.65}
\newcommand{\wid}{2.0}
\newcommand{\memDistance}{1.5}
\newcommand{\bracketWid}{0.45}
\newcommand{\bracketHeight}{0.25}
\tikzstyle{memEvent}=[draw=none]
\scalebox{0.7}{
\begin{tikzpicture}
\draw[rounded corners, line width=0.4mm , color=colorClause, dashed]  (-0.47 * \xstep, -0.6 * \ystep) rectangle (2.5 * \xstep, 5.5 * \ystep) {};

\node[event] (t1-3) at (0 * \xstep, 0 * \ystep)  {$\wt(\beta_1)$};
\node[event] (t1-4) at (0 * \xstep, 3 * \ystep)  {$\rd(\beta_4)$};
\node[event] (t1-5) at (0 * \xstep, 4 * \ystep)  {$\rd_\top(\ch_{j_1}^{m_1})$};
\node[event] (t1-6) at (0 * \xstep, 5 * \ystep)  {$\rd_\bot(\ch_{j_2}^{m_2})$};

\node[event] (t2-5) at (1 * \xstep, 0 * \ystep) {$\rd(\beta_1)$};
\node[event] (t2-6) at (1 * \xstep, 1 * \ystep) {$\wt(\beta_2)$};
\node[event] (t2-7) at (1 * \xstep, 2 * \ystep) {$\rd(\beta_3)$};
\node[event] (t2-8) at (1 * \xstep, 3 * \ystep) {$\wt(\beta_4)$};
\node[event] (t2-9) at (1 * \xstep, 4 * \ystep) {$\rd_\top(\ch_{j_2}^{m_2})$};
\node[event] (t2-10) at (1 * \xstep, 5 * \ystep) {$\rd_\bot(\ch_{j_3}^{m_3})$};

\node[event] (t3-3) at (2 * \xstep, 1 * \ystep) {$\rd(\beta_2)$};
\node[event] (t3-4) at (2 * \xstep, 2 * \ystep) {$\wt(\beta_3)$};
\node[event] (t3-5) at (2 * \xstep, 4 * \ystep) {$\rd_\top(\ch_{j_3}^{m_3})$};
\node[event] (t3-6) at (2 * \xstep, 5 * \ystep) {$\rd_\bot(\ch_{j_1}^{m_1})$};

\draw [memEvent] ($ (t1-3) + (-\memDistance, 0) $)  to node [midway,fill=white] {$B^r_i$} ($ (t1-6) + (-\memDistance, 0) $);
    \draw[line width=1.5pt] ($(t1-3) + (-\memDistance - 0.5 * \bracketWid, 0.5*\height - \bracketHeight)$) rectangle ($(t1-3) + (-\memDistance + 0.5 * \bracketWid, 0.5*\height)$);    
    \draw[white, line width=3pt] ($(t1-3) + (-\memDistance - 0.6 * \bracketWid, 0.5*\height - \bracketHeight)$) -- ($(t1-3) + (-\memDistance + 0.6 * \bracketWid, 0.5*\height - \bracketHeight)$); 

    \draw[line width=1.5pt] ($(t1-6) + (-\memDistance - 0.5 * \bracketWid, -0.5*\height)$) rectangle ($(t1-6) + (-\memDistance + 0.5 * \bracketWid, -0.5*\height + \bracketHeight)$);    
    \draw[white, line width=3pt] ($(t1-6) + (-\memDistance - 0.6 * \bracketWid, -0.5*\height + \bracketHeight)$) -- ($(t1-6) + (-\memDistance + 0.6 * \bracketWid, -0.5*\height + \bracketHeight)$); 
    
\node[threadName] at (0 * \xstep, -1.2 * \ystep) {$\thread_{1}$ };
\node[threadName] at (1 * \xstep, -1.2 * \ystep) {$\thread_{2}$ };
\node[threadName] at (2 * \xstep, -1.2 * \ystep) {$\thread_{3}$ };

\draw [rfEdge] (t1-3) to (t2-5);
\draw [rfEdge] (t2-6) to (t3-3);
\draw [rfEdge] (t2-8) to (t1-4);
\draw [rfEdge] (t3-4) to (t2-7);

\begin{scope}[on background layer]
\draw [thread] (0 * \xstep, -0.9 * \ystep) -- (0 * \xstep, 5.8 * \ystep);
\draw [thread] (1 * \xstep, -0.9 * \ystep) -- (1 * \xstep, 5.8 * \ystep);
\draw [thread] (2 * \xstep, -0.9 * \ystep) -- (2 * \xstep, 5.8 * \ystep);
\end{scope}
\end{tikzpicture}
}
\caption{Sub-phase for clause $C_j$.}
\figlabel{reduction-fix-t-k-2-subphase}
\end{subfigure}
\caption{
Reduction from 3SAT to $\vchRf$. 
Here $\cp{\lk} = \cp{\beta_i} = 1$, 
and $\cp{\alpha} = \cp{\ch_i^s} = 2$.
Reads-from relations are either depicted using red arrows or are described in texts. 
}
\figlabel{reduction-fix-t-k}
\end{figure}

\myparagraph{Reduction}{
The $\vchRf$ instance $\tuple{\AbstractExecution, \cpFunc, \rf{}}$ we construct 
has $3$ threads $\thread_1, \thread_2, \thread_3$.  
It uses the following sets of distinct channels
$\mathcal{C}_{1} \uplus \mathcal{C}_{2}$,
where $\mathcal{C}_1 = \set{\lk, \beta_1, \beta_2, \beta_3, \beta_4}$
is the set of asynchronous channels with capacity $1$,
while $\mathcal{C}_2 = \set{\alpha} \uplus \setpred{\ch_i^s}{1 \leq s \leq f_i}$
is the set of asynchronous channels with capacity $2$,
where $f_i$ denotes the number of occurrences of variable $x_i$ in formula $\psi$,
and the channel $\ch_i^{s}$ will represent the $s^\text{th}$ occurrence of $x_i$. 
For each thread $\thread_r$ ($r \in \set{1, 2, 3}$),
the sequence $\rho_r$ of events in $\thread_r$ is of the form $\rho_r = A^r \cdot B^r$,
where $A^r$ and $B^r$ are sequences of events corresponding
to the first and second phases respectively and have the form
\begin{align*}
\begin{array}{rcl}
A^r = A^r_1 \cdot A^r_2 \cdots A^r_\numVar & \quad & B^r = B^r_1 \cdot B^r_2 \cdots B^r_\numClause
\end{array}
\end{align*}
The sequence $A^r_i$ encodes some choice of assignments
to variable $x_i$. 
Each $A^r_i$ contains an atomic event sequence for $r = 1, 2$ 
and the atomicity is guaranteed by channel $\lk$ with capacity 1 (see \figref{lock-demo-capacity-one}). 
In particular, $A^3_i = \epsilon$ is the empty sequence, while
$A^1_i$ and $A^2_i$ are described next:
\begin{align*}
\begin{array}{rcl}
A^1_i &=& \snd(\alpha) \cdot \rcv(\alpha) \cdot \snd(\lk) \cdot \snd_\top(\ch_i^1) \cdots \snd_\top(\ch_i^{f_i}) \cdot \rcv(\lk) \\
A^2_i &=& \snd(\alpha) \cdot \rcv(\alpha) \cdot \snd(\lk) \cdot \snd_\bot(\ch_i^1) \cdots \snd_\bot(\ch_i^{f_i}) \cdot \rcv(\lk)
\end{array}
\end{align*} 
Consider the clause $C_j = \gamma_1 \lor \gamma_2 \lor \gamma_3$,
where $\gamma_s$ is a literal over variable $x_{j_s}$
(we assume $j_1 < j_2 < j_3$),
and let $C_j$ be respectively the $m_1^\text{th}$, $m_2^\text{th}$ and $m_3^\text{th}$
occurrence of $x_{j_1}, x_{j_2}, x_{j_3}$ in $\psi$.
Then, $B^1_j, B^2_j, B^3_j$ are the following sequences corresponding to $C_j$ in threads $\thread_1, \thread_2, \thread_3$ respectively:
\begin{align*}
\begin{array}{rcl}
B^1_j &=& \snd(\beta_1) \cdot \rcv(\beta_4) \cdot \rcv_\top(\ch^{m_1}_{j_1}) \cdot \rcv_\bot(\ch^{m_2}_{j_2}) \\
B^2_j &=& \rcv(\beta_1) \cdot \snd(\beta_2) \cdot \rcv(\beta_3) \cdot\snd(\beta_4) \cdot \rcv_\top(\ch^{m_2}_{j_2})\cdot \rcv_{\bot}(\ch^{m_3}_{j_3}) \\
B^3_j &=& \rcv(\beta_2) \cdot \snd(\beta_3) \cdot \rcv_\top(\ch^{m_3}_{j_3}) \cdot \rcv_\bot(\ch^{m_1}_{j_1})
\end{array}
\end{align*}
We now specify the reads-from relation: 
\begin{itemize}
\item $\rcv(\lk)$ in $A^r_i$ observes $\snd(\lk)$ in $A^r_i$ ($r \in \set{1, 2}, i \in \set{1, \ldots, \numVar}$).
    
\item $\rcv(\alpha)$ in $A^{r}_i$ observes $\snd(\alpha)$ in $A^{\bar{r}}_i$ ($\set{r, \bar{r}} = \set{1, 2}, i \in \set{1, \ldots, \numVar}$).
The events on channel $\alpha$ thus ensure that all events (belonging to the first phase) of $x_i$ will appear before those of $x_{i+1}$ in any concretization.

\item $\rcv(\beta_s)$ in $B^r_j$ observes the event $\snd(\beta_s)$ in $B^s_j$ ($s \in \set{1, 2, 3, 4}, r \in \set{1, 2}, j \in \set{1, 2, \ldots, \numClause}$).

\item Recall that in clause $C_j = \gamma_1 \lor \gamma_2 \lor \gamma_3$
are such that the literal $\gamma_p$ is either $x_{j_p}$ or $\neg x_{j_p}$, and $C_j$ is the $m_p^\text{th}$
occurrence of $x_{j_p}$ in $\psi$ ($p \in \set{1, 2, 3}$).
In the former case (i.e., $\gamma_p = x_{j_p}$),
we pair the receive events $\rcv_\top(\ch^{m_p}_{j_p})$ and $\rcv_\bot(\ch^{m_p}_{j_p})$ 
to the send events $\snd_\top(\ch^{m_p}_{j_p})$ and $\snd_\bot(\ch^{m_p}_{j_p})$ in $A^1_{j_p}$ and $A^2_{j_p}$, respectively.
Otherwise (i.e., $\gamma_p = \neg x_{j_p}$),
we pair the receive events $\rcv_\top(\ch^{m_p}_{j_p})$ and $\rcv_\bot(\ch^{m_p}_{j_p})$ 
to the send events $\snd_\bot(\ch^{m_p}_{j_p})$ and $\snd_\top(\ch^{m_p}_{j_p})$ in $A^2_{j_p}$ and $A^1_{j_p}$, respectively.
\end{itemize}

The following lemma states the correctness of the above construction.

\begin{restatable}{lemma}{vchRfConstTK}
    \lemlabel{vch-rf-t=3-k=2}
    $\psi$ is satisfiable iff $\tuple{\AbstractExecution, \cpFunc, \rf{}}$ is consistent. 
\end{restatable}

Finally, the number of events in $\tuple{\AbstractExecution, \cpFunc, \rf{}}$ is $O(\numVar+\numClause)$,
which concludes case (ii) of \thmref{vch-rf-hardness}.

\begin{proof}
We prove each direction separately.

\myparagraph{Correctness (Satisfiability $\Rightarrow$ Consistency)}
If $\psi$ is satisfiable, then there is a concretization $\trace$, 
and we sketch it as following. 
In general, $\trace = \trace_1 \circ \trace_2$. 
We first describe a linear sequence $\trace_1$ of first phase ($A^r_i$).
Then we describe the linear sequence $\trace_2$ of the second phase ($B^r_j$).
In general, $\trace_1$ is of the following form.
\begin{equation}
    \sigma_1 = A_1 \circ A_2 \circ \dots \circ A_\numVar
    \nonumber
\end{equation}
where $A_i$ is a linear sequence of $A^1_i, A^2_i$ in $\thread_1, \thread_2$. 
Here we use superscript to denote the thread each event belongs to.
If $x_i$ is assigned to be true, then 
\begin{equation} 
\begin{aligned}
    A_i &= \snd^{\thread_2}(\alpha) \cdot \snd^{\thread_1}(\alpha) \cdot \rcv^{\thread_1}(\alpha) \cdot
            \rcv^{\thread_2}(\alpha) \\
        & \cdot 
            \snd^{\thread_1}(\lk) \cdot \snd_{\top}^{\thread_1}(\ch_i^1) \cdots \snd_{\top}^{\thread_1}(\ch_i^{f_i}) \cdot \rcv^{\thread_1}(\lk) \\
        & \cdot 
            \snd^{\thread_2}(\lk) \cdot \snd_{\bot}^{\thread_2}(\ch_i^1) \cdots \snd_{\bot}^{\thread_2}(\ch_i^{f_i}) \cdot \rcv^{\thread_2}(\lk)
    \nonumber
\end{aligned}
\end{equation}
If $x_i$ is assigned to be false, then 
\begin{equation} 
\begin{aligned}
    A_i &= \snd^{\thread_2}(\alpha) \cdot \snd^{\thread_1}(\alpha) \cdot \rcv^{\thread_1}(\alpha) \cdot
            \rcv^{\thread_2}(\alpha) \\
        & \cdot 
            \snd^{\thread_2}(\lk) \cdot \snd_{\bot}^{\thread_2}(\ch_i^1) \cdots \snd_{\bot}^{\thread_2}(\ch_i^{f_i}) \cdot \rcv^{\thread_2}(\lk) \\
        & \cdot 
            \snd^{\thread_1}(\lk) \cdot \snd_{\top}^{\thread_1}(\ch_i^1) \cdots \snd_{\top}^{\thread_1}(\ch_i^{f_i}) \cdot \rcv^{\thread_1}(\lk)
    \nonumber
\end{aligned}
\end{equation}
One can easily verify these two concretizations satisfy the reads-from relation within each $A^r_i$. 
Now we turn to $\trace_2$ and 
$\trace_2$ is of the following pattern.
\begin{equation}
    \trace_2 = B_1 \circ \dots \circ B_\numClause
    \nonumber
\end{equation}
where $B_j$ is a linear sequence of all events in $B^r_j$ for all $1 \le r \le 3$. 
$B_j$ depends on the value of each literal in $C_j = \gamma_1 \lor \gamma_2 \lor \gamma_3$. 
Since $C_j$ is satisfied, there exists one literal to be true. 
Without loss of generality, we assume $\gamma_1$ is true (other cases can be solved similarly). 
Then we have four possibilities, as the value of $\gamma_2, \gamma_3$ can be either true of false. 
\begin{itemize}
    \item If $\gamma_2=$ true and $\gamma_3$= true, then 
    \begin{equation}
        B_j = \rcv_{\top}(\ch^{m_1}_{j_1}) \cdot \rcv_{\top}(\ch^{m_2}_{j_2})
        \cdot \rcv_{\top}(\ch^{m_3}_{j_3})
        \cdot \rcv_{\bot}(\ch^{m_1}_{j_1})
        \cdot \rcv_{\bot}(\ch^{m_2}_{j_2})
        \cdot \rcv_{\bot}(\ch^{m_3}_{j_3})
        \nonumber
    \end{equation}

    \item If $\gamma_2=$ true and $\gamma_3$= false, then 
    \begin{equation}
        B_j = \rcv_{\top}(\ch^{m_1}_{j_1}) \cdot \rcv_{\top}(\ch^{m_2}_{j_2}) 
        \cdot \rcv_{\bot}(\ch^{m_3}_{j_3})
        \cdot \rcv_{\top}(\ch^{m_3}_{j_3})
        \cdot \rcv_{\bot}(\ch^{m_1}_{j_1}) 
        \cdot \rcv_{\bot}(\ch^{m_2}_{j_2})
        \nonumber
    \end{equation}

    \item If $\gamma_2=$ false and $\gamma_3$= true, then 
    \begin{equation}
        B_j = \rcv_{\top}(\ch^{m_1}_{j_1}) \cdot \rcv_{\bot}(\ch^{m_2}_{j_2}) \cdot \rcv_{\top}(\ch^{m_2}_{j_2}) \cdot \rcv_{\top}(\ch^{m_3}_{j_3})  \cdot \rcv_{\bot}(\ch^{m_3}_{j_3}) \cdot \rcv_{\bot}(\ch^{m_1}_{j_1})
        \nonumber
    \end{equation}

    \item If $\gamma_2=$ false and $\gamma_3$= false, then 
    \begin{equation}
        B_j = \rcv_{\top}(\ch^{m_1}_{j_1}) \cdot \rcv_{\bot}(\ch^{m_2}_{j_2}) \cdot \rcv_{\top}(\ch^{m_2}_{j_2}) \cdot \rcv_{\bot}(\ch^{m_33}_{j_3}) \cdot \rcv_{\top}(\ch^{m_3}_{j_3}) \cdot \rcv_{\bot}(\ch^{m_1}_{j_1})
        \nonumber
    \end{equation}
\end{itemize}
One can easily verify each $B_j$ satisfies $\po{}$ and 
we show they also satisfy the reads-from relation for channel $\ch_{j_l}^{m_l}$. 
\begin{itemize}
    \item If $\gamma_l = x_{j_l}$ and $x_{j_l}$ is assigned to be true, then the send to $\ch_{j_l}^{m_l}$ in $A^1_{{j_l}}$ gets ordered before the send in $A^2_{{j_l}}$, 
    which is inline with $\rcv_{\top}(\ch_{j_l}^{m_l})$ getting ordered before $\rcv_{\bot}(\ch_{j_l}^{m_l})$.
    \item If $\gamma_l = x_{j_l}$ and $x_{j_l}$ is assigned to be false, then the send to $\ch_{j_l}^{m_l}$ in $A^1_{{j_l}}$ gets ordered after the send in $A^2_{{j_l}}$, 
    which is inline with $\rcv_{\top}(\ch_{j_l}^{m_l})$ getting ordered after $\rcv_{\bot}(\ch_{j_l}^{m_l})$. 
    \item If $\gamma_l = \neg x_{j_l}$ and $x_{j_l}$ is assigned to be true, then the send to $\ch_{j_l}^{m_l}$ in $A^1_{{j_l}}$ gets ordered after the send in $A^2_{{j_l}}$, 
    which is inline with $\rcv_{\top}(\ch_{j_l}^{m_l})$ getting ordered after $\rcv_{\bot}(\ch_{j_l}^{m_l})$.
    \item If $\gamma_l = \neg x_{j_l}$ and $x_{j_l}$ is assigned to be false, then the send to $\ch_{j_l}^{m_l}$ in $A^1_{{j_l}}$ gets ordered before the send in $A^2_{{j_l}}$, 
    which is inline with $\rcv_{\top}(\ch_{j_l}^{m_l})$ getting ordered before $\rcv_{\bot}(\ch_{j_l}^{m_l})$.
\end{itemize}
Therefore, the relative order of receive events to $\ch_{j_l}^{m_l}$ in the second phase matches the send events to $\ch_{j_l}^{m_l}$ in the first phase 
and $\trace$ is a valid concretization. 

\myparagraph{Correctness (Consistency $\Rightarrow$ Satisfiability)}
Now we show the reverse direction, 
i.e. if there is a concretization $\reordering$, 
then $\psi$ can be satisfied. 
First we construct the valuation function for each variable $x_i$ and 
then proceed to show this assignment makes $\psi$ true. 

We consider the events in the first phase, 
and can notice that for an arbitrary fixed $1 \le j \le n$, 
$\snd_{\top}(\ch_j^1)$ in $\thread_1$ gets ordered before $\snd_{\bot}(\ch_j^1)$ in $\thread_2$, 
iff $\snd_{\top}(\ch_j^k)$ in $\thread_1$ gets ordered before $\snd_{\bot}(\ch_j^k)$ in $\thread_2$ for all $1 \le k \le f_j$. 
This is because the channel $\lk$ behaves like a lock, 
so that $\snd(\ch_j^k)$ in $\thread_1, \thread_2$ must be executed atomically  (see \figref{lock-demo-capacity-one} for explanation). 
Our valuation function will assign $x_i = $ true iff in $\reordering$, $\snd_{\top}(\ch_j^1)$ in $\thread_1$ gets ordered before $\snd_{\bot}(\ch_j^1)$ in $\thread_2$. 

Now we proceed to show this assignment makes $\psi$ true. 
That is, we need to prove each clause $C_j$ is satisfied. 
By our encoding, the concretization of each clause is sequential, 
i.e. for any possible concretization, 
all events encoded for $C_j$ must be executed before events encoded for $C_{j'}$, s.t. $j' > j$. 
Events from different clauses cannot overlap, 
because of channels $\beta_1, \beta_2, \beta_3, \beta_4$. 
Then we pick an arbitrary clause $C_j$ and prove it is satisfied. 

For the encoding of $C_j$, 
the $\rf{}$ ensures that if $\rcv_{\top}(\ch^{m_l}_{j_l})$ is ordered before $\rcv_{\bot}(\ch^{m_l}_{j_l})$ in $\trace$, 
then the literal $\gamma_l$ corresponding to $x_{j_l}$ in clause $C_j$ will be true and otherwise false. 
We also guarantee that at least one of $\rcv_{\top}(\ch^{m_l}_{j_l})$ will be before $\rcv_{\bot}(\ch^{m_l}_{j_l})$ in $\trace$, 
otherwise, $\trace$ violates program order. 
Therefore, there is at least one literal in $C_j$ being assigned true and $C_j$ is satisfied. 
This completes the reduction. 
\end{proof}



\subsection{Proof for \secref{hardness-const-t-m}}

\applabel{proof-np-m=5-t=3-k=unbounded}


\vchRfConstTM*
\begin{proof}
We prove each direction separately.

\myparagraph{Proof of correctness (Consistency $\Rightarrow$ Satisfiability)}
Assuming there is a valid concretization $\trace$, 
we construct a valuation function that satisfies $\psi$,
by checking the relative order of events in the first phase. 
For variable $x_i$, we consider events in $I^1_i, I^2_i$, 
and assign $x_i$ to be true iff $\snd_{\top}^i(\ch_1) \trord{\tr} \snd_{\bot}^i(\ch_1)$. 
Now we prove this assignment makes an arbitrary clause true. 
Without loss of generality, we show an arbitrary clause $C_j$ can be satisfied. 
We assume $C_j = \gamma_1 \wedge \gamma_2 \wedge \gamma_3$, and 
the variable in $\gamma_1, \gamma_2, \gamma_3$ are $x_{j_1}, x_{j_2}, x_{j_3}$. 

The outline of the proof is the following. 
We have an observation that in $\trace$, 
for an arbitrary $q \in \set{1, 2, 3}$, $\rcv_{\top}(c_q) \trord{\tr} \rcv_{\bot}(c_q)$ iff $\gamma_q$ is true. 
Given the fact that this observation holds, 
if all $\gamma_q$ are false, 
then $\rcv_{\bot}(c_q) \trord{\tr} \rcv_{\top}(c_q)$ holds for all $q \in \set{1, 2, 3}$. 
In this case, $\trace$ is not a valid concretization, as it violates $\po{}$. 
Therefore, we have at least one of $\gamma_1, \gamma_2, \gamma_3$ is true, which satisfies $C_j$. 
Now in the following content, 
we prove $\rcv_{\top}(c_q) \trord{\tr} \rcv_{\bot}(c_q)$ iff $\gamma_q$ is true, 
and we first prove the following two lemmas. 

Firstly, we show that, for all $1 \le i \le \numVar$, 
in $I^1_i$ and $I^2_i$, 
$\snd_{\bot}^i(\ch_1) \trord{\tr} \snd_{\top}^i(\ch_1)$ iff $\snd_{\bot}^i(\ch_2) \trord{\tr} \snd_{\top}^i(\ch_2)$. 
If $\snd_{\bot}^i(\ch_1) \trord{\tr} \snd_{\top}^i(\ch_1)$ in $I^1_i, I^2_i$, 
then we have $\rcv_{\bot}^i(\ch_1)$ in $A^1_{1, i}$ is ordered before $\rcv_{\top}^i(\ch_1)$ in $A^2_{1, i}$. 
As $\rcv_{\bot}^i(\ch_2) \poord{\trace} \rcv_{\bot}^i(\ch_1)$ in $A^1_{1, i}$, 
and $\rcv_{\top}^i(\ch_1) \poord{\trace} \rcv_{\top}^i(\ch_2)$ in $A^2_{1, i}$, 
by transitivity, we have $\rcv_{\bot}^i(\ch_2) \trord{\tr} \rcv_{\top}^i(\ch_2)$, 
and thus $\snd_{\bot}^i(\ch_2) \trord{\tr} \snd_{\top}^i(\ch_2)$ in $I^1_i, I^2_i$. 
In the other direction, if $\snd_{\top}^i(\ch_1) \trord{\tr} \snd_{\bot}^i(\ch_1)$ in $I^1_i, I^2_i$,  
then since $\snd_{\bot}^i(\ch_1) \poord{\trace} \snd_{\bot}^i(\ch_2)$ in $I^1_i$, 
and $\snd_{\top}^i(\ch_2)$ is program ordered before $\snd_{\top}^i(\ch_1)$ in $I^2_i$, 
we have $\snd_{\top}^i(\ch_2) \trord{\trace} \snd_{\bot}^i(\ch_2)$ by transitivity. 

The same reasoning above can be used to show in $A^1_{j, l}$ and $A^2_{j, l}$, 
$\snd_{\bot}^l(\ch_1) \trord{\tr} \snd_{\top}^l(\ch_1)$ iff $\snd_{\bot}^l(\ch_2) \trord{\tr} \snd_{\top}^l(\ch_2)$ 
for all $j$, where $1 \le j \le \numClause$. 

Secondly, we show that for all $j$, s.t. $1 \le j \le \numClause$, 
in $A^1_{j, l}$ and $A^2_{j, l}$, 
$\snd^l_{\top}(\ch_1) \trord{\tr} \snd^l_{\bot}(\ch_1)$ iff in $A^1_{j-1, l}$ and $A^2_{j-1, l}$, 
$\snd^l_{\top}(\ch_1) \trord{\tr} \snd^l_{\bot}(\ch_1)$. 
That is, the assignments for values are consistent across phases. 
The reasoning is similar to previous one. 
If $\snd^l_{\top}(\ch_1) \trord{\tr} \snd^l_{\bot}(\ch_1)$ in $A^1_{j-1, l}$ and $A^2_{j-1, l}$, 
then we consider the events in $A^1_{j, l}$ and $A^2_{j, l}$. 
If $\snd^l_{\top}(\ch_1) \trord{\tr} \snd^l_{\bot}(\ch_1)$, 
since $\rcv^l_{\top}(\ch_1) \poord{\trace} \snd^l_{\top}(\ch_1)$ 
and $\snd^l_{\bot}(\ch_1) \poord{\trace} \rcv^l_{\bot}(\ch_1)$, 
by transitivity, 
we have $\rcv^l_{\top}(\ch_1) \trord{\trace} \rcv^l_{\bot}(\ch_1)$. 
Therefore, we have in $A^1_{j-1, l}$ and $A^2_{j-1, l}$, 
$\snd^l_{\top}(\ch_1) \trord{\tr} \snd^l_{\bot}(\ch_1)$. 
In the other direction, if in $A^1_{j-1, l}$ and $A^2_{j-1, l}$, 
$\snd^l_{\top}(\ch_1) \trord{\tr} \snd^l_{\bot}(\ch_1)$, 
then we have $\rcv^l_{\top}(\ch_1) \trord{\trace} \rcv^l_{\bot}(\ch_1)$ in $A^1_{j, l}, A^2_{j, l}$. 
As $\snd^i_{\top}(\ch_2) \poord{\trace} \rcv^l_{\top}(\ch_1)$, 
and $\rcv^l_{\top}(\ch_1) \poord{\trace} \rcv^l_{\bot}(\ch_1)$, 
by transitivity, we have $\snd^l_{\top}(\ch_2) \trord{\tr} \snd^l_{\bot}(\ch_2)$ in $A^1_{j, l}, A^2_{j, l}$. 
Then following the lemma we proved previously, 
we have in $A^1_{j, l}$ and $A^2_{j, l}$, 
$\snd^l_{\top}(\ch_1) \trord{\tr} \snd^l_{\bot}(\ch_1)$. 
Intuitively, this observation captures the fact that 
our valuation for every variable is properly maintained across phases. 
That is, in each phase, 
we will copy the valuation once and use copied valuation to satisfy the clause constraints. 

Combining these two lemmas, 
we prove that in each phase, 
$\snd_{\bot}(c_q) \trord{\tr} \snd_{\top}(c_q)$ iff $x_{j_q}$ = false. 
If $\snd_{\bot}(c_q) \trord{\tr} \snd_{\top}(c_q)$, then by transitivity, 
we have $\rcv_{\bot}(\ch_2) \trord{\tr} \rcv_{\top}(\ch_2)$ in $A^1_{j, j_q}, A^2_{j, j_q}$, 
which means $\snd_{\bot}^i(\ch_1) \trord{\tr} \snd_{\top}^i(\ch_1)$ in $I^1_i, I^2_i$, 
so that $x_{j_q}$ = false. 
In the other direction, 
if $x_{j_q}$ = false, then $\snd_{\bot}^i(\ch_1) \trord{\tr} \snd_{\top}^i(\ch_1)$ in $I^1_i, I^2_i$. 
By the previous lemmas, we have $\rcv_{\bot}^{j_q}(\ch_1) \trord{\tr} \rcv_{\top}^{j_q}(\ch_1)$ in $A^1_{j, j_q}, A^2_{j, j_q}$ 
and thus $\snd_{\bot}(c_q) \trord{\tr} \snd_{\top}(c_q)$. 

Then we are ready to show that $\rcv_{\top}(c_q) \trord{\tr} \rcv_{\bot}(c_q)$ iff $\gamma_q$ is true. 
If $\gamma_q = x_{j_q}$, then $\rcv_{\top}(c_q), \rcv_{\bot}(c_q)$ observe $\snd_{\top}(c_q), \snd_{\bot}(c_q)$, respectively. 
Since $\snd_{\bot}(c_q) \trord{\tr} \snd_{\top}(c_q)$ iff $x_{j_q}$ = false and $\gamma_q$ is true iff $x_{j_q}$ is true, 
we have $\snd_{\top}(c_q) \trord{\tr} \snd_{\bot}(c_q)$ iff $\gamma_q$ is true 
and thus $\rcv_{\top}(c_q) \trord{\tr} \rcv_{\bot}(c_q)$ iff $\gamma_q$ is true. 
Otherwise, if $\gamma_q = \neg x_{j_q}$, 
then $\rcv_{\top}(c_q), \rcv_{\bot}(c_q)$ observe $\snd_{\bot}(c_q), \snd_{\top}(c_q)$, respectively. 
Since $\snd_{\bot}(c_q) \trord{\tr} \snd_{\top}(c_q)$ iff $x_{j_q}$ = false and $\gamma_q$ is true iff $x_{j_q}$ is false, 
we have $\snd_{\bot}(c_q) \trord{\tr} \snd_{\top}(c_q)$ iff $\gamma_q$ is true 
and thus $\rcv_{\top}(c_q) \trord{\tr} \rcv_{\bot}(c_q)$ iff $\gamma_q$ is true. 

This completes one direction of the reduction. 

\myparagraph{Proof of correctness (Satisfiability $\Rightarrow$ Consistency)}
Now assuming there is a valuation function which satisfies $\psi$, 
we construct a valid concretization $\trace$. 
In general, $\trace$ is of the following pattern. 
\begin{equation}
    \trace = I \circ A_1 \circ \dots \circ A_{\numClause}
    \nonumber
\end{equation}
where $I$ is a linear sequence of $I^1, I^2, I^3$ and $A_i$ is a linear sequence of $A^1_i, A^2_i, A^3_i$ for all $1 \le i \le \numClause$. 
We first discuss the details of $I$. 
\begin{equation}
    I = I_1 \circ \dots I_{\numVar}
    \nonumber
\end{equation}
where $I_j$ is a linear sequence of $I^1_j$ and $I^2_j$. 
If $x_j$ is assigned to be true, then 
\begin{equation}
    I_j = \snd_{\bot}^j(\ch_1)
    \cdot \snd_{\bot}^j(\ch_2) \cdot \snd_{\top}^j(\ch_2) \cdot \snd_{\top}^j(\ch_1)
    \nonumber
\end{equation}
If $x_j$ is assigned to be false, then 
\begin{equation}
    I_j = \snd_{\top}^j(\ch_2) \cdot \snd_{\top}^j(\ch_1) \cdot \snd_{\bot}^j(\ch_1) \cdot  \snd_{\bot}^j(\ch_2)
    \nonumber
\end{equation}
It is obvious that the program order is satisfied. 
Next we discuss the details of $A_j$. 
Generally, all $A_j$ are of the following pattern. 
\begin{equation}
    A_j = A_{j, 1} \circ \dots \circ A_{j, n} \circ B_j
    \nonumber
\end{equation}
where $A_{j, l}$ is a linear sequence of $A^1_{j, l}, A^2_{j, l}$ 
and $B_j$ is a linear sequence of $B^1_j, B^2_j, B^3_j$. 
If $x_l$ is assigned to be true, then $A_{j, l}$ is of the following form. 
\begin{equation}
\begin{aligned}
    A_{j, l} &= \snd_{\top}^l(\ch_2) \cdot \rcv_{\top}^l(\ch_1) \cdot {\snd_{\top}(c_q)} \cdot  \rcv_{\top}^l(\ch_2) \cdot \snd_{\top}^l(\ch_1) \\
    & \; \cdot \snd_{\bot}^l(\ch_1) \cdot \rcv_{\bot}^l(\ch_2) \cdot  {\snd_{\bot}(c_q)} \cdot  \rcv_{\bot}^l(\ch_1) \cdot   \snd_{\bot}^l(\ch_2) 
    \nonumber
\end{aligned}
\end{equation}
where {$\snd_{\top}(c_q), \snd_{\bot}(c_q)$} exist iff $x_l$ appears as the $q$-th literal of clause $C_j$. 
Otherwise, if $x_l$ is assigned to be false, then 
\begin{equation}
\begin{aligned}
    A_{j, l} &= \snd_{\bot}^l(\ch_1) \cdot  \rcv_{\bot}^l(\ch_2) \cdot  {\snd_{\bot}(c_q)} \cdot  \rcv_{\bot}^l(\ch_1) \cdot  \snd_{\bot}^l(\ch_2) \\
    & \; \cdot \snd_{\top}^l(\ch_2) \cdot  \rcv_{\top}^l(\ch_1) \cdot {\snd_{\top}(c_q)} \cdot 
 \rcv_{\top}^l(\ch_2)  \cdot \snd_{\top}^l(\ch_1) 
    \nonumber
\end{aligned}
\end{equation}
where {$\snd_{\top}(c_q), \snd_{\bot}(c_q)$} exists if $x_l$ appears as the $q$-th literal of $C_j$. 
Finally, we discuss the linear sequence of $B^1_j, B^2_j, B^3_j$. 
Since $C_j = \gamma_1 \wedge \gamma_2 \wedge \gamma_3$ is satisfied, then at least one literal will be true. 
Without loss of generality, 
we assume $\gamma_1$ = true. 
Then we have four possible situations, depending on the value of $\gamma_2$ and $\gamma_3$. 

If $\gamma_2$ = true and $\gamma_3$ = true, then 
\begin{equation}
    C_j = \rcv_{\top}(c_1) \cdot \rcv_{\top}(c_2) \cdot \rcv_{\top}(c_3) \cdot \rcv_{\bot}(c_2) \cdot \rcv_{\bot}(c_3) \cdot \rcv_{\bot}(c_1)
    \nonumber
\end{equation}

If $\gamma_2$ = true and $\gamma_3$ = false, then 
\begin{equation}
    C_j = \rcv_{\top}(c_1) \cdot \rcv_{\top}(c_2) \cdot \rcv_{\bot}(c_2) \cdot \rcv_{\bot}(c_3) \cdot \rcv_{\top}(c_3) \cdot \rcv_{\bot}(c_1)
    \nonumber
\end{equation}

If $\gamma_2$ = false and $\gamma_3$ = true, then 
\begin{equation}
    C_j = \rcv_{\top}(c_1) \cdot \rcv_{\bot}(c_2) \cdot \rcv_{\top}(c_2) \cdot \rcv_{\top}(c_3) \cdot \rcv_{\bot}(c_3) \cdot \rcv_{\bot}(c_1)
    \nonumber
\end{equation}

If $\gamma_2$ = false and $\gamma_3$ = false, then 
\begin{equation}
    C_j = \rcv_{\top}(c_1) \cdot \rcv_{\bot}(c_2) \cdot \rcv_{\top}(c_2) \cdot \rcv_{\bot}(c_3) \cdot \rcv_{\top}(c_3) \cdot \rcv_{\bot}(c_1)
    \nonumber
\end{equation}

Firstly, it's easy to verify these linear sequences respect program order. 
To argue that they also satisfy the reads-from relation, 
we have the following observations. 
\begin{itemize}
    \item If $x_l$ is assigned to be true, then $\snd^l_{\top}(\ch_1) \trord{\tr} \snd^l_{\bot}(\ch_1)$ and $\snd^l_{\top}(\ch_2) \trord{\tr} \snd^l_{\bot}(\ch_2)$ in $A_{i, l}$ for any $l, i$. 

    \item If $x_l$ is assigned to be false, then $\snd^l_{\bot}(\ch_1) \trord{\tr} \snd^l_{\top}(\ch_1)$ and $\snd^l_{\bot}(\ch_2) \trord{\tr} \snd^l_{\top}(\ch_2)$ in $A_{i, l}$ for any $l, i$. 

    \item The way we assign reads-from relation for $\rcv_{\top}(c_q), \rcv_{\bot}(c_q)$ matches the order we send to $c_q$ in $A_j$ 
    and in $A_j$, 
    $\snd_{\bot}(c_q) \trord{\tr} \snd_{\top}(c_q)$ iff $x_{j_q}$ = false.
\end{itemize}
Therefore, the reads-from relation is also satisfied. 
This proves $\trace$ is indeed a valid concretization and thus the $\vchRf$ problem is consistent. 

\end{proof}



\subsection{Proof for \secref{two-threads-lower-bound-sec}}


\applabel{proof-lower-bound-t=2}
\ovcorrectness*
\begin{proof}
In the following proofs, we refer to the lexicographical order of pairs of indices $\tuple{i,j}$. A pair $\tuple{i_1,j_1}$ is lexicographically before $\tuple{i_2,j_2}$ if $i_1 < i_2$, or, in the case where $i_1 = i_2$, if $j_1 < j_2$. To say that $\tuple{i_1,j_1}$ is lexicographically before $\tuple{i_2,j_2}$, we write $\tuple{i_1, j_1} <_\textit{lex} \tuple{i_2, j_2}$, and we use $\leq_\textit{lex}$ as the reflexive closure of $<_\textit{lex}$. This can be extended to pairs of vectors $\tuple{a_i,b_j} \in A \times B$, referring to the indices of the vectors.

\myparagraph{Proof of correctness (Orthogonal pair $\Rightarrow$ Consistency)}
We first prove that if an orthogonal pair $a_i \in A, b_j \in B$ exists, the resulting total execution ${\AbstractExecution} = \tuple{\eventSet, \po{}, \rf{}}$ is consistent. Let $\tuple{a_{i_1}, b_{j_1}}$ be the lexicographically earliest pair of orthogonal vectors, and let $\tuple{a_{i_2}, b_{j_2}}$ be the lexicographically last pair of orthogonal vectors. We now define a partial order $<_\textit{sat}$ on the events of the total execution based on these vectors. Afterwards, we show that there exists a linearization of $<_\textit{sat}$ that is a well-formed execution.
It is defined by (the transitive closure of):
\begin{enumerate}
    \item $e_1 <_\textit{sat} e_2$ for all $e_1, e_2$ where $\tuple{e_1, e_2} \in (\po{} \cup \rf{})^+$.
    \item $\snd_{a_i}(\alpha) <_\textit{sat} \snd_{b_j}(\alpha)$ for all $\tuple{a_i,b_j} \in A \times B$, where $\tuple{i, j} \leq_\textit{lex} \tuple{i_1, j_1}$.
    \item $\rcv_{a_i}(\ch_k) <_\textit{sat} \rcv_{b_j}(\ch_k)$ for all $\tuple{a_i,b_j} \in A \times B$ and $1 \leq k \leq d$, where both events exist and $\tuple{i,j} <_\textit{lex} \tuple{i_1,j_1}$.
    \item $\rcv_{a_i}(\beta) <_\textit{sat} \rcv_b (\beta)$ for all $i < i_1$.
    \item $\rcv_{b_j}(\alpha) <_\textit{sat} \rcv_{a_i}(\alpha)$ for all $\tuple{a_i,b_j} \in A \times B$, where $\tuple{i_2, j_2} <_\textit{lex} \tuple{i,j}$.
    \item $\snd_{b_j}(\ch_k) <_\textit{sat} \snd_{a_i}(\ch_k)$ for all $\tuple{a_i,b_j} \in A \times B$ and $1 \leq k \leq d$, where both events exist and $\tuple{i_2, j_2} <_\textit{lex} \tuple{i,j}$.
    \item $\snd_b(\beta) <_\textit{sat} \snd_{a_i}(\beta)$ for all $i \geq i_2$.
\end{enumerate}

To verify that this is indeed a (strict) partial order, we need to show that it is asymmetric, or, if we think of the execution as a graph, acyclic. To show this, we first show that if there is a cycle, then there is a cycle of the following form: First a $\po{}$ step in $\thread_1$, then a step from rules (2)-(4), followed by a $\po{}$ step in $\thread_2$, and finally a step from (5)-(7).

It is trivial from the construction that (1) itself does not create a cycle, so a rule from (2)-(7) is needed, but these are all edges between $\thread_1$ and $\thread_2$. Furthermore, a $\rf{} \setminus \po{}$ edge cannot be part of the cycle for the following reason: The only two cross-thread $\rf{}$ edges are on the $\gamma$ and $\delta$ channels. Examining the read of $\gamma$, the only edge from $\thread_2$ to $\thread_1$ that could form a cycle based on this read would have to go from $\rcv_{b_1}(\alpha)$ to $\rcv_{a_1}(\alpha)$, since only rule (5) could apply. But this is impossible, since $\tuple{a_1, b_1}$ is the lexicographically first pair. Next, the read of $\delta$ fails for a similar reason: The only possible cycle caused by this read would be from $\rcv_{a_n}(\ch_k)$ to $\rcv_{b_n}(\ch_k)$ for some $k$, but only rule (3) could apply, which is impossible, since $\tuple{a_n, b_n}$ is lexicographically last.
Therefore, we need a non-$\rf{}$ edge from $\thread_1$ to $\thread_2$ and a non-$\rf{}$ edge from $\thread_2$ to $\thread_1$, which, by inspection, can only come from (2)-(4) and (5)-(7), respectively. We can notice that no edges go both ways on the same combination of event type (read or write) and channel, so we need a $\po{}$ step on both threads. Finally, if there is a cycle with more than four events, it is easy to convince oneself that it is possible to find a subset of four of those events that also form a cycle. 

We can now look at all combinations of rules between (2)-(4) and (5)-(7) to show that none of them can cause a cycle. Two of the cases are not possible due to the source event kind (event type and channel) of the (2)-(4) edge only ever appearing before the destination event kind of the (5)-(7) edge in the construction. This is the case for (2) and (5) as well as (2) and (7). Similarly, sometimes the destination event kind of the (2)-(4) edge only appears after the source event kind of the (5)-(7) edge. These pairs are (3) and (6), (4) and (6), and (4) and (7). We look at the remaining pairs below:
\begin{description}
    \item[(2) and (6).] We must have $\snd_{a_i}(\alpha) <_\textit{sat} \snd_{b_j}(\alpha)$ and $\snd_{b_{j'}}(\ch_k) <_\textit{sat} \snd_{a_{i'}}(\ch_k)$ for some $i$, $i'$, $j$, $j'$, and $k$. Due to the $\po{}$ ordering, we have $i' \leq i$ and $j' \leq j$. This contradicts $\tuple{i,j} \leq_\textit{lex} \tuple{i_1,j_1} \leq_\textit{lex} \tuple{i_2,j_2} <_\textit{lex} \tuple{i', j'}$.
    \item[(3) and (5).] We have $\rcv_{a_i}(\ch_k) <_\textit{sat} \rcv_{b_j}(\ch_k)$ and $\rcv_{b_{j'}}(\alpha) <_\textit{sat} \rcv_{a_{i'}}(\alpha)$ for some $i$, $i'$, $j$, $j'$, and $k$, where $i' \leq i$ and $j' \leq j + 1$ by the $\po{}$ ordering. This contradicts $\tuple{i,j} <_\textit{lex} \tuple{i_1,j_1} \leq_\textit{lex} \tuple{i_2,j_2} <_\textit{lex} \tuple{i', j'}$, since $\tuple{i', j'}$ can at most be $\tuple{i, j+1}$, but there has to be a pair between $\tuple{i,j}$ and $\tuple{i',j'}$.
    \item[(3) and (7).] We have $\rcv_{a_i}(\ch_k) <_\textit{sat} \rcv_{b_j}(\ch_k)$ and $\snd_b(\beta) <_\textit{sat} \snd_{a_{i'}}(\beta)$ for $j = n$ and some $i$, $i'$, and $k$, where $i_2 \leq i' \leq i$ by $\po{}$. This contradicts $\tuple{i, n} = \tuple{i, j} <_\textit{lex} \tuple{i_1, j_1}$.
    \item[(4) and (5).] Finally, we have $\rcv_{a_i}(\beta) <_\textit{sat} \rcv_b (\beta)$ and $\rcv_{b_j}(\alpha) <_\textit{sat} \rcv_{a_{i'}}(\alpha)$ for some $i$, $i'$, $j$, and $k$, where $i' - 1 \leq i < i_1$ by $\po{}$, and thereby $i' \leq i_1$. This contradicts $\tuple{i_2, j_2} <_\textit{lex} \tuple{i', j}$, unless $i' = i_1 = i_2$ and $j > j_2$. But the $\po{}$ ordering of $\rcv_b (\beta)$ before $\rcv_{b_j}(\alpha)$ means that $j=1$, and thus, $j \leq j_2$. 
\end{description}

With this, we have shown that $<_\textit{sat}$ is indeed a partial order. What remains is to show that there is a linearization $\tr$ of $<_\textit{sat}$ (respecting ${\AbstractExecution}$) that is a well-formed execution.
To do so, we first show that $<_\textit{sat}$ is \emph{saturated}, i.e. $\snd_1(\ch) <_\textit{sat} \snd_2(\ch)$ iff $\rcv_1(\ch) <_\textit{sat} \rcv_2(\ch)$. The method will be to look at each rule in turn and verify the property for all events ordered by this rule.
Note that the property is easy to verify if the two events are on the same thread, since, except for $\gamma$ and $\delta$ (which are trivial), all reads are on the same channels as the writes.
Therefore, we look at rule ($i$), ordering $e_1 <_\textit{sat} e_2$ across both threads, and consider all pairs of events where the first is $\po{}$ before $e_1$ (including $e_1$) and the second is $\po{}$ after $e_2$ (including $e_2$).
The reason we do not have to consider e.g. other events $(<_\textit{sat} \setminus\ \po{})$-before $e_1$ is that such events on the same thread as $e_2$, and vice versa for events after $e_2$.

\begin{enumerate}
    \item Per the reasoning above, it is sufficient to consider $\tuple{e_1, e_2} \in (\rf{} \setminus \po{})$. Considering $(\snd(\gamma) <_\textit{sat} \rcv(\gamma))$ first, we can see that $\rcv_{a_1}(\alpha)$ and $\rcv_{b_1}(\alpha)$ are ordered, so we have to show that $\snd_{a_1}(\alpha) <_\textit{sat} \snd_{b_1}(\alpha)$. This follows immediately from rule (2). Next, we consider $(\snd(\delta) <_\textit{sat} \rcv(\delta))$, which orders $\rcv_{b_n}(\ch_k)$ before $\rcv_{a_n}(\ch_k)$ for all $k$ where both events exist. If there is such a $k$, $a_n$ and $b_n$ must not be orthogonal. Thus, $\tuple{i_2, j_2} <_\textit{lex} \tuple{n,n}$, and the rest follows from rule (6).
    
    \item $\snd_{a_i}(\alpha) <_\textit{sat} \snd_{b_j}(\alpha)$: We have to show $\rcv_{a_i}(\alpha) <_\textit{sat} \rcv_{b_j}(\alpha)$. We look at three cases for the value of $\tuple{i,j}$: $\tuple{1, 1}$, $\tuple{i',1}$, and $\tuple{i, j'}$, where $i' \neq 1$ and $j' \neq 1$.
    \begin{itemize}
        \item[$\tuple{1,1}$] Follows immediately from the read of $\gamma$.
        \item[$\tuple{i',1}$] It follows from $\tuple{i', 1} \leq_\textit{lex} \tuple{i_1, j_1}$ that $i'-1 < i_1$. The ordering then follows transitively from rule (4).
        \item[$\tuple{i,j'}$] From $\tuple{i,j'} \leq_\textit{lex} \tuple{i_1, j_1}$ we get that $\tuple{i, j'-1} <_\textit{lex} \tuple{i_1, j_1}$, which, by transitivity, gives the correct ordering from rule (3).
    \end{itemize}
    We also have to consider anything $\po{}$-before $\snd_{a_i}(\alpha)$ against anything $\po{}$-after $\snd_{b_j}(\alpha)$. This may include other writes to $\alpha$, but these cases are already covered transitively by what we have shown. For writes to $\ch_k$ we still need to show that $\rcv_{a_i}(\ch_k) <_\textit{sat} \rcv_{b_j}(\ch_k)$ (other such writes less than $\tuple{i,j}$ are covered transitively). The fact that $\snd_{a_i}(\ch_k)$ and $\snd_{b_j}(\ch_k)$ both exist means that $a_i$ and $b_j$ are not orthogonal, which means that $\tuple{i_1,j_1} \neq \tuple{i,j}$. Together with $\tuple{i,j} \leq_\textit{lex} \tuple{i_1,j_1}$, this means that rule (3) applies, and we are done.
    
    \item $\rcv_{a_i}(\ch_k) <_\textit{sat} \rcv_{b_j}(\ch_k)$: We must show $\snd_{a_i}(\ch_k) <_\textit{sat} \snd_{b_j}(\ch_k)$. Rule (2) can be applied immediately, which transitively gives the right ordering. 
    Reads of $\alpha$ and reads/writes for $\beta$ could also be ordered by this rule. For $\alpha$, we have to show $\snd_{a_i}(\alpha) <_\textit{sat} \snd_{b_{j+1}}(\alpha)$ (for $j < n$). This follows from rule (2), since $\tuple{i,j} <_\textit{lex} \tuple{i_1, j_1}$, so $\tuple{i,j+1} \leq_\textit{lex} \tuple{i_1, j_1}$.
    For reads of $\beta$ we have to show $\snd_{a_{i-1}}(\beta) <_\textit{sat} \snd_{b}(\beta)$ (for $i > 1$). It holds that $\tuple{i-1,n} <_\textit{lex} \tuple{i,j} <_\textit{lex} \tuple{i_1, j_1}$, so we can apply rule (3) to get $\rcv_{a_{i-1}}(\ch_k) <_\textit{sat} \rcv_{b_n}(\ch_k)$, from which the ordering follows transitively. Finally, for writes to $\beta$ (where $i > 1$ and $j = n$), we have to show $\rcv_{a_i}(\beta) <_\textit{sat} \rcv_b(\beta)$. $\tuple{i,n} = \tuple{i,j} <_\textit{lex} \tuple{i_1,j_1}$ implies that $i < i_1$, so rule (4) gives us the ordering.

    \item $\rcv_{a_i}(\beta) <_\textit{sat} \rcv_b(\beta)$: We have to show $\snd_{a_i}(\beta) <_\textit{sat} \snd_{b}(\beta)$. From $i < i_1$ we get $\tuple{i,n} <_\textit{lex} \tuple{i_1, i_2}$, so the ordering follows transitively from (3). This also transitively orders $\rcv_{a_{i+1}}(\alpha)$ with $\rcv_{b_1}(\alpha)$ (for $i < n$), so we show $\snd_{a_{i+1}}(\alpha) <_\textit{sat} \snd_{b_1}(\alpha)$. Since $i+1 \leq i_1$, $\tuple{i+1,1} \leq_\textit{lex} \tuple{i_1, j_1}$, which means that rule (2) applies.

    \item $\rcv_{b_j}(\alpha) <_\textit{sat} \rcv_{a_i}(\alpha)$: We must show $\snd_{b_j}(\alpha) <_\textit{sat} \snd_{a_i}(\alpha)$. Since $\tuple{i_2, j_2} <_\textit{lex} \tuple{i,j}$, $a_i$ and $b_j$ are not orthogonal, which means that $\snd_{b_j}(\ch_k)$ and $\snd_{a_i}(\ch_k)$ exist for some $k$. Furthermore, these are ordered by rule (6), which transitively orders the $\alpha$ writes.
    Orderings between reads of $\alpha$ can transitively order reads of $\ch_k$ (for some $k$) as well as reads/writes for $\beta$.
    For $\ch_k$, we must show $\snd_{b_{j-1}}(\ch_k) <_\textit{sat} \snd_{a_i}(\ch_k)$ (for $j > 1$). Since $a_i$ and $b_{j-1}$ have reads to $k$ in common, they are not orthogonal. From this, along with $\tuple{i,j-1}$ being the pair just before $\tuple{i,j}$, we get $\tuple{i_2,j_2} <_\textit{lex} \tuple{i, j-1}$, which means rule (6) applies.
    Looking now at $\rcv_b(\beta) <_\textit{sat} \rcv_{a_{i-1}}(\beta)$ (for $j = 1$ and $i > 1$), the ordering on writes follows from rule (7), since $\tuple{i_2, j_2} <_\textit{lex} \tuple{i,1}$ and, hence, $i - 1 \geq i_2$.
    Finally, writes to $\beta$ can also be ordered, so we must show $\rcv_b(\beta) <_\textit{sat} \rcv_{a_i}(\beta)$ (for $i < n$). We instead show $\rcv_{b_1}(\alpha) <_\textit{sat} \rcv_{a_{i+1}}(\alpha)$, from which the ordering follows transitively. This follows by rule (5), since $\tuple{i_2, j_2} <_\textit{lex} \tuple{i,j} <_\textit{lex} \tuple{i+1, 1}$.

    \item $\snd_{b_j}(\ch_k) <_\textit{sat} \snd_{a_i}(\ch_k)$: We first show $\rcv_{b_j}(\ch_k) <_\textit{sat} \rcv_{a_i}(\ch_k)$. There are three cases for $\tuple{i,j}$: $\tuple{n,n}$, $\tuple{i',n}$ and $\tuple{i,j'}$, where $i' \neq n$ and $j' \neq n$.
    \begin{itemize}
        \item[$\tuple{n,n}$] Follows immediately from the read of $\delta$.
        \item[$\tuple{i',n}$] It follows from $\tuple{i_2,j_2} <_\textit{lex} \tuple{i',n}$ that $i' \geq i_2$. From here, we apply rule (7) and get the ordering transitively.
        \item[$\tuple{i,j'}$] We have $\tuple{i_2,j_2} <_\textit{lex} \tuple{i,j'} <_\textit{lex} \tuple{i,j'+1}$, which means rule (5) applies, transitively ordering the reads to $\ch_k$.
    \end{itemize}
    The ordering of writes to $\ch_k$ may also order writes to $\alpha$, but $\rcv_{b_j}(\alpha) <_\textit{sat} \rcv_{a_i}(\alpha)$ follows immediately from rule (5).

    \item $\snd_b(\beta) <_\textit{sat} \snd_{a_i}(\beta)$: We must show that $\rcv_b(\beta) <_\textit{sat} \rcv_{a_i}(\beta)$, which we do by showing that $\rcv_{b_1}(\alpha) <_\textit{sat} \rcv_{a_{i+1}}(\alpha)$.
    From $i \geq i_2$ it follows that $\tuple{i_2, j_2} <_\textit{lex} \tuple{i+1,1}$, which means that rule (5) can be applied to get the desired ordering. A read $\rcv_{b_n}(\ch_k)$ could be ordered before another read $\rcv_{a_i}(\ch_k)$, so we show that $\snd_{b_n}(\ch_k) <_\textit{sat} \snd_{a_i}(\ch_k)$.
    It follows from $i \geq i_2$ that $\tuple{i_2, j_2} \leq_\textit{lex} \tuple{i, n}$. But the existence of $\rcv_{b_n}(\ch_k)$ and $\rcv_{a_i}(\ch_k)$ means that $a_i$ and $b_n$ are not orthogonal, so $\tuple{i_2, j_2} <_\textit{lex} \tuple{i, n}$, and rule (6) applies.
\end{enumerate}

Define $\tr$ as follows: Given two events $e_1, e_2 \in S$, order $e_1$ before $e_2$ if $e_1 <_\textit{sat} e_2$ and vice versa, otherwise order the event from $\thread_1$ first (if they are on the same thread, they are ordered by $<_\textit{sat}$). This is a total order because it is the order you get by greedily picking events from $\thread_1$ as long as no unpicked event from $\thread_2$ is $<_\textit{sat}$-before.

We have to show that for each channel $\ch \in \channels{\tr}$ the $i$-th read of $\proj{\tr}{\ch}$ reads the $i$-th write. To do so, we show the equivalent property that, if $\snd_1(\ch) \stricttrord{\tr} \snd_2(\ch)$, then $\rcv_1(\ch) \stricttrord{\tr} \rcv_2(\ch)$.
Let $\snd_1, \snd_2 \in \proj{\tr}{\snd(\ch)}$ be two writes such that $\snd_1 \stricttrord{\tr} \snd_2$. If $\snd_1 <_\textit{sat} \snd_2$, then $\rcv_1 <_\textit{sat} \rcv_2$ and thus $\rcv_1 \stricttrord{\tr} \rcv_2$, since $<_\textit{sat}$ is saturated.
If $\snd_1 \not<_\textit{sat} \snd_2$, assume for contradiction that $\rcv_1$ and $\rcv_2$ are ordered by $<_\textit{sat}$. Either ordering ($\rcv_1 <_\textit{sat} \rcv_2$ or $\rcv_2 <_\textit{sat} \rcv_1$) would imply that $\snd_1$ and $\snd_2$ are also ordered, since $<_\textit{sat}$ is saturated. 
In the case of $\snd_1 <_\textit{sat} \snd_2$ this directly contradicts the premise, and $\snd_2 <_\textit{sat} \snd_1$ contradicts $\snd_1 \stricttrord{\tr} \snd_2$. Hence, both the reads and the writes are unordered by $<_\textit{sat}$.
In all channels except $\gamma$ and $\delta$ the reads are on the same threads as the corresponding writes, thus they will be ordered the same in $\tr$. Lastly, the $\gamma$ and $\delta$ channels only have one write each, so the property is trivial for these.

\myparagraph{Proof of correctness (Consistency $\Rightarrow$ Orthogonal pair)}
Next, we prove that if the total execution is consistent, there is an orthogonal pair. To show this, we prove the contra-positive statement: If there are no orthogonal pairs, the total execution is not consistent. More specifically, we will show that the lack of an orthogonal pair leads to the derivation of a cyclic ordering between $\snd_{a_n}(\alpha)$ and $\snd_{b_n}(\alpha)$ by saturation.

To show one direction, we prove by induction that $\snd_{a_i}(\alpha)$ is ordered before $\snd_{b_j}(\alpha)$ for all $i$ and $j$ less than $n$. The induction is in the lexicographical order of $\tuple{i,j}$, i.e. $\tuple{1,1}$ is the base case, and the next element after $\tuple{i,j}$ is either $\tuple{i,j+1}$ if $j \neq n$ or $\tuple{i+1,1}$ if $j = n$.
\begin{itemize}
    \item \textbf{Base Case:} Initially, $\rcv_{a_1}(\alpha)$ is ordered before $\rcv_{b_1}(\alpha)$ through the read of $\gamma$, which orders $\snd_{a_1}(\alpha)$ before $\snd_{b_1}(\alpha)$ by saturation.
    \item \textbf{Induction:} We prove that $\snd_{a_i}(\alpha)$ is ordered before $\snd_{b_j}(\alpha)$ (for $\tuple{i,j} \neq \tuple{1,1}$). We do case distinction on (1) if $j = 1$ or (2) if $j \neq 1$.
    \begin{enumerate}
        \item If $j = 1$, the induction hypothesis states that $\snd_{a_{i-1}}(\alpha)$ is ordered before $\snd_{b_n}(\alpha)$. Since $a_{i-1}$ and $b_n$ are not orthogonal, they must both have a 1 in some dimension $k$, and there are therefore writes $\snd_{a_{i-1}}(\ch_k)$ and $\snd_{b_n}(\ch_k)$.
        By the induction hypothesis, these writes are ordered, which, by saturation, orders $\rcv_{a_{i-1}}(\ch_k)$ before $\rcv_{b_n}(\ch_k)$.
        This orders $\snd_{a_{i-1}}(\beta)$ before $\snd_{b}(\beta)$. Applying saturation, the reads of $\alpha$ for $a_i$ and $b_1$ are thus ordered.
        By a final application of saturation, $\snd_{a_{i}}(\alpha)$ is thus ordered before $\snd_{b_1}(\alpha)$.
        \item In the case of $j \neq 1$, the induction hypothesis states that $\snd_{a_i}(\alpha)$ is ordered before $\snd_{b_{j-1}}(\alpha)$. As before, $a_i$ and $b_{j-1}$ are not orthogonal, so for some $k$, $\snd_{a_i}(\ch_k)$ and $\snd_{b_{j-1}}(\ch_k)$ exist and are ordered.
        By saturation, $\rcv_{a_i}(\ch_k)$ is ordered before $\snd_{b_{j-1}}(\ch_k)$, which orders $\rcv_{a_i}(\alpha)$ before $\rcv_{b_j}(\alpha)$. A final application of saturation then gives the desired ordering.
    \end{enumerate} 
\end{itemize}

The last thing to show is that $\snd_{b_n}(\alpha)$ is ordered before $\snd_{a_n}(\alpha)$. The read of $\delta$ orders $\rcv_{b_n}(\ch_k)$ before $\rcv_{a_n}(\ch_k)$ for some $k$ (both of which exist, since $a_n$ and $b_n$ are not orthogonal). The proof is then concluded by an application of saturation.

\end{proof}

%


\section{Details in evaluation}
\applabel{evaluation-details}
\subsection{$\smt$ encodings}
\applabel{smt-encoding}
Given a $\vchRf$ input $\tuple{\AbstractExecution, \cpFunc, \rf{}}$, 
where ${\AbstractExecution} = \tuple{\eventSet, \po{}}$. 
We encode a $\smt$ formula $\psi$, s.t. $\psi$ is satisfiable iff $\AbstractExecution$ is consistent. 

We first discuss the variables in $\psi$. 
For each event $e \in \eventSet$, we allocate an integer variable $0 \le x_e \le \NumEvents-1$, where $\NumEvents$ is the total number of events in $\AbstractExecution$. 
$x_e$ indicates the position of $e$ in a possible concretization. 
More over, for each channel $\ch$, 
we allocate $2\NumEvents + 2$ variables $y^{\ch}_{\snd, i}$ and $y^{\ch}_{\rcv, i}$, where $0 \le i \le n$. 
Intuitively, $y^{\ch}_{\snd, i}$ and $y^{\ch}_{\rcv, i}$ stands for the total number of send / receive events to $\ch$ at the prefix with $i$ events of a possible concretization. 

Now we discuss the content of $\psi$. 
At a high level, $\psi$ can be decomposed into several components. 
\begin{equation}
    \psi = \psi_{unique} \wedge \psi_{porf} \wedge \psi_{\fifo} \wedge \psi_{cap}
    \nonumber
\end{equation}
where $\psi_{unique}$ ensures for all $e \in \eventSet$, $x_e$ is between $0$ and $\NumEvents - 1$, and $x_e$ is unique among all events (i.e. if $e \neq e'$, then $x_e \neq x_{e'}$). 
$\psi_{porf}$ ensures program order and reads-from relation. 
$\psi_{\fifo}$ ensures the $\fifo$ property of channels. 
$\psi_{cap}$ ensures capacity constraints of channels. 

$\psi_{unique}$ is of the following form.
\begin{equation}
    \psi_{unique} = (\bigwedge_{e \in \eventSet}{0 \le x_e \le \NumEvents - 1}) \wedge (\bigwedge_{e, e' \in \eventSet, \; e \neq e'}{x_e \neq x_{e'}})
    \nonumber
\end{equation}

Recall that we use $\sucr{}{e}$ to denote the immediate thread successor of $e$. 
$\psi_{porf} = \psi^{po}_{porf} \wedge \psi^{rf-sync}_{porf} \wedge \psi^{rf-async}_{porf}$, where 
\begin{equation}
    \psi^{po}_{porf} = \bigwedge_{e \in \eventSet, \; \sucr{}{e} \neq \bot}{x_e < x_{\sucr{}{e}}}
    \nonumber
\end{equation}
\begin{equation}
    \psi^{rf\text{-}sync}_{porf} = \bigwedge_{(\snd(\ch), \rcv(\ch) \in \rf{}, \cpFunc(\ch) = 0}{x_{\snd(\ch)} + 1 = x_{\rcv(\ch)}}
    \nonumber
\end{equation}

\begin{equation}
    \psi^{rf\text{-}async}_{porf} = \bigwedge_{(\snd(\ch), \rcv(\ch) \in \rf{}, \cpFunc(\ch) > 0}{x_{\snd(\ch)} < x_{\rcv(\ch)}}
    \nonumber
\end{equation}
$\psi^{po}_{porf}$ ensures program order. 
$\psi^{rf\text{-}sync}_{porf}$ requires that for any synchronous channel $\ch$, all receive events should be immediately after its matching send event. 
$\psi^{rf\text{-}async}_{porf}$ requires that for any asynchronous channel $\ch$, all receive events should be after its matching send event, but there can be some other events in between. 

$\psi_{\fifo} = \psi^{matched}_{\fifo} \wedge \psi^{unmatched}_{\fifo}$, where  
\begin{equation}
    \psi^{matched}_{\fifo} = \bigwedge_{\ch \in \Channels}{\bigwedge_{
    \scriptsize
    \begin{aligned}
    \begin{array}{c}
    (\snd_1(\ch), \rcv_1(\ch)) \in \rf{} \\
    (\snd_2(\ch), \rcv_2(\ch)) \in \rf{} \\ 
    \snd_1(\ch) \neq \snd_2(\ch)
    \end{array}
    \end{aligned}
    }
    {
    \begin{aligned}
    \begin{array}{c}
    (x_{\snd_1(\ch)} < x_{\snd_2(\ch)} \wedge x_{\rcv_1(\ch)} < x_{\rcv_2(\ch)}) \vee \\
    (x_{\snd_1(\ch)} > x_{\snd_2(\ch)} \wedge x_{\rcv_1(\ch)} > x_{\rcv_2(\ch)})
    \end{array}
    \end{aligned}
    }}
    \nonumber
\end{equation}

\begin{equation}
    \psi^{unmatched}_{\fifo} = \bigwedge_{\ch \in \Channels}{\bigwedge_{
    \scriptsize
    \begin{aligned}
    \begin{array}{c}
    (\snd_1(\ch), \rcv_1(\ch)) \in \rf{} \\
    \snd_2(\ch) \; is \; unmatched
    \end{array}
    \end{aligned}
    }
    {
    \begin{aligned}
    \begin{array}{c}
    x_{\snd_1(\ch)} < x_{\snd_2(\ch)}
    \end{array}
    \end{aligned}
    }}
    \nonumber
\end{equation}
In other words, $\psi^{matched}_{\fifo}$ requires that for every channel $\ch$, for every two distinct send/receive pairs $(\snd_1(\ch), \rcv_1(\ch)), (\snd_2(\ch), \rcv_2(\ch))$, 
either $\snd_1(\ch)$ is before $\snd_2(\ch)$ and $\rcv_1(\ch)$ is before $\rcv_2(\ch)$, 
or $\snd_1(\ch)$ is after $\snd_2(\ch)$ and $\rcv_1(\ch)$ is after $\rcv_2(\ch)$. 
This encoding exactly captures the $\fifo$ property of channels. 
Moreover $\psi^{unmatched}_{\fifo}$ requires that for any channel $\ch$, all unmatched sends to $\ch$ should be ordered after the matched sends to $\ch$.

$\psi_{cap} = \psi^{cap}_{cap} \wedge \psi^{\snd}_{cap} \wedge \psi^{\rcv}_{cap}$, 
where 
\begin{equation}
    \psi^{cap}_{cap} = \bigwedge_{0 \le i \le n, \; \ch \in \Channels}{y^{\rcv}_{i} \le y^{\snd}_{i} \le y^{\rcv}_{i} + \cpFunc(\ch)}
    \nonumber
\end{equation}

\begin{equation}
    \psi^{\snd}_{cap} = (\bigwedge_{\ch \in \Channels}{y^{\ch}_{\snd, 0} = 0}) \wedge (\bigwedge_{0 \le i \le n-1, \ch \in \Channels}{
    \begin{aligned}
    \begin{array}{c}
    ((\exists \snd(\ch) \in \eventSet, x_{
    \snd(\ch)} = i) \wedge y^{\ch}_{\snd, i} + 1 = y^{\ch}_{\snd, i+1}) \vee \\
    ((\nexists \snd(\ch) \in \eventSet, x_{\snd(\ch)} = i) \wedge y^{\ch}_{\snd, i} = y^{\ch}_{\snd, i+1})
    \end{array}
    \end{aligned}
    })
    \nonumber
\end{equation}

\begin{equation}
    \psi^{\rcv}_{cap} = (\bigwedge_{\ch \in \Channels}{y^{\ch}_{\rcv, 0} = 0}) \wedge (\bigwedge_{0 \le i \le n-1, \ch \in \Channels}{
    \begin{aligned}
    \begin{array}{c}
    ((\exists \rcv(\ch) \in \eventSet, x_{
    \rcv(\ch)} = i) \wedge y^{\ch}_{\rcv, i} + 1 = y^{\ch}_{\rcv, i+1}) \vee \\
    ((\nexists \rcv(\ch) \in \eventSet, x_{\rcv(\ch)} = i) \wedge y^{\ch}_{\rcv, i} = y^{\ch}_{\rcv, i+1})
    \end{array}
    \end{aligned}
    })
    \nonumber
\end{equation}

$\psi^{cap}_{cap}$ explicitly encodes the capacity constraints, 
i.e., at any prefix of a possible concretization, for any channel $\ch$, we require $num_{\rcv(\ch)} \le num_{\snd(\ch)} \le num_{\rcv(\ch)} + \cpFunc(\ch)$, 
where $num_{\rcv(\ch)}$ and $num_{\snd(\ch)}$ denote the number of receive/send events to $\ch$ in this prefix. 
$\psi^{\snd}_{cap}$ poses constraints on $y^{\snd}_{cap}$, 
where we require (1) $y^{\ch}_{\snd, 0}$ equals 0 for the prefix with no events, (2) $y^{\ch}_{\snd, i} + 1= y^{\ch}_{\snd, i+1}$, if there exists a send to $\ch$ whose position is $i$, and (3) $y^{\ch}_{\snd, i} = y^{\ch}_{\snd, i+1}$, if no send event to $\ch$ is located at position $i$. 
Similarly encoding $\psi^{\rcv}_{cap}$ are also applied to $y^{\ch}_{\rcv, i}$. 

Now we show $\psi$ is satisfiable iff $\tuple{\AbstractExecution, \cpFunc, \rf{}}$ is consistent.

\myparagraph{Satisfiability $\Rightarrow$ consistency}
Assuming $\psi$ is satisfiable, then each event $e \in \eventSet$ must have been assigned a unique index $0 \le x_e \le \NumEvents - 1$, as required by $\psi_{unique}$. 
We claim that a valid concretization $\trace$ of $\AbstractExecution$ can be obtained by ordering all events by their assigned integer variable, i.e. the $i$-th event in $\trace$ is the event $e$, s.t. $x_e = i$. 
Clearly, $\trace$ satisfies program order, as in $\psi_{porf}$, we require every event to be ordered before their immediate thread successors. 
Secondly, $\trace$ satisfies the capacity constraints for channels. 
Following the encoding of $\psi^{\snd}_{cap}$ and $\psi^{\rcv}_{cap}$, for any channel $\ch$, 
$y^{\ch}_{\snd, i}$ and $y^{\ch}_{\rcv, i}$ represent the number of send/receive events to $\ch$ in the prefix $\pi$ of $\trace$, where $\pi$ contains $i$ events. 
Then $\psi^{cap}_{cap}$ explicitly ensures the capacity constraints.
Lastly, we show $\rf{\trace} = \rf{}$. 
For synchronous channels, the receive event is immediately after its matching send event, as required by $\psi^{rf\text{-}sync}_{porf}$. 
Therefore, the reads-from constraints is satisfied. 
For asynchronous channels, $\psi^{rf\text{-}async}_{porf}$ guarantees every receive event is after its matching send event. 
Moreover, $\psi^{unmatched}_{\fifo}$ ensures for any channel $\ch$, every unmatched send to $\ch$ is ordered after every match send event. 
Finally, $\psi^{matched}_{\fifo}$ explicitly encodes the $\fifo$ property of all send/receive pairs for all channels. 
Therefore, we conclude $\trace$ is indeed a valid concretization. 

\myparagraph{Consistency $\Rightarrow$ satisfiability}
Now we show if $\tuple{\AbstractExecution, \cpFunc, \rf{}}$ is consistent, then $\psi$ is satisfiable. 
This direction is easier. 
We take an arbitrary valid concretization $\trace$ of $\AbstractExecution$. 
Based on $\trace$, we assign $x_e = i$ iff $e$ is the $i$-th event in $\trace$. 
Moreover, we assign $y^{\ch}_{\snd, i} = j$ ($y^{\ch}_{\rcv, i} = j$) iff there are $j$ send (receive) events to $\ch$ in the prefix $\pi$ of $\trace$, 
where $\pi$ contains $i$ events. 
Following the definition of valid concretization, $\psi$ is obviously satisfied. 

\subsection{Statistics of consistent/mutated instances}
\applabel{pos-statistics}
For each consistent instance, 
we report the instance name (instance), consistency (cc), event number ($\NumEvents$), thread number ($\NumThreads$), channel number ($\NumChannels$), maximal capacity ($\MaxCapacity$) as well as the running time of each algorithm. 
For algorithm that times out or run out of memory, the running time is set to be 10800s (the timeout value). 
The details of statistics can be found via an artifact uploaded to  \href{https://zenodo.org/records/17322633}{Zenodo}. 


\end{document}